\newcommand{\todo}[1]{}

\documentclass[a4paper,10pt]{article}

\usepackage{a4wide}
\usepackage{oubraces}

\newcommand{\condscale}{1}
\newcommand{\condwidth}{0.97\textwidth}

\usepackage{afterpage}

\usepackage{microtype}
\title{Metric Reasoning About $\lambda$-Terms: The General Case\\ (Long Version)
}
\author{Rapha\"elle Crubill\'e \and Ugo Dal Lago}

\newcommand{\qed}{}

\usepackage{mathrsfs}
\usepackage{amscd}
\usepackage{pdfsync} 
\usepackage{amssymb,amsmath}
\usepackage{proof}
\usepackage{url}
\usepackage{bussproofs}
\usepackage{stmaryrd}
\usepackage{upgreek}
\usepackage{tikz-cd}

\newtheorem{theorem}{Theorem}
\newtheorem{proposition}{Proposition}
\newtheorem{lemma}{Lemma}
\newtheorem{definition}{Definition}
\newtheorem{example}{Example}
\newtheorem{corollary}{Corollary}

\newenvironment{proof}{\begin{trivlist}
       \item[\hskip \labelsep {\bfseries Proof.}]}{\hfill $\Box$ \end{trivlist}}

\newtheorem{fact}{Fact}

\newcommand{\hide}[1]{}

\newcommand{\bnf}{::=}

\usepackage{mathtools}

\newcommand{\midd}{\; \; \mbox{\Large{$\mid$}}\;\;}

\newcommand{\dirac}[1]{{\{{#1}^1\}} }
\newcommand{\diracpar}[2]{{\{{#1}^{#2}\}} }
\newcommand{\fillc}[2]{#1[#2]}
\newcommand{\subst}[3]{#1\{#2/#3\}}

\newcommand{\distrpar}[3]{\{{#1}^{#2}\}_{#3}}

\newenvironment{varitemize}
{
\begin{list}{{\labelitemi}}
{\setlength{\itemsep}{0pt}
 \setlength{\topsep}{0pt}
 \setlength{\parsep}{0pt}
 \setlength{\partopsep}{0pt}
 \setlength{\leftmargin}{15pt}
 \setlength{\rightmargin}{0pt}
 \setlength{\itemindent}{0pt}
 \setlength{\labelsep}{5pt}
 \setlength{\labelwidth}{10pt}
}}
{
 \end{list} 
}

\newcounter{numberone}

\newcounter{numbertwo}

\newcommand{\BB}{\mathbb{B}}
\newcommand{\NN}{\mathbb{N}}

\newcommand{\DD}{\mathbb{D}}
\newcommand{\RR}{\mathbb{R}}

\newcommand{\QQ}{\mathbb{Q}}

\newcommand{\ctsinfd}{\mathbf{ C \times \Delta(\tuples)}}

\newcommand{\stateonetd}{h}
\newcommand{\statetwotd}{k}
\newcommand{\statethreetd}{l}
\newcommand{\forget}[1]{\textbf{F}(#1)}

\newcommand{\distrs}[1]{\text{Distr}(#1)}
\newcommand{\distrone}{\mathcal{D}}
\newcommand{\distrtwo}{\mathcal {E}}
\newcommand{\distrthree}{\mathcal {F}}

\newcommand{\sumdistr}[1]{{\sum\nolimits_{#1}}}

\newcommand{\lfin}{\stackrel{f}{\leq}}

\newcommand{\types}{\mathscr{A}}

\newcommand{\tms}{\mathscr{T}}

\newcommand{\pair}[2]{\langle#1,#2\rangle}
\newcommand{\varone}{x}
\newcommand{\vartwo}{y}
\newcommand{\varthree}{z}

\newcommand{\vars}{\mathscr{X}}

\newcommand{\supp}{\textsf{S}}
\newcommand{\sem}[1]{\llbracket #1 \rrbracket}

\newcommand{\semct}[1]{{\llbracket #1 \rrbracket }^{\ctsinfd}}
\newcommand{\sssp}[2]{#1 \Rightarrow #2}
\newcommand{\ssspn}[3]{#1 \stackrel{#3}{\Rightarrow} #2}
\newcommand{\derivone}{{\Updelta}}
\newcommand{\derivtwo}{\Upxi}
\newcommand{\wtrcct}[2]{{#1} {\Rightarrow}_{\ctsinfd} #2}
\newcommand{\wtrcctnp}{{\Rightarrow}_{\ctsinfd}}

\newcommand{\osct}[2]{{#1}\rightarrow_{\ctsinfd} {#2} }
\newcommand{\osctnp}{\rightarrow_{\ctsinfd}}
\newcommand{\osctl}[3]{{#1}{\stackrel{#2}\rightarrow}_{\ctsinfd} {#3} }
\newcommand{\osctplus}[2]{{#1}\rightarrow_{\ctsinfd}^{+} {#2} }
\newcommand{\osctrefl}[2]{{#1}\rightarrow_{\ctsinfd}^{\star} {#2} }
\newcommand{\unary}{\text{Unary}^{\ctsinfd}}
\newcommand{\unaryd}[1]{\overline{#1}^{\text{U}}}

\newcommand{\relone}{\mathrel{R}}

\newcommand{\contsone}{\mathscr{C}}
\newcommand{\contsonecbv}{\mathscr{C}_{\text{cbv}}}
\newcommand{\contsonecbn}{\mathscr{C}_{\text{cbn}}}
\newcommand{\contextst}{\mathscr{C}^{\mathbf{T}}}
\newcommand{\contone}{C}
\newcommand{\conttwo}{D}
\newcommand{\contthree}{F}
\newcommand{\contfour}{G}

\newcommand{\stateone}{t}
\newcommand{\statetwo}{s}
\newcommand{\actone}{a}

\newcommand{\traces}{\mathscr{T}}

\newcommand{\traceone}{\mathsf{T} }
\newcommand{\tracetwo}{\mathsf{S}}

\newcommand{\hole}{[\cdot]}

\newcommand{\valset}{\mathscr{V}}

\newcommand{\bang}[1]{!{#1}}
\newcommand{\abstrexp}[2]{\lambda {\bang #1}. #2}

\newcommand{\abstr}[2]{\lambda #1. #2}
\newcommand{\psum}[2]{#1 \,\oplus\, #2}
\newcommand{\paror}[3]{\, \left([#1 \parallel  #2 ] \rightarrowtail #3\right)\,}
\newcommand{\fv}[1]{FV(#1)}
\newcommand{\termone}{M}
\newcommand{\termtwo}{N}
\newcommand{\termthree}{L}
\newcommand{\termfour}{P}

\newcommand{\valone}{V}
\newcommand{\valtwo}{W}

\newcommand{\contevone}{E}
\newcommand{\contevtwo}{F}

\newcommand{\onestepr}[2]{#1 \rightarrow #2}
\newcommand{\redonestepr}[2]{#1 \hookrightarrow #2}
\newcommand{\redonestep}{\hookrightarrow}

\newcommand{\LOPLUS}{\Lambda_\oplus}

\newcommand{\LBANG}{\Lambda_\oplus^{!}}
\newcommand{\LBANGT}{\Lambda_\oplus^{!,\downarrow}}
\newcommand{\LBANGOR}{\Lambda_\oplus^{!,\parallel}}

\newcommand{\typone}{\sigma}
\newcommand{\typtwo}{\tau}
\newcommand{\typthree}{\gamma}
\newcommand{\typfour}{\eta}
\newcommand{\typfive}{\iota}

\newcommand{\tarr}[2]{#1 \multimap #2}

\newcommand{\rectype}[2]{\mu {#1}. #2}
\newcommand{\vartypone}{\alpha}

\newcommand{\equaltypes}{=^{\types}}
\newcommand{\substy}[3]{#1[#2 \rightarrow #3]}
\newcommand{\wfjt}[3]{#1 \vdash #2 : #3}

\newcommand{\contbangone}{! \Gamma}
\newcommand{\contlinone}{\Delta}
\newcommand{\contlintwo}{\Theta}
\newcommand{\contlinthree}{ \Xi}

\newcommand{\judgone}{\kappa}

\newcommand{\vvalidjudgss}{\mathscr{J}^{\valset}}
\newcommand{\vjcbn}{\mathscr{J}^{\valset}_{\text{cbn}}}
\newcommand{\vjcbv}{\mathscr{J}^{\valset}_{\text{cbv}}}
\newcommand{\validjudgst}{\mathscr{J}}
\newcommand{\judgsubs}{\mathscr{J}^{\text{subst}}}

\newcommand{\appl}[3]{{#1}(#2,#3)}
\newcommand{\obs}[1]{\text{Obs}(#1)}

\newcommand{\tuples}{ \mathscr{U} }
\newcommand{\tuplone}{K}
\newcommand{\tupltwo}{H}
\newcommand{\tuplonea}[1]{\left[#1\right]}
\newcommand{\typtuplone}{A}
\newcommand{\typtupltwo}{B}
\newcommand{\typtuplthree}{C}
\newcommand{\typtuples}{\mathbf T}

\newcommand{\remove}[1]{\text{remove}(#1)}
\newcommand{\add}[1]{\text{add}(#1)}

\newcommand{\nforms}[1]{NF(#1)}

\newcommand{\mchain}{\mathscr M}
\newcommand{\mcstates}{\mathcal S}
\newcommand{\mclabels}{\mathcal A}
\newcommand{\mcprob}{\mathcal P}
\newcommand{\flmone}{\mathscr{F}}
\newcommand{\fstone}{\hat{\mathcal S}}
\newcommand{\flabels}{\hat{\mathcal A}}

\newcommand{\lmcbang}{\mathscr{M}_\oplus^{\bang{}}}
\newcommand{\statesexp}{\mathcal{S}_{\lmcbang}}
\newcommand{\labelsexp}{\mathcal{A}_{\lmcbang}}
\newcommand{\probexp}{\mathcal{P}_{\lmcbang}}

\newcommand{\ltsone}{\mathscr{L}}

\newcommand{\ltsact}[3]{{#1}{\stackrel{#2}{\rightarrow}}{#3} }
\newcommand{\ltsactf}[4]{{#1}{\stackrel{#2}{\rightarrow}}_{#4}{#3} }
\newcommand{\ltsbang}{\mathscr{L}_\oplus^{\bang{}}}
\newcommand{\ltsf}[1]{\mathscr{L}_{#1}}
\newcommand{\ltsstates}{\mathcal{S}_{\ltsbang}}
\newcommand{\ltslabels}{\mathcal{A}_{\ltsbang}}
\newcommand{\ltsfst}[1]{\mathcal{S}_{#1}}
\newcommand{\ltsflabels}[1]{\mathcal{A}_{ #1}}
\newcommand{\ltstrans}{\ltsact{}{\cdot}{}}
\newcommand{\ltsftrans}[1]{\ltsact{}{\cdot}{_{#1}}}
\newcommand{\stateterm}[2]{\hat{s}_{#1}(#2)}
\newcommand{\statetermcbn}[1]{\hat{s}^{\text{cbn}}(#1)}
\newcommand{\tupletermcbn}[1]{\hat{U}^{\text{cbn}}(#1)}
\newcommand{\weight}[1]{w(#1)}

\newcommand{\lmcbangcbn}{\mathscr{M}_\oplus^{\text{cbn}}}

\newcommand{\statesexpcbn}{\mathcal S_{\lmcbangcbn}}
\newcommand{\labelsexpcbn}{\mathcal A_{\lmcbangcbn}}

\newcommand{\typzero}{\typtuplone^{0}}
\newcommand{\typun}{\typtuplone^{1}}

\newcommand{\flmccbn}{\mathscr{F}^{\text{cbn}}}

\newcommand{\ltscbv}{\mathscr{L}_\oplus^{\text{cbv}}}
\newcommand{\ltsstatescbv}{\mathcal{S}_{\ltscbv}}
\newcommand{\ltslabelscbv}{\mathcal{A}_{\ltscbv}}
\newcommand{\lmcbangcbv}{\mathscr{M}_\oplus^{\text{cbv}}}
\newcommand{\lmccbv}{\mathscr{M}_\oplus^{\text{cbv}}}
\newcommand{\stcbv}{\mathscr{S}_\oplus^{\text{cbv}}}
\newcommand{\actcbv}{\mathscr{A}_\oplus^{\text{cbv}}}
\newcommand{\probcbv}{\mathscr{P}_\oplus^{\text{cbv}}}

\newcommand{\flmccbv}{\mathscr{F}^{\text{cbv}}}

\newcommand{\statesexpcbv}{\mathcal S_{\lmcbangcbv}}
\newcommand{\labelscbv}{\mathcal A_{\lmcbangcbv}}

\newcommand{\labelsappl}{\mathcal A_{@}}
\newcommand{\labelsbangb}{{\mathcal A_{?}}}
\newcommand{\labelsbang}{\mathcal A_{{!}}}

\newcommand{\interpone}{\varphi}
\newcommand{\interptwo}{\psi}

\newcommand{\actappltupl}[2]{ @^{#2}_{#1}}
\newcommand{\actbang}[1]{(!^{#1})}
\newcommand{\evalbang}[1]{({?}^{#1})}

\newcommand{\relate}[3]{#2\,{#1}\,#3}
\newcommand{\lifting}[1]{\widehat{#1}}

\newcommand{\statethree}{u}
\newcommand{\statefour}{v}

\newcommand{\metrtr}[1]{\delta^{b}_{#1}}
\newcommand{\mctxbangor}[1]{\delta_{#1,\bang{}, \parallel}^{c}}
\newcommand{\mcbangor}{\delta_{\bang{}, \parallel}^{c}}
\newcommand{\mctxbang}[1]{\delta^{c}_{#1,!}}

\newcommand{\mctbang}{\delta_{\bang{},\downarrow}^{c}}
\newcommand{\mctxcbn}{\delta_{\text{cbn}}^{c}}
\newcommand{\mctxcbv}{\delta_{\text{cbv}}^{c}}
\newcommand{\mtbang}{\delta_{\bang{}}^t}

\newcommand{\mtbangt}[1]{\delta_{#1,\bang{}}^b}
\newcommand{\mtcbn}{\delta_{\text{cbn}}^b}
\newcommand{\mtcbv}{\delta_{\text{cbv}}^b}

\newcommand{\metrctx}[1]{\delta^{\text{c}}_{#1}}
\newcommand{\embcbn}[1]{\langle #1 \rangle^{\text{cbn}}}
\newcommand{\embcbv}[1]{\langle #1 \rangle^{\text{cbv}}}
\newcommand{\typecbn}{\typone^{\text{cbn}}}
\newcommand{\typecbv}{\typone^{\text{cbv}}}
\newcommand{\typtwocbv}{\typtwo^{\text{cbv}}}
\newcommand{\tracescbn}{\mathscr T _r^{\text{cbn}}}

\newcommand{\emptytr}{\top}
\newcommand{\concat}[2]{{#1 \cdot #2}}
\newcommand{\probt}[2]{\mathsf{PR}_{#1}(#2)}
\newcommand{\probtrcbn}[2]{Pr_{\text{cbn}}(#1,#2)}
\newcommand{\metrcbn}{\delta_{\text{cbn}}}
\newcommand{\metrcbv}{\delta_{\text{cbv}}}

\newcommand{\ttrue}{\textbf{true}}
\newcommand{\ffalse}{\textbf{false}}
\newcommand{\psumindex}[3]{#1 \oplus^{#2} #3}

\newcommand{\orn}[2]{{\bigvee}^{#2}(#1)}
\newcommand{\andn}[2]{{\bigwedge}^{#2}(#1)}

\newcommand{\upto}[2]{#1 \triangleright #2}
\newcommand{\uptor}[2]{#1 \rightarrow^{\triangleright}_{\ctsinfd} #2}
\newcommand{\valct}[1]{\mathscr V(#1)}
\newcommand{\seqone}{\mathit s}
\newcommand{\semctu}[2]{\semct{#1}_{#2}}

\newcommand{\elone}{e} 
\newcommand{\seq}[1]{[#1]}
\newcommand{\isone}{E} 
\newcommand{\istwo}{F} 
\newcommand{\ms}[1]{\boldsymbol{#1}} 
\newcommand{\emptyseq}{\ms {\bar 0}}

\usepackage{tikz}

\begin{document}

\renewcommand{\textfraction}{0.2}
\setcounter{topnumber}{3}
\renewcommand{\topfraction}{0.9}
\renewcommand{\bottomfraction}{0.8}
\renewcommand{\dbltopfraction}{0.9}

\date{}

\maketitle

\begin{abstract}
  In any setting in which observable properties have a quantitative
  flavour, it is natural to compare computational objects by way of
  \emph{metrics} rather than \emph{equivalences} or \emph{partial
    orders}. This holds, in particular, for probabilistic higher-order
  programs. A natural notion of comparison, then, becomes context
  \emph{distance}, the metric analogue of Morris' context
  \emph{equivalence}. In this paper, we analyze the main properties of
  the context distance in fully-fledged probabilistic
  $\lambda$-calculi, this way going beyond the state of the art, in
  which only affine calculi were considered. We first of all study to
  which extent the context distance trivializes, giving a sufficient
  condition for trivialization. We then characterize context distance
  by way of a coinductively defined, \emph{tuple-based} notion of
  distance in one of those calculi, called $\Lambda^\oplus_!$. We
  finally derive pseudometrics for call-by-name and call-by-value
  probabilistic $\lambda$-calculi, and prove them fully-abstract.
\end{abstract}
\section{Introduction}
Probability theory offers computer science models which enable system
abstraction (at the price of introducing uncertainty), but which can
also be seen as a \emph{a way to compute}, like in randomized
computation. Domains in which probabilistic models play a key role
include machine learning~\cite{pearl1988probabilistic},
robotics~\cite{thrun2002robotic}, and
linguistics~\cite{manning1999foundations}. In cryptography, on the
other hand, having access to a source of uniform randomness is
essential to achieve security, e.g., in the public key
setting~\cite{GoldwasserMicali}. This has stimulated the development
of concrete and abstract programming languages, which most often are
extensions of their deterministic siblings. Among the many ways
probabilistic choice can be captured in programming, the simplest one
consists in endowing the language of programs with an operator
modelling the flipping of a fair coin. This renders program evaluation
a probabilistic process, and under mild assumptions the language
becomes universal for probabilistic computation. Particularly fruitful
in this sense has been the line of work on the functional paradigm,
both at a theoretical~\cite{JonesPlotkin,Ramsey,Pfenning} and at a
more practical level~\cite{Church}.

We are still far, however, from a satisfactory understanding of
higher-order probabilistic computation. As an example, little is known
about how much of the classic, beautiful, theory underlying the
$\lambda$-calculus~\cite{Barendregt84} can be lifted to probabilistic
$\lambda$-calculi, although the latter have been known from forty
years now \cite{SahebDjahromi}. Until the beginning of this decade,
indeed, most investigations were directed towards domain theory, which
has been proved to be much more involved in presence of probabilistic
choice than in a deterministic scenario~\cite{JungTix1998}. In the
last ten years, however, some promising results have appeared. As an
example, both quantitative semantics and applicative bisimilarity have
been shown to coincide with context equivalence for certain kinds of
probabilistic
$\lambda$-calculi~\cite{EhrhardTassonPagani2014POPL,CrubilleDalLago2014ESOP}. This
not only provides us with new proof methodologies for program
equivalence, but also sheds new light on the very nature of probabilistic
higher-order computation.  As an example, recent results tell us that
program equivalence in presence of probabilistic choice lies somehow
in between determinism and non-determinism
\cite{CrubilleDalLago2014ESOP}.

But are equivalences the most proper way to compare terms? Actually,
this really depends on what the underlying observable is. If
observables are boolean, then equivalences (and preorders) are indeed
natural choices: two programs are dubbed equivalent if they give rise
to the same observable (of which there are just two!) in any
context. If, on the other hand, the observable is an element of a
metric space, which happens for example when we observe (the
probability of) convergence in a probabilistic setting, one may wonder
whether replacing equivalences with metrics makes sense.  This is a
question that has recently been given a positive answer in the
\emph{affine} setting~\cite{CrubilleDalLagoLICS2015}, i.e., in a
$\lambda$-calculus in which copying is simply not available. More
specifically, a notion of context distance has been shown to model
differences between terms satisfactorily, and has also been shown to
be characterized by notions of trace metrics, and to be approximated
from below by behavioural metrics.

Affine $\lambda$-calculi are very poor in terms of the computations
they are able to model. Measuring the distance between terms in
presence of copying, however, is bound to be problematic.  On the one
hand, allowing contexts to copy their argument has the potential risk
of \emph{trivialising} the underlying metric.  On the other hand, finding
handier characterizations of the obtained notion of metric in the
style of behavioural or trace metrics is inherently hard. A
more thorough discussion on these issues can be found in
Section~\ref{sect:mti} below.

In this paper, we attack the problem of analyzing the distance between
$\lambda$-terms in its full generality. More specifically, the
contributions of this paper are fourfold:
\begin{varitemize}
\item
  First of all, we define a linear probabilistic $\lambda$-calculus,
  called $\LBANGOR$, in which copying and a nonstandard construct,
  namely Plotkin's parallel disjunction, are both available. A very
  liberal type system prevents deadlocks, but nevertheless leaves the
  expressive power of the calculus very high. This choice has been
  motivated by our will to put ourselves in the most general setting,
  so as to be able to talk about different fragments.  The calculus is
  endowed with a notion of context distance, in Morris' style.  This
  is covered in Section~\ref{sect:syntax} below.
\item
  We study trivialization of the obtained notion(s) of metric for
  different fragments of $\LBANGOR$, showing that both parallel
  disjunction and strong normalization give us precisely the kind of
  discriminating power we need to arbitrarily amplify distances, while
  in the most natural fragment, namely $\LBANG$, trivialization does
  \emph{not} hold. This is the subject of Section~\ref{sect:trivialisation}.
\item
  In Section \ref{sect:metrtupl}, we prove that context distance can
  be characterized as a coinductively defined distance on a
  labelled Markov chain of \emph{tuples}.  The way (tuples of) terms
  interact with their environment makes proofs of soundness laborious
  and different from their affine counterparts from
  \cite{CrubilleDalLagoLICS2015}. An up-to-context notion of
  bisimulation is proved to be sound, and to be quite useful when
  evaluating the distance between concrete programs.
\item
  Finally, we show that the results from Section \ref{sect:metrtupl} can be lifted
  back to ordinary probabilistic $\lambda$-calculi from the
  literature~\cite{DalLagoZorzi,CrubilleDalLago2014ESOP}. Both when
  call-by-name evaluation and call-by-value are considered, our
  framework can be naturally adapted, and helps in facilitating
  concrete proofs. This is in Section~\ref{sect:nonlinear}.
\end{varitemize}

\section{Metrics and Trivialisation, Informally}\label{sect:mti}
The easiest way to render the pure $\lambda$-calculus a universal
probabilistic computation model~\cite{DalLagoZorzi} consists in
endowing it with a binary construct $\oplus$ for probabilistic
choice. The term $\termone\oplus\termtwo$ evolves as either $\termone$
or $\termtwo$, each with probability $\frac{1}{2}$. The obtained
calculus can be given meaning by an operational semantics which puts
terms in correspondence with \emph{distributions of} values. The
natural notion of observation, at least in an untyped setting like the
one we will consider in this paper, is thus the \emph{probability of
  convergence} of the observed term $\termone$, which will be denoted
as $\sumdistr{\sem{\termone}}$. One could then define a notion of
\emph{context equivalence} following Morris' pattern, and stipulate
that two terms $\termone$ and $\termtwo$ should be equivalent whenever
they terminate with \emph{exactly} the same probability when put in
\emph{any} context:
$$
\termone\equiv\termtwo\;\Leftrightarrow\;\forall\contone.\;\sumdistr{\sem{\contone[\termone]}}=\sumdistr{\sem{\contone[\termtwo]}}.
$$
The anatomy of the obtained notion of equivalence has been recently
studied extensively, the by-products of this study being powerful
techniques for it in the style of bisimilarity and logical relations
\cite{DalLagoSangiorgiAlberti2014,CrubilleDalLago2014ESOP,BizjakBirkedal}.

As observed by various authors (see, e.g., \cite{MardareDoctoral} for
a nice account), probabilistic programs and processes are
naturally compared by \emph{metrics} rather than \emph{equivalences}:
the latter do not give any quantitative information about \emph{how
  different} two non-equivalent programs are. Given that the
underlying notion of observation is inherently quantitative, on the
other hand, generalizing context equivalence to a
\emph{pseudometric} turns out to be relatively simple:
$$
\delta(\termone,\termtwo)=\sup_{\contone}\left|\sumdistr{\sem{\contone[\termone]}}-\sumdistr{\sem{\contone[\termtwo]}}\right|.
$$ 
Observe that the obtained notion of context \emph{distance} between two terms
is a real number between $0$ and $1$, which is minimal precisely when
the considered terms are context equivalent. It is the least
discriminating pseudometric which is non-expansive and adequate, and
as such it provides some quite precise information about how far the
two argument programs are, observationally. A similar notion has
recently been studied by the authors~\cite{CrubilleDalLagoLICS2015},
but only in a purely affine setting.

Let us now consider two prototypical examples of non-equivalent terms,
namely $I=\lambda\varone.\varone$ (the identity) and $\Omega$ (the
always-divergent term). The context distance $\delta^c(I,\Omega)$
between them is maximal: when applied, e.g., to the trivial context
$[\cdot]$, they converge with probability $1$ and $0$, respectively. A
term which is conceptually ``in the middle'' of them is
$\termone=I\oplus\Omega$. Indeed, in a purely affine
$\lambda$-calculus,
$\delta^c(I,\termone)=\delta^c(\termone,\Omega)=\frac{1}{2}$.

If we render the three terms duplicable (by putting them in the scope
of a $!$-operator), however, the situation becomes much more
complicated. Consider the terms $!I$ and $!(I\oplus\Omega)$. One can
easily define a family of contexts $\{C_n\}_{n\in\NN}$ such that the probability
of convergence of $C_n[!I]$ and $C_n[!(I\oplus\Omega)]$ tend to $1$
and $0$ (respectively) when $n$ tends to infinity. It suffices to take
$C_n$ as
$
(\lambda!\varone.\underbrace{\varone\ldots\varone}_{\mbox{$n$ times}})\hole.
$
Allowing contexts to have the capability to duplicate their
argument seems to mean that they can arbitrarily \emph{amplify} distances.
Indeed, the argument above also works when $(I\oplus\Omega)$ is
replaced by any term which behaves as $\Omega$ with probability $\varepsilon$
and as $I$ with probability $1-\varepsilon$, provided of course $\varepsilon >0$.
But how about $!\Omega$ and $!(I\oplus\Omega)$? Are they at maximal
distance, i.e.  is it that $\delta^c(!\Omega,!(I\oplus\Omega))=1$?
Apparently, this is \emph{not} the case. The previously defined contexts $C_n$ cannot
amplify the ``linear'' distance between the two terms above, namely
$\frac{1}{2}$, up to $1$. But what is the distance between 
$!\Omega$ and $!(I\oplus\Omega)$, then? Evaluating it is hard, since you need to consider
all contexts, which do not have a nice structure. In Section \ref{sect:metrtupl},
we will introduce a different, better behaved, notion of distance, this
way being able to prove that, indeed,
$\delta^c(!\Omega,!(I\oplus\Omega))=\frac{1}{2}$. 

All this hints at even more difficult examples, like the one in which
$\termone_\varepsilon = \bang {(\Omega \oplus^\varepsilon I)}$, where
$\oplus^\varepsilon$ is the natural generalization of $\oplus$ to a
possibly unfair coin flip, and one is interested in evaluating
$\delta^c(\termone_{\varepsilon},\termone_{\mu})$. In that case, we
can easily see that the "linear" distance between them is $\lvert
\varepsilon - \mu\rvert$. In some cases, it is possible to amplify it:
the most natural way is again to consider the contexts $C_n$ defined
above. Indeed, we see that the probabilities of convergence of
$\fillc{C_n}\termone$ and $\fillc{C_n}\termtwo$ are $\varepsilon^n$
and $\mu^n$, respectively. It follows that
$\delta^c(\termone_\varepsilon, \termone_\mu)\geq\sup_{n \in
  \NN}{\lvert \varepsilon^n - \mu^n \rvert}$. For some $\varepsilon$
and $\mu$ (for example if $\varepsilon + \mu >1$), the context
distance can be greater than $|\varepsilon-\mu|$. But there is no easy
way to know \emph{how far} amplification can lead us. The terms
$\termone_\varepsilon$ and $\termone_\mu$ will be running examples in
the course of this paper. Despite their simplicity, evaluating the
distance between them is quite challenging.

We are also going to consider the case where contexts can evaluate
terms \emph{in parallel}, converging if and only if at least one of
the copies converges. This behaviour is not expressible in the usual
$\lambda$-calculus, but is captured by well-known constructs and in
particular by Plotkin's parallel disjunction \cite{Plotkin77}. In
Section \ref{sect:trivialisation} below, we prove that all this is not
accidental: the presence of parallel disjunction turns a
non-trivialising metric into a trivialising one. The proof of it, by
the way, relies on building certain amplifying contexts which are then
shown to be universal using tools from functional analysis.
\section{A Linear Probabilistic $\lambda$-Calculus}\label{sect:syntax}
In this section, we present the syntax and operational semantics of
our language $\LBANGOR$, on which we will later define metrics.
$\LBANGOR$ is a probabilistic and linear $\lambda$-calculus, designed
not only to allow copying, but to have a better control on it. It is
based on a probabilistic variation of the calculus defined in
\cite{Simpson2005}, whose main feature is to never reduce inside
exponential boxes. As we will see in Section~\ref{sect:nonlinear}, the
calculus is capable of encoding both call-by-value and call-by-name
fully-fledged probabilistic $\lambda$-calculi. We add a parallel
disjunction construct to the calculus, being inspired by Plotkin's
parallel disjunction \cite{Plotkin77}. Noticeably, it has been
recently shown \cite{VignudelliCrubilleDalLagoSangiorgi} that adding
parallel disjunction to a (non-linear) $\lambda$-calculus increases
the expressive power of contexts to the point of enabling coincidence
between the contextual preorder and applicative similarity. The choice
of studying a very general calculus is motivated by our desire to be
as general as possible. This being said, many of our results hold only
in \emph{absence} of parallel disjunction.
\begin{definition}
We assume a countable set of variables $\vars$.  The set of \emph{terms} of
$\LBANGOR$ (denoted $\tms$) is defined by the following grammar:
$$
\termone \in \tms \bnf\;\varone \midd \termone \termone  \midd
\abstr \varone \termone  \midd \abstrexp \varone \termone  \midd \bang
\termone \midd\psum \termone \termone \midd \paror \termone \termone \termone,
$$
where $\varone \in \vars$. The fragment of
$\LBANGOR$ without the $\paror{\cdot}{\cdot}{\cdot}$ construct will
be indicated as $\LBANG$. \emph{Values}
are those terms derived from the following grammar:
$$
\valone \in \valset \bnf \abstr \varone \termone \midd \abstrexp
\varone \termone \midd \bang \termone.
$$
\end{definition}
As already mentioned, $\psum \termone \termtwo$ can evolve to either
$\termone$ or $\termtwo$, each with probability $\frac{1}{2}$. The
term $\bang \termone$ is a duplicable version of $\termone$, often
called an \emph{(exponential) box}. We have two distinct abstraction
operators: $\abstr\varone \termone$ is a \emph{linear} abstraction,
while the \emph{non-linear} abstraction $\abstrexp \varone \termone$
requires exponential boxes as arguments. The term $\paror \termone
\termtwo \termthree$ behaves as $\termthree$ if either $\termone$ or
$\termtwo$ converges. Please observe that both abstractions and boxes
are values---our notion of reduction is \emph{weak} and \emph{surface}
\cite{Simpson2005}.

We are now going to define an operational semantics for the closed
terms of $\LBANGOR$ in a way similar to the one developed for a
(non-linear) $\lambda$-calculus in \cite{DalLagoZorzi}. We need to
first define a family of \emph{approximation semantics}, and then to
take the semantics of a term as the least upper bound of all its
approximations. The approximation semantics relation is denoted $\sssp
\termone \distrone$, where $\termone$ is a closed term of $\LBANGOR$,
and $\distrone$ is a \emph{(sub)distribution on values} with finite
support, i.e., a function from $\valset$ to $\RR_{[0,1]}$ which sums
to a real number $\sum_\distrone\leq 1$.  For any distribution
$\distrone$ on a set $X$, we call \emph{support of $\distrone$}, and we
note $\supp (\distrone)$, the set $ \{x \in X \mid \distrone(x)>0
\}$. We say that $\distrone$ is \emph{finite} if $\supp(\distrone)$ is a
finite set.

The rules deriving the approximation semantics relation are given in
Figure \ref{apprsem}, and are based on the notion of an
\emph{evaluation context}, which is an expression generated from the
following grammar:
$$
\contevone \bnf\;\hole \midd \contevone \valone \midd \termone
\contevone \midd \paror \termone \contevone \termtwo \midd 
   \paror\contevone \termone \termtwo.
$$
As usual, $\contevone[\termone]$ stands for the term obtained by
filling the sole occurrence of $\hole$ in $\contevone$ with
$\termone$.  In Figure \ref{apprsem} and elsewhere in this paper, we
indicate the distribution assigning probability $p_i$ to $\valone_i$
for every $i\in\{1,\ldots,n\}$ as
$\{\valone_1^{p_1},\ldots,\valone_n^{p_n}\}$.  Similarly for the
expression $\{\valone_i^{p_i}\}_{i\in I}$, where $I$ is any countable
index set.  Observe how we first define a one-step reduction relation
$\onestepr{\cdot}{\cdot}$ between closed terms and \emph{sequences} of
terms, only later extending it to a small-step reduction relation
$\sssp{\cdot}{\cdot}$ between closed terms and \emph{distributions} on
values.
\begin{figure*}[!h]
\begin{center}
\fbox{
  \begin{minipage}{\condwidth}
    \footnotesize
    \begin{center}
      $$
      \AxiomC{}
      \UnaryInfC{$\redonestepr {\psum\termone \termtwo }{\termone,\termtwo} $}
      \DisplayProof
      \qquad
      \AxiomC{}
      \UnaryInfC{$\redonestepr {(\abstr \varone \termone )\valone}{{\subst\termone \varone \valone}} $}
      \DisplayProof
      \qquad 
      \AxiomC{}
      \UnaryInfC{$\redonestepr {(\abstrexp \varone \termone){\bang\termtwo} }{{\subst\termone \varone \termtwo}} $}
      \DisplayProof
      $$
      \\ \vspace{-15pt}
      $$
      \AxiomC{\strut} 
      \UnaryInfC{$\redonestepr { \paror \valone \termone \termtwo} \termtwo $} \DisplayProof
      \quad 
      \AxiomC{\strut} 
      \UnaryInfC{$\redonestepr { \paror \termone \valone \termtwo} \termtwo $} \DisplayProof
      \quad
      \AxiomC{\strut $\redonestepr \termone {\termtwo_1, \ldots, \termtwo_n}$}
      \UnaryInfC{$\onestepr {\fillc \contevone \termone }{{\fillc \contevone {\termtwo_1}}, \ldots {\fillc \contevone {\termtwo_n}}} $}
      \DisplayProof 
      $$
      \\ \vspace{-10pt} 
      $$
      \AxiomC{\strut}
      \UnaryInfC{$\sssp{\valone}{\{\valone^1\}}$}
      \DisplayProof
      \qquad\quad
      \AxiomC{\strut}
      \UnaryInfC{$\sssp{\termone}{\emptyset}$}
      \DisplayProof
      \qquad\quad
      \AxiomC{\strut $\onestepr{\termone}{\termtwo_1, \ldots, \termtwo_n}$}
      \AxiomC{$(\sssp{\termtwo_i}{\distrone_i})_{1\leq i \leq n}$}
      \BinaryInfC{$\sssp{\termone}{\sum_{1\leq i \leq n} {\frac 1 n} \cdot \distrone_i }$}
      \DisplayProof
      $$
      \vspace{1pt}
    \end{center}
\end{minipage}}

\end{center}
\caption{Approximation Semantics for $\LBANG$}\label{apprsem}
\end{figure*}
A reduction step can be a linear or non-linear $\beta$-reduction, 
or a probabilistic choice.  Moreover, there can be more than
one active redex in any closed term $\termone$, due to the presence
of parallel disjunction.  For any term $\termone$, the set of
sub-distributions $\distrone$ such that $\sssp \termone \distrone$ is
a countable directed set. Since the set of sub-distributions (with
potentially infinite support) is an $\omega$-complete partial order, we
can define the \emph{semantics} of a term $\termone$ as $\sem \termone
= \sup \{\distrone \mid \sssp \termone \distrone \} $. We could also
define alternatively a big-step semantics, again in the same way as
that of the probabilistic $\lambda$-calculus considered in
\cite{DalLagoZorzi}.

Not all irreducible terms are values in $\LBANGOR$, e.g.  $(\abstrexp
\varone \varone) (\abstr \varone \varone)$.  We thus need a
\emph{type-system} which guarantees the absence of deadlocks. Since we
want to be as general as possible, we consider recursive types as
formulated in~\cite{positivetype}, which are expressive enough to type the image
of the embeddings we will study in Section~\ref{sect:nonlinear}. The grammar of
\emph{types} is the following:
$$
\typone \in \types \bnf  \vartypone \midd\rectype\vartypone
\tarr{\typone}{\typone}  \midd\rectype\vartypone\bang{\typone}\midd
\tarr{\typone}{\typone}  \midd\bang{\typone}
$$
Types are defined up to the equality $\equaltypes$,
defined in Figure \ref{typs}. $\substy{\typone}{\vartypone}{\typtwo}$
stands for the type obtained by substituting all free occurrences
of $\vartypone$ by $\typtwo$ in $\typone$.
\begin{figure*}[!h]
\begin{center}
\fbox{
\begin{minipage}{\condwidth}
\begin{center}\footnotesize
\vspace{-10pt}
$$
\AxiomC{$\strut $}  \UnaryInfC{$\rectype\vartypone {\tarr \typone \typtwo}
    \equaltypes {\tarr{\substy
        \typone \vartypone {(\rectype \vartypone {\tarr \typone \typtwo})}}{\substy
        \typtwo \vartypone {(\rectype \vartypone {\tarr \typone \typtwo})}}}$}
 \DisplayProof
$$
\\ \vspace{-10pt}
$$
\AxiomC{$\strut $}  \UnaryInfC{$\rectype\vartypone {\bang \typone}
    \equaltypes {\bang{(\substy
        \typone \vartypone {\rectype \vartypone {\bang \typone}})}}$}
 \DisplayProof 
\qquad\quad
 \AxiomC{ \strut $\typone  \equaltypes \substy \typthree \vartypone
    \typone$} \AxiomC{\strut $\typtwo  \equaltypes \substy \typthree
    \vartypone \typtwo$}  \BinaryInfC{$
    \typone \equaltypes \typtwo $} \DisplayProof
$$
\\ \vspace{-4pt}
\end{center}
\end{minipage}}
\end{center}
\caption{Equality of Types}\label{typs}
\end{figure*}
An \emph{environment} is a set of expressions in the form
$\varone:\typone$ or $\bang\varone:\bang\typone$ in which any variable
occurs at most once. Environments are often indicated with
metavariables like $\contbangone$, which stands for an environment in
which all variables occur as $\bang\varone$, or $\contlinone$ in
which, on the contrary, \emph{only} variables can occur,
i.e. $\contlinone$ is of the form
$\varone_1:\typone_1,\ldots,\varone_n:\typone_n$. \emph{Typing
  judgments} are thus of the form
$\wfjt{\contbangone,\contlinone}{\termone}{\typone}$.
The typing system is given in Figure \ref{judg}. The role of this type system is
\emph{not} to guarantee termination, but rather to guarantee a form of
type soundness:
\begin{lemma}
  If $\wfjt{}{\termone}{\typone}$ and $\sssp{\termone}{\distrone}$,
  then $\wfjt{}{\valone}{\typone}$ for every $\valone$ in the support
  of $\distrone$. Moreover, if $\wfjt{}{\termone}{\typone}$ and
  $\termone$ is irreducible (i.e. $\termone\not\redonestep\termtwo$
  for every $\termtwo$), then $\termone$ is value.
\end{lemma}
\begin{figure*}
\begin{center}
\fbox{\footnotesize
\begin{minipage}{\condwidth}
\begin{center}
$$
\AxiomC{} 
\UnaryInfC{$\wfjt{\contbangone,\bang\varone : \bang
    \typone}{\varone}{\typone}$}
\DisplayProof 
\qquad\quad 
\AxiomC{}
\UnaryInfC{$\wfjt{\contbangone,\varone : \typone}{\varone}{\typone}$}
\DisplayProof
\qquad\quad
\AxiomC{$\wfjt{\contbangone,\varone:\typone, \contlinone}{\termone}\typtwo$}
\UnaryInfC{$\wfjt{\contbangone, \contlinone}{\abstr{\varone}{\termone}}{\tarr \typone
    \typtwo}$} 
\DisplayProof 
$$
$$
\AxiomC{$\wfjt{\bang\varone:\typone, \contbangone, \contlinone}{\termone}\typtwo$}
\UnaryInfC{$\wfjt{\contbangone, \contlinone}{\abstrexp{\varone}{\termone}}{\tarr
    \typone \typtwo}$} 
\DisplayProof
\quad
\AxiomC{$\wfjt{\contbangone, \contlinone}{\termone}{\tarr \typone
    \typtwo}$}
\AxiomC{$\wfjt{\contbangone,
    \contlintwo}{\termtwo}\typone$} 
\BinaryInfC{$\wfjt{\contbangone,
    \contlinone,\contlintwo}{\termone\termtwo} \typtwo$} \DisplayProof
$$
\\ \vspace{-7pt}
$$
\AxiomC{$\wfjt{\contbangone}{\termone}\typone$}
\UnaryInfC{$\wfjt{\contbangone}{\bang \termone}{\bang \typone}$}
\DisplayProof
\quad
\AxiomC{$\wfjt{\contbangone, \contlinone}{\termone} \typone$}
\AxiomC{$\wfjt{\contbangone, \contlinone}{\termtwo} \typone$}
\BinaryInfC{$\wfjt{\contbangone, \contlinone}{\psum{\termone}{\termtwo}} \typone$}
\DisplayProof  
$$
\\ \vspace{-9pt}
$$
\AxiomC{$\wfjt{\contbangone, \contlinone}{\termone}{\typone}$}
\AxiomC{$\wfjt{\contbangone,\contlintwo}{\termtwo}\typone$} 
\AxiomC{$\wfjt{\contbangone,\contlinthree}{\termthree}\typtwo $}
\TrinaryInfC{$\wfjt{\contbangone,
    \contlinone,\contlintwo, \contlinthree}{\paror\termone\termtwo\termthree} \typtwo$} \DisplayProof
$$
\\ \vspace{-3pt}
\end{center}
\end{minipage}}
\end{center}
\caption{Typing Rules}\label{judg}
\end{figure*}

\begin{example}\label{example0}
  The term $I = \abstr \varone \varone$ can be typed as $\wfjt{}I
  {\tarr {\typone} \typone}$ for every $\typone \in \types$.  We
  define ${\Omega_!}$ to be the term $(\abstrexp \varone {\varone
    {\bang \varone}})(\bang{(\abstrexp \varone{\varone{\bang
        \varone}})})$, which is the counterpart in our linear calculus
  of the prototypical diverging term of the $\lambda$-calculus, namely
  $\Omega=(\lambda x.xx)(\lambda x.xx)$. We can type this divergent
  term with any possible type: indeed, if we take $\typtwo
  \bnf\rectype\vartypone {\tarr{\bang \vartypone}\typone}$, then
  $\typtwo \equaltypes \tarr{\bang \typtwo}{\typone} $ and
  $\wfjt{}{\abstrexp \varone{\varone {\bang \varone}}}{\typone}$.
  Using that, we can see that $\wfjt{}{{\Omega_!}}{\typone}$ for every
  type $\typone$.  We will see in Section \ref{sect:nonlinear} that,
  more generally, there are several ways to turn any pure
  $\lambda$-term $\termone$ into a $\LBANG$ term in such a way as to
  obtain meaningful typing and semantics: $\LBANG$ is actually at
  least as powerful as the usual untyped probabilistic
  $\lambda$-calculus \cite{DalLagoZorzi}.
\end{example}
Termination could in principle be guaranteed if one considers
\emph{strictly positive} types, as we will do in Section
\ref{sect:strictlypositive} below.
Let $\DD$ be the set of dyadic numbers (i.e. those rational numbers in
the form $\frac{n}{2^m}$ (with $n,m\in\NN$ and $n\leq 2^m$).  It is
easy to derive, for every $\varepsilon\in\DD$, a new binary operator
on terms $\psumindex{\cdot}{\varepsilon}{\cdot}$ such that
$\sem{\psumindex{\termone}{\varepsilon}{\termtwo}}=(1-\varepsilon)\sem{\termone}+
\varepsilon\sem{\termtwo}$ for every closed $\termone,\termtwo$. 
It can
be defined, e.g., as follows by induction on $m$:
\begin{align*}
  \psumindex{\termone}{0}{\termtwo}&=\termone\\
  \psumindex{\termone}{1}{\termtwo}&=\termtwo\\
  \psumindex{\termone}{\frac{n}{2^m}}{\termtwo}&=
  \left\{
  \begin{array}{ll}
    \psum{\termone}{(\psumindex{\termone}{\frac{n}{2^{m-1}}}{\termtwo})}&\mbox{if $n\leq 2^{m-1}$}\\
    \psum{\termtwo}{(\psumindex{\termone}{\frac{n-2^{m-1}}{2^{m-1}}}{\termtwo})}&\mbox{if $n>2^{m-1}$}
  \end{array}
  \right.
\end{align*}
\begin{example}\label{example1}
  We define here a family of terms that we use as a running example.  We
  consider terms of the form $\termone_\varepsilon = \bang{(\psumindex
    {\Omega_!} \varepsilon I)}$, for $\varepsilon \in \DD$. It holds that
  $\wfjt {}{\termone_\varepsilon}{\bang{(\tarr \typone \typone)}}$ for
  every $\typone$. $\termone_\varepsilon$ corresponds to a duplicable
  term each copy of which behaves as $I$ with probability
  $\varepsilon$, and does not terminate with probability $1-\varepsilon$.
\end{example}

\subsection{Some Useful Terminology and Notation}

In this paper, we will make heavy use of sequences of terms and
types. It is thus convenient to introduce some terminology and
notation about them. 

A finite (ordered) sequence whose elements are
$\elone_1,\ldots,\elone_n$ will be indicated as
$\ms{\elone}=\seq{\elone_1,\ldots,\elone_n}$, and called an
\emph{$n$-sequence}. Metavariables for sequences are boldface
variations of the metavariables for their elements. Whenever
$\isone=\{i_1,\ldots,i_m\}\subseteq\{1,\ldots,n\}$ and
$i_1<\ldots<i_m$, the sub-sequence
$\seq{\elone_{i_1},\ldots,\elone_{i_m}}$ of an $n$-sequence
$\ms{\elone}$ will be indicated as $\ms{\elone}_\isone$.  If the above
holds, $\isone$ will be called an \emph{$n$-set}. If $\ms \elone$ is an
$n$-sequence, and $\interpone$ is a permutation on $\{1, \ldots,n\}$, we
note $\elone_\interpone$ the $n$-sequence
$\seq{\elone_{\interpone(1)}, \ldots, \elone_{\interpone{(n)}}}$.  We
can turn an $n$-sequence into a $(n+1)$-sequence by adding an
element at the end: this is the role of the semicolor operator.  We denote by
$\seq{\elone^n}$ the $n$-sequence where all components are equal to
$\elone$.

Whenever this does not cause ambiguity, notations like the ones above
will be used in conjunction with syntactic constructions. For example,
if $\ms{\typone}$ is an $n$-sequence of types, then
$\bang{\ms{\typone}}$ stands for the sequence
$[\bang{\typone_1},\ldots,\bang{\typone_n}]$.  As
another example, if $\ms{\typone}$ is an $n$-sequence of types and
$\isone$ is an $n$-set, then $\ms{\varone}_\isone:\ms{\typone}_\isone$
stands for the environment assigning type $\typone_i$ to $\varone_i$
for every $i\in\isone$. As a final example, if $\ms{\termone}$ is an
$n$-sequence of terms and $\ms{\typone}$ is an $n$-sequence of types,
$\wfjt{}{\ms{\termone}}{\ms{\typone}}$ holds iff
$\wfjt{}{\termone_i}{\typone_i}$ is provable for every
$i\in\{1,\ldots,n\}$.
\subsection{Context Distance}\label{sect:contextdistance}
\newcommand{\ctxs}[2]{\mathcal{C}^{#1}_{#2}}
A \emph{context of type $\typone$ for terms of type $\typtwo$} is a
term $\contone$ which can be typed as
$\wfjt{\mathit{hole}\,:\,\typtwo}{\contone}{\typone}$,
where $\mathit{hole}$ is a distinguished
variable. $\ctxs{\typtwo}{\typone}$ collects all such terms. If
$\contone\in\ctxs{\typtwo}{\typone}$ and $\termone$ is a closed term
of type $\typtwo$, then the closed term $\subst \contone {\mathit{hole}}
\termone$ has type $\typone$ and is often indicated as
$\fillc\contone\termone$.

The \emph{context distance}~\cite{CrubilleDalLagoLICS2015} is the
natural quantitative refinement of context equivalence. Intuitively,
it corresponds to the maximum separation that contexts can induce
between two terms. Following \cite{CrubilleDalLagoLICS2015}, we take as
observable the probability of convergence: for any term
$\termone$, we define its \emph{observable} $\obs \termone$ as
$\sumdistr {\sem \termone}$. Then, for any terms $\termone$,
$\termtwo$ such that $\wfjt{}{\termone}{\typone}$ and
$\wfjt{}{\termtwo}{\typone}$, we define:
$$
\appl {\mctxbangor \typone} \termone \termtwo = \sup_{\contone\in\ctxs{\typtwo}{\typone}} {\lvert{\obs
    {\fillc\contone\termone} - \obs{\fillc\contone
      \termtwo}}\rvert}.
$$
Please observe that this distance is a pseudometric, and that
moreover we can recover context equivalence by considering its
\emph{kernel}, that is the set of pairs of terms which are at distance
$0$. The binary operator $\mctxbang\typone$ is defined similarly,
but referring to terms (and contexts) from $\LBANG$.
\begin{example}
  What can we say about
  $\appl{\mctxbangor\typone}{\termone_{\varepsilon}}{\termone_{\mu}}$?
  Not much apparently, since \emph{all} contexts should be considered. Even
  if we put ourselves in the fragment $\LBANG$, the best we can
  do is to conclude that
  $\appl{\mctxbang\typone}{\termone}{\termtwo}\geq\sup_{n\in\NN}\lvert\varepsilon^n-\mu^n\rvert$,
  as explained in Section~\ref{sect:mti}.
\end{example}

\section{On Trivialisation}\label{sect:trivialisation}
As we have already mentioned, there can well be classes
of terms such that the context distance collapses to context equivalence,  due to the
copying abilities of the language. The question of trivialisation can
in fact be seen as a question about the expressive power of contexts:
given two duplicable terms, how much can a context amplify the observable
differences between their behaviours? 

More precisely, we would like to identify \emph{trivialising}
fragments of $\LBANGOR$, that is to say fragments such that for any
pair of duplicable terms, their context distance (with respect to the
fragment) is either 0 or 1. This is not the case in $\LBANG$ (see
Example~\ref{exa:nontrivial} below).

In fact, a sufficient condition to trivialization  is to require the
existence of \emph{amplification contexts}: for every observable type $ \sigma$, for every
$\alpha, \beta \in [0,1]$ distinct, for every $\gamma >0$, we want
to have a context $\contone_\typone^{\alpha,\beta, \gamma}$ 
such that: 
$$ 
\left. \begin{array}{r} 
\wfjt{} {\termone, \termtwo}{\typone}\\
\obs \termone = \alpha\\
\obs \termtwo = \beta 
\end{array}
\right\} \Rightarrow \left\lvert \obs{\fillc {\contone_\typone^{\alpha, \beta, \gamma}} {\bang\termone}} -\obs{\fillc {\contone_\typone^{\alpha, \beta, \gamma}} {\bang\termtwo}} \right\rvert \geq 1-\gamma.
$$
\begin{fact}
Any fragment of $\LBANGOR$ admitting all amplifications contexts
trivializes.
\end{fact}
\subsection{Strictly Positive Types}\label{sect:strictlypositive}
First, let us consider the case of the fragment $\LBANGT$ of $\LBANG$
obtained by considering strictly positive types, only (in a similar
way to \cite{positivetype}), and by dropping parallel
disjunction. Every term $\termone$ of $\LBANGT$ is terminating
(i.e. $\sum\sem{\termone}=1$), so we need to adapt our notion of
observation: we define the type $\BB = \tarr{\bang
  \vartypone}{\tarr{\bang \vartypone}{\vartypone}}$, which can be seen
as boolean type using a variant of the usual boolean encoding in
$\lambda$-calculi. Our new notion of observation, defined only at type
$\BB$, is $\obs \termone = \sumdistr {\sem {\termone \bang{I}\bang{
      {\Omega_!}}}}$, which corresponds to the probability that
$\termone$ evaluates to $\ttrue$. While this notion of observation
uses the full power of $\LBANG$,
the context distance $\mctbang$ based on it only consider contexts in
$\LBANGT$.
\begin{theorem}\label{theo:sptriv}
$\mctbang$ trivialises.
\end{theorem}
The proof of Theorem~\ref{theo:sptriv} is based on the construction
of amplification contexts. We are going to use Bernstein
constructive proof of the Stone-Weierstrass theorem. Indeed,
Bernstein showed that for every continuous function $f :
[0,1]\rightarrow \RR$, the following sequence of polynomials converges
uniformly towards $f$:

{\footnotesize
$$
P_n^f(x) = \sum\nolimits_{0 \leq k \leq n} f\left(\frac k n\right) \cdot B_k^n(x),
\text{ where }B_k^n(x) = \binom n k \cdot x^k\cdot (1-x)^{n-k}.
$$}

\noindent Let us consider the following continuous function: we fix
$f(\alpha) = 0$, $f(\beta) = 1$, and $f$ defined elsewhere in such a
way that it is continuous, that it has values in $[0,1]$, and that moreover
$f(\QQ)\subseteq\QQ$. We can easily implement $P_n^f$ by a context,
i.e. define $\contone$ such that for every $\termone$,  $\obs{
  \fillc \contone \termone} = P_n^f(\obs \termone) $. In $\LBANGT$, we can indeed copy an argument $n$
times, then evaluate it, and then for every $k$ between $0$ and $n$,
if the number of $\ttrue$s obtained is exactly $k$, return the term
$\psumindex \ffalse {f(\frac k n)} \ttrue$ (that corresponds to a term
returning $\ttrue$ with probability ${f \left(\frac k n\right) }$). 
Please observe that this construction works only because in $\LBANGT$ all
terms converge with probability $1$.
\subsection{Parallel Disjunction}
As we have seen, trivialization can be enforced by restricting the
class of terms, but we can also take the opposite road, namely
increasing the discriminating power of contexts. Indeed, consider the
full language $\LBANGOR$, with the usual notion of observation.

We can first see how parallel disjunction increases the expressive
power of the calculus on a simple example. Consider the following two
terms: $ \termone = \bang {\Omega_!}$ and $\termtwo = \bang {(\psum
  {\Omega_!} I)}$. We will see later that these two terms are the
simplest example of non-trivialization in $\LBANG$: indeed $\appl
{\mctxbang{\bang{(\tarr \typtwo \typtwo)}}}\termone \termtwo = \frac 1
2$, while $ \appl {\mctxbangor{\bang{(\tarr \typtwo
      \typtwo)}}}\termone \termtwo = 1$. In $\LBANGOR$, we are able to
define a family of contexts $(\contone_n)_{n\in\NN}$ as follows:
$$
\contone_n =\left(\abstrexp \varone \paror{\varone} {\paror {\varone}{\ldots}{I}} I \right)\hole.
$$
Essentially, $\contone_n$ makes $n$ copies of its
argument, and then converges towards $I$ if \emph{at least} one of
these copies itself converges. When we apply the context $\contone_n$ to
$\termone$ and $\termtwo$, we can see that the convergence probability
of $\fillc {\contone_n} {\termone}$ is always $0$ independently of
$n$, whereas the convergence probability of $\fillc {\contone_n}
\termtwo$ tends towards $1$ when $n$ tends to infinity.
\begin{theorem}\label{theo:parortrivial}
$\mcbangor$ trivializes.
\end{theorem}
The proof is based on the construction of amplification contexts
$\contone_\typone^{\alpha, \beta, \gamma}$. If $\max(\alpha, \beta) =
1$, we can extend the informal argument from Section
\ref{sect:mti}, by taking contexts that copy an arbitrary number of
times their argument. If $\min (\alpha, \beta) = 0$, we can use the
same idea as in the example above, by taking contexts that do an
arbitrary number of disjunctions.
What remains to be done to obtain the trivialization result is
treating the case in which $0< \alpha, \beta < 1$.  The overall idea
is to somehow mix the contexts we use in the previous
cases. More precisely, we define a family of contexts
$(\contone_n^m)_{n, m \in \NN}$ as follows:
$$\contone_n^m = \abstrexp \vartwo {\left( {\andn{\orn {\vartwo,
        \ldots , \vartwo} m, \ldots,{\orn {\vartwo, \ldots , \vartwo}
        m}}n} \right)}{\hole}$$
where
\begin{align*}
\orn  {\termone_1, \ldots \termone_n} n &= \paror{\termone_1}{\paror {\termone_2} {\ldots} I} I ;\\
\andn {\termone_1\ldots \termone_n} n &= {(\abstr{\varthree_1} \abstr
  {\varthree_2} \ldots \abstr \vartwo (\vartwo \varthree_1 \ldots
  \varthree_n))}{\termone_1 \ldots \termone_n}.
\end{align*}
The term $\orn {\termone_1, \ldots, \termone_n} n$ behaves
as a $n$-ary disjunction: it terminates if \emph{at least one} of the
$\termone_i$ terminates. On the other hand, $\andn{\termone_1, \ldots,
  \termone_n} n$ can be seen as a $n$-ary conjunction: it terminates
if \emph{all} the $\termone_i$ terminates. The contexts $\contone_n^\alpha$
compute $m$-ary conjunctions of $n$-ary disjunction. Now,
let $\iota$ be such that $\alpha < \iota < \beta$. We need
to show that for every $n$, we can choose $m(n,
\iota) \in \NN$ such that:
$$
\lim_{n \rightarrow \infty} \obs{\fillc {\contone_n^{m(n, \iota)}}{\bang\termone}} = 
\begin{cases}1 \text{ if } \obs \termone > \iota; \\ 0 \text{ if } \obs \termone < \iota.\end{cases}.
$$
We can express this problem purely in terms of functional analysis, by
observing that $ \obs{\fillc{\contone_n^m}{\bang \termone}} =
(1-(1-\obs\termone)^m)^n$. We are now going to show the result by applying the
dominated convergence theorem to a well-chosen sequence of functions.

We choose $m(n,
  \gamma) = \lceil r(n, \iota) \rceil$, with
  $$
  r(n,\iota)  = \left(\frac 1 {1-\iota}\right)^n .
  $$
  Now, we see that $\obs{\fillc{\contone_n^{m(n,\iota)}}{\bang
      \termone}}$, is equal to $f_{n, \iota}(\obs \termone)$, where
  $f_{n, \iota}$ is the real function defined by:
  $$
  f_{n,\iota}(x)  = (1-(1-x)^n)^{\lceil r(n, \iota) \rceil} 
  $$
  The result we are trying to show can now be expressed by way of the following lemma:
  \begin{lemma}\label{limit}
    $$
    \lim_{n \rightarrow \infty} f_{n, \iota}(x) = \begin{cases}
      0 \text{ if } x < \iota \\
      1 \text{ if } x > \iota 
    \end{cases} 
    $$
  \end{lemma}
  \begin{proof}
    In order to simplify the proof, we show an equivalent result on
    $g_{n,\alpha}(x) = f_{n, 1-\alpha}(1-x) $. Lemma \ref{limit} is equivalent to say that:
    $
    \lim_{n \rightarrow \infty} g_{n, \alpha}(x) = \begin{cases}
      1 \text{ if } x < \alpha, \\
      0 \text{ if } x > \alpha. 
    \end{cases} $ 
    We express $g$ as an integral, since we want to use known results about inversion of integral and limits:
    $$g_{n, \alpha}(x) = \begin{cases}
      \int_0^x {g'_{n, \alpha}(t)dt} + 1  \qquad \text{ if } x < \alpha\\
      - \int_x^1 g'_{n,\alpha}(t) dt  \text{ if } x> \alpha
    \end{cases}$$
    We can now express our goal using $g'$: we have to show that
    $$
    \lim_{n \rightarrow \infty} \int_0^x {g'_{n, \alpha}(t)dt} = 0 \qquad \text{and}\qquad
    \lim_{n \rightarrow \infty}  \int_x^1 g'_{n,\alpha}(t) dt = 0.
    $$
    We first establish a bound on the $g'_{n, \alpha}$:
    $$g'_{n, \alpha}(x) = -\lceil{r(n, 1-\alpha)}\rceil \cdot n \cdot \left(1-x^n \right)^{\lceil r(n, 1-\alpha) \rceil -1} \cdot x^{n-1} $$
    And consequently:
    \begin{equation}{\label{eq1}}
      \mid g'_{n,\alpha}(x)\mid \leq \frac 1 {\alpha^2} \cdot \left(\frac x \alpha \right)^{n-1}\cdot n \cdot (1-x^n)^{(\frac 1 \alpha)^n-1} 
    \end{equation}
    We have to use the following well-known inequality, that we recall here: 
    \begin{equation}\label{expl}
      \lim_{x \rightarrow \infty} x^m \cdot e^{-x} = 0  \text{ for every }m \in \NN .
    \end{equation}
    For every $x<\alpha$, we're going to use the bounded convergence theorem, which is a well-known theorem in real analysis, on $[0, x]$ and on $[x, 1]$ for every $x>\alpha$ . We recall this theorem here:
    \begin{theorem}[Bounded Convergence Theorem]
      If ${k_n}$ is a sequence of uniformly bounded real-valued measurable functions which converges pointwise on a bounded measure space $(S, \Sigma, \mu)$ to a function $k$, then the limit $k$ is an integrable function and:
      $\int_S {\lim_{n \rightarrow \infty} k_n} = \lim_{n \rightarrow \infty} \int_S k_n $. 
    \end{theorem}
    \begin{varitemize}
    \item Let's first take $0 \leq x <\alpha$. Here, we take $(S, \Sigma,\mu)$ as the segment $[0, x]$ with the Lebesgue measure (corresponding to the usual notion of integration on real functions). We obtain using (\ref{eq1}), that for every $t \in [0,x]$, it holds:
      $$\mid g'_{n, \alpha}(t)\mid \leq K \cdot n \cdot \left(\frac t \alpha\right)^{n-1}$$
      where $K$ is a constant. 
      We can already see that $g'_{n,\alpha}$ converge pointwise to the $0$-function on $[0,x]$ (using (\ref{expl})). We're now going to show it's bounded. Using the fact that $t \leq x$, we obtain: 
      \begin{align*}
        \mid g'_{n, \alpha}(t)\mid &\leq K \cdot n \cdot \left(\frac x \alpha\right)^{n-1}\\
        & \leq K \cdot \max_{n\in \NN}{n \cdot \left(\frac x \alpha\right)^{n-1} }
      \end{align*}
      Please observe that the $\max$ is well-defined and $< \infty$ because the sequence $n \cdot \left(\frac x \alpha\right)^{n-1}  $ has a finite limit (by (\ref{expl})). So we have achieved to bound $g'_{n, \alpha}(t)$ by a constant depending neither on $t$ nor on $n$. Consequently, we can apply the bounded convergence theorem, and we obtain : 
      $$\lim_{n \rightarrow \infty} \int_0^x {g'_{n, \alpha}(t)dt} = 0$$
    \item Now we consider the case where $\alpha < x \leq 1$.
      Here, we take $(S, \Sigma,\mu)$ as the segment $[0, x]$ with the Lebesgue measure.
      By using \ref{eq1}, we can see that for every $t \in [x, 1]$:
      \begin{align*}
        \mid g'_{n,\alpha}(t)\mid &\leq {\frac 1 {\alpha^2}} \cdot \left(\frac t \alpha \right)^{n-1}\cdot n \cdot exp \left(\left(\left(\frac 1 \alpha\right)^n -1\right) \cdot (- t^n)\right)\\ 
        & \leq {\frac 1 {\alpha^2}}\cdot \left(\frac t \alpha \right)^{n-1}\cdot n \cdot exp \left(-\left(\frac t \alpha\right)^n \right)\cdot exp(t^n)\\
        & \leq \frac 1 {\alpha^2} \cdot \left(\frac 1 \alpha \right)^{n-1}\cdot n \cdot exp \left(-\left(\frac x \alpha\right)^n \right)\cdot exp(1)\\
      \end{align*}
      Using \ref{expl} (and doing several computation), we can see that for every $a,b> 1$, it holds that:
      \begin{equation}
        \lim_{m \rightarrow \infty}{a^{m}\cdot m \cdot exp(-b^m)} = 0 
      \end{equation}
      By applying this, we obtain:
      \begin{equation}\label{eq2}
        \lim_{n \rightarrow \infty}  \frac 1 {\alpha^2} \cdot \left(\frac 1 \alpha \right)^{n-1}\cdot n \cdot exp \left(-\left(\frac x \alpha\right)^n \right)\cdot exp(1) = 0
      \end{equation}
      Now, please observe that \ref{eq2} gives us both the pointwise limit and the uniform bound (since $g_{n, \alpha}(t)$ is bounded by a sequence that doesn't depends on t, and has a finite limit when n goes towards infinity).
      So we can apply the bounded convergence theorem, and we obtain that:
      $$\lim_{n \rightarrow \infty} \int_x^1 {g'_{n, \alpha}(t)dt} = 0$$
    \end{varitemize}
    This concludes the proof
  \end{proof}

\section{Tuples and Full Abstraction}\label{sect:metrtupl}
This section is structured as follows: first, we define a labelled
Markov chain (LMC) which expresses the semantics of our calculus in an
interactive way, and then we use it to give a coinductively defined
notion of distance on an LTS of \emph{distributions}, which coincides with the context
distance as defined in Section \ref{sect:contextdistance}.  We are not
considering parallel disjunction here: the motivations for that should
be clear from Theorem \ref{theo:parortrivial}.
\subsection{A Labelled Markov Chain over Tuples}
Labelled Markov chains are the probabilistic analogues to labelled
transition systems. Formally, a LMC is a triple $\mchain =
(\mcstates,\mclabels, \mcprob)$, where $\mcstates$ is a countable set
of states, $\mclabels$ is a countable set of \emph{labels}, and
$\mcprob : \mcstates \times \mclabels \rightarrow {\distrs \mcstates}$
is a \emph{transition probability matrix} (where $\distrs X$ is the
set of all distributions over $X$).

Following \cite{DalLagoSangiorgiAlberti2014POPL}, the interactive
behaviour of probabilistic $\lambda$-terms can be represented by an
LMC whose states are the terms of the language, whose actions are
values, and where doing the action $\valone$ starting from a state
$\termone$ corresponds to applying the value $\valone$ to
$\termone$. This approach is bound \emph{not} to work in presence of
pairs when metrics take the place of equivalences, due to the
unsoundness of projective actions. In \cite{CrubilleDalLagoLICS2015},
this observation lead us to introduce a new LMC whose states are
\emph{tuples} of terms, and whose actions include one \emph{splitting}
a pair: $\mcprob(\tuplonea {\pair \termone \termtwo})(
{\text{destruct}}) = \dirac {\tuplonea{\termone, \termtwo}}$. This
turns out to work well in an affine
setting~\cite{CrubilleDalLagoLICS2015}.  We are going to define a LMC
$\lmcbang = (\statesexp, \labelsexp, \probexp)$ which is an extension
of the one from \cite{CrubilleDalLagoLICS2015}, and which is adapted
to a language with copying capabilities. The idea is to treat
exponentials in the spirit of Milner's Law: $\tarr{\bang A} {A \otimes
  \bang A}$.
\subsubsection{States}
\emph{Tuples} are pairs of the form $\tuplone = ( \ms {\termone},{\ms
  \valone})$ where ${\ms \termone}$ and $\ms \valone$ are a sequence
of terms and values, respectively. The set of all such tuples is
indicated as $\tuples$. The first component of a tuple is called its
\emph{exponential part}, while the second one is called its
\emph{linear part}. We write $\wfjt{}{(\ms \termone, \ms
  \valone)}{(\ms \typone, \ms \typtwo)}$ if $\wfjt {}{\ms
  \termone}{\ms \typone}$ and $\wfjt{}{\ms \valone}{\ms \typtwo}$. We
note $\typtuples$ the set of pairs $\typtuplone = (\ms \typone,\ms
\typtwo)$, and we call \emph{tuple types} the elements of
$\typtuples$. Moreover, we say that $(\ms \typone,\ms \typtwo)$ is a
\emph{$(n,m)$ tuple type} if $\ms \typone$ and $\ms \typtwo$ are,
respectively, an $n$-sequence and an $m$-sequence. To any term
$\termone$, we associate a tuple in a natural way: we note $\dot M $
the tuple $(\tuplonea{},\tuplonea{\termone})$, and similarly if
$\typone$ is a type, we indicate the tuple type $(\tuplonea{},
\tuplonea{\typone})$ as $\dot \typone$. Please observe that if
$\wfjt{} \termone \typone$, then it holds that $\wfjt{}{\dot
  \termone}{\dot \typone}$.

A sequence of the form $(\isone, \istwo, \ms{\typone}, \ms{\typtwo},
\termone, \typthree )$ is said to be an \emph{applicative typing
  judgment} when $\ms{\typone}$ and $\ms{\typtwo}$ are, respectively, an
$n$-sequence and an $m$-sequence of types, $\isone$ and $\istwo$ are
respectively an $n$-set and an $m$-set, and moreover it holds that
$\wfjt{\bang {\ms\varone_\isone} : \bang{\ms \typone_\isone}, \ms
  \vartwo_\istwo: \ms \typtwo_\istwo}{\termone}{\typthree}$.
Intuitively, this means that if we have a tuple $\tuplone = (\ms
\termtwo, \ms \valone)$ of type $(\ms \typone, \ms \typtwo)$, we can
replace free variables of $\termone$ by \emph{some} of the terms from
$\tuplone$. More precisely, we can replace variables in linear
position by the $\valone_i$ with $i \in \istwo$, and variables in non
linear position by $\termtwo_j$, with $j \in \isone$. We note as
$\fillc \termone \tuplone$ the closed term of type $\typthree$ that we
obtain this way. We note $\validjudgst$ the set of all applicative
typing judgments. We are specially interested in those judgments
$(\isone, \istwo, \ms{\typone}, \ms{\typtwo}, \termone, \typthree )$
in $\validjudgst$ such that for every tuple $\tuplone$, the resulting
term $\fillc \termone \tuplone$ is a \emph{value}: that is when either
$\termone = \vartwo_i$ for a $i \in \NN$, or $\termone$ is of the form
$\abstr \varthree \termtwo$, $\abstrexp \varthree \termtwo$, or $\bang
\termtwo$. We note $\vvalidjudgss$ the set of those judgments.

We are now in a position to define $\lmcbang$ formally. The
set of its states is indeed defined as
$
\statesexp = \{(\tuplone,
\typtuplone) \mid \tuplone \in \tuples,\, \typtuplone \in \typtuples,
\, \wfjt{\,}{\tuplone}{\typtuplone} \}.
$

\subsubsection{Labels and Transitions}
How do states in $\statesexp$ interact with the environment?  This is
captured by the labels in $\labelsexp$, and the associated probability
matrix. We define $\labelsexp$ as the disjoint union of $\labelsbangb$, $\labelsbang$ and $\labelsappl$, where:
$$
 \labelsbang =\labelsbangb = \{i \mid i \in \NN \}; \qquad\qquad
 \labelsappl =  \{ (\judgone,  i) \mid i \in \NN, \, \judgone \in \vvalidjudgss \}.$$
In order to distinguish actions in $\labelsbang$ and $\labelsbangb$,
we write the action $i\in\NN$ as $\evalbang i$ if it comes from
$\labelsbangb$, and as $\actbang i$ if it comes from $\labelsbang$.
The action $(\judgone,i)\in\labelsappl$ is often indicated as
$\actappltupl \judgone i$.
The probability matrix $\probexp$ is defined formally in
Figure \ref{fig:lmc}. We give below some intuitions about it. The
general idea is that $\lmcbang$ is designed to express every possible
effect that a context can have on tuples.  $\labelsbangb$ and
$\labelsbang$ are designed to model copying capabilities, while $\labelsappl$
corresponds to applicative interactions.

Actions in $\labelsbangb$ take any term of the form $\bang \termone$
from the linear part of the underlying tuple, unbox it and transfer
$\termone$ to the exponential part of the tuple. Please observe that
this action is in fact deterministic: the resulting tuple is uniquely
determined. Labels in $\labelsbang$, on the other hand, model the act
of \emph{copying} terms in the exponential part. We call its elements
\emph{Milner's actions}. More specifically, the action $\actbang i$
takes the $i$-th term in the exponential part of the tuple, makes a
copy of it, evaluates it and adds the result to the linear
part. Please observe that, contrary to $\evalbang i$, this action can
have a probabilistic outcome: the transferred term is evaluated.

Labels in $\labelsappl$ are analogues of the applicative actions from
applicative bisimulation over terms (see,
e.g. \cite{DalLagoSangiorgiAlberti2014POPL}). As such, they model
environments passing arguments to programs. Here, we have to adapt
this idea to our tuple framework: indeed, we can see the tuple as a
collection of programs available to the environment, who is free
to choose \emph{with which} of the programs to interact with
by passing it an argument. This argument, however, could depend on
other components of the tuple, which need to be removed
from it if lying in its linear part. Finally, please observe that all
this should respect types. Labels in $\labelsappl$ are indeed defined
as a pair of an index $i$ corresponding to the position in the tuple
of the term the environment chooses, and an applicative typing judgment, used to
specify the argument.  Please observe that in the definition of
probability matrix for applicative actions, in Figure \ref{fig:lmc},
the condition on $i$ implies that the $i$-th linear component of the
tuple is not used to construct the argument term.

\begin{figure}
\begin{center}
\fbox{
\begin{minipage}{0.95 \textwidth}
\begin{center}
\scalebox{\condscale}{
\begin{tikzpicture}[auto]
\node [draw, rectangle, rounded corners] (M) at (-0.4,-2.5) {\small $
\begin{array}{c}
( ({\ms \termone}, \ms {\valtwo}),  
({\ms \typone}, \ms \typtwo)) \\ \text{ with } {\typtwo_i = \tarr \typthree \typfour}
\end{array}
$};

\node [draw, rectangle, rounded corners] (P) at (7,-2.5) {\small $
\begin{array}{c}
(({\ms \termone}, \ms{\valtwo}_{\{1, \ldots, n\} \setminus {(\istwo\cup\{i \})}};\valone),  \\
 ( {\ms \typone}, \ms{\typtwo}_{\{1, \ldots, n\} \setminus {(\istwo\cup\{ i\})}};\typfour ))\\
\end{array}
$};
\draw[->](M) to node[above]{
$\actappltupl \judgone i$} node[below]{\small $
\begin{array}{c}
\sem {\valtwo_i {(\fillc \termtwo {(\ms \termone, \ms \valtwo)})}} (\valone)\\
\text{ with }\judgone = (\isone, \istwo, \ms{\typone}, \ms{\typtwo}, \termtwo, \typthree ) \text{ and } i \not \in \istwo
\end{array}$
} (P);

\node [draw, rectangle, rounded corners] (Q) at (-0.4,0) {\small $
\begin{array}{c}
((\ms \termone, \ms \valtwo),
(\ms \typone, \ms \typtwo ))\\
\text{ with }\valtwo_i = \bang \termtwo \text{ and } \typtwo_i = \bang \typthree
\end{array}
$};
\node [draw, rectangle, rounded corners] (R) at (6.6,0) {\small $
\begin{array}{c}
(\ms \termone;\termtwo, \ms \valtwo_{\{1, \ldots, n\} \setminus \{i\}}),  \\
(\ms \typone;\typthree,\ms \typtwo_{\{1, \ldots, n\} \setminus \{i\}}) )
\end{array}
$};
\draw[->](Q) to node[above]{\small
$\evalbang i$} node[below]{\small $1$} (R);
\node [draw, rectangle, rounded corners] (S) at (0,-1.25) {\small$
((\ms\termone, \ms \valtwo),  
(\ms \typone, \ms\typtwo))
$};

\node [draw, rectangle, rounded corners] (T) at (7,-1.25) {\small $
\begin{array}{c}
((\ms \termone, \ms \valtwo; \valone),  
(\ms \typone, \ms\typtwo; \typone_i))
\end{array}
$};
\draw[->](S) to node[above]{\small
$\actbang i$} node[below]{\small $\sem {\termone_i}(\valone)$} (T);
\end{tikzpicture}
}
\end{center}
\end{minipage}}
\end{center}
\caption{Definition of $\probexp$}\label{fig:lmc}
\end{figure}

\begin{example}\label{example2}
  We give in Figure~\ref{fig:fragmentLMC} a fragment of $\lmcbang$
  illustrating our definitions on an example. Let $\typtwo$ be
  an element of $\types$. We consider terms of the form
  $\termone_\varepsilon = \bang{(\psumindex {\Omega_!}  \varepsilon
    I)}$, for $\varepsilon\in\DD $ and we look at \emph{some} of the possible
  evolutions in $\lmcbang$ from the associated state $(\dot
  {\termone_\varepsilon}, \dot{\bang{({\tarr \typtwo \typtwo})}}) = (
  \tuplonea{}, \tuplonea{\termone_\varepsilon}),(\tuplonea{},
  \tuplonea{\bang{({\tarr \typtwo \typtwo})}} )$. In
  Figure~\ref{fig:fragmentLMC}, we denote by $\typone$ the type $\tarr
  \typtwo \typtwo$.
\end{example}
\begin{figure}
\begin{center}
    \fbox{
      \begin{minipage}{\condwidth}
        \begin{center}
          \scalebox{\condscale}{
        \begin{tikzpicture}[scale = 0.5]
          \node [draw, rectangle, rounded corners] (s1) at (1,0) {\scriptsize$
            (\dot {\termone_\varepsilon}, \dot{\bang{\typone}})
            $};
          \node [draw, rectangle, rounded corners] (s2) at (7,0) {\scriptsize $
            \begin{array}{c}
              ( \tuplonea{{\psumindex {\Omega_!} \varepsilon I}}, \tuplonea{}) , \\(\tuplonea{{\typone}}, \tuplonea{} )
\end{array}
            $};
          \node [draw, rectangle, rounded corners] (s3) at (14,0) {\scriptsize $
            \begin{array}{c}
              ( \tuplonea{{\psumindex {\Omega_!} \varepsilon I}}, \tuplonea{I}),\\(\tuplonea{{ \typone}}, \tuplonea{ \typone} )
            \end{array}
            $};
          \node [draw, rectangle, rounded corners, inner sep=1pt] (s4) at (24,0) {\scriptsize $
            \begin{array}{c}
              ( \tuplonea{{\psumindex {\Omega_!} \varepsilon I}}, \tuplonea{I,I}) ,\\(\tuplonea{{ \typone}}, \tuplonea{ \typone, \typone} )
            \end{array}
            $};
          \node [draw, rectangle, rounded corners, inner sep=1pt] (s5) at (24,2) {\scriptsize $
            \begin{array}{c}
              ( \tuplonea{{\psumindex {\Omega_!} \varepsilon I}}, \tuplonea{\valone}) ,(\tuplonea{{\typone}}, \tuplonea{ \typtwo} )\\
              \text{with }\wfjt{}{\valone}{\typtwo}
            \end{array}
            $};
          \node [draw, rectangle, rounded corners, inner sep=1pt] (s6) at (24,-2) {\scriptsize $
            \begin{array}{c}
              ( \tuplonea{{\psumindex {\Omega_!} \varepsilon I}}, \tuplonea{\subst\valtwo{\varone_1} {\psumindex {\Omega_!} \varepsilon I} }) ,\\(\tuplonea{{\typone}}, \tuplonea{ \typtwo} )\\
              \text{with }  \wfjt{\bang\varone_1:\bang {\typone}}{\valtwo}{\typtwo}
            \end{array}
            $};
          \draw[->](s1) to node[below]{$1$} node[above]{$\evalbang 1 $} (s2);
          \draw[->](s2) to node[below]{$\varepsilon $} node[above]{$\actbang 1 $} (s3);
          \draw[->](s3) to node[below,near end]{$\varepsilon$} node[above, near end]{$\actbang 1 $} (s4);
          \draw[->](s3) to node[below]{${1}$} node[above, near start, xshift=0.0cm, yshift=0.2cm]{$
            \actappltupl {( \emptyseq, \emptyseq,\tuplonea{{\typone}}, \tuplonea{ {\typone }} , \valone, \typtwo)
            } 1 
            $} (s5);
          \draw[->](s3) to node[above]{$1$} node[below, near start]{$\actappltupl {(\tuplonea{1}, \emptyseq, \tuplonea{{\typone}}, \tuplonea{ { \typone}}, \valtwo, \typtwo)} 1 $} (s6);
        \end{tikzpicture}
          }
        \end{center}
      \end{minipage}}
  \end{center}
 \caption{A Fragment of $\probexp$}\label{fig:fragmentLMC}
\end{figure}
\subsection{Distributions as States}
Now that we have a LMC $\probexp$ modelling interaction between
(tuple of) terms and their environment, we could define notions
of metrics following one of the abstract definitions from
the literature, e.g. by defining the \emph{trace distance}
or the \emph{behavioural distance} between terms. This is,
by the way, the approach followed in \cite{CrubilleDalLagoLICS2015}.
We prefer, however, to first turn $\probexp$ into a transition system
$\ltsbang$ whose states are \emph{distributions} of tuples.  This
supports both a simple, coinductive presentation of the trace
distance, but also up-to techniques, as we will see in Section
\ref{subsect:Upto} below. Both will be very convenient when
evaluating the distance between concrete terms, e.g. our running
example.

It turns out that the usual notion of an LTS is not sufficient
for our purposes, since it lacks a way to \emph{expose} the
observables of each state, i.e., its sum. We thus adopt
the following definition:
\begin{definition}
  A \emph{weighted labelled transition system} (WLTS for short) is a
  quadruple in the form $\ltsone = (\mcstates, \mclabels,
  \ltsact{}{\cdot}{}, w) $ where:
  \begin{varitemize}
  \item
    $\mcstates$ is a set of states and $\mclabels$ is a countable set of actions,
  \item
    $\ltsact{}{\cdot}{}$ is a transition function, such that, for every
    $\stateone \in \mcstates$ and $\actone \in \mclabels$, there exists
    \emph{at most} one $\statetwo \in \mcstates$, such that $\ltsact \stateone
    \actone \statetwo$,
  \item
    $w : \mcstates \rightarrow [0,1] $.
  \end{varitemize}
\end{definition}
Please observe how WLTSs are \emph{deterministic} transition
systems. We define the WLTS $\ltsbang$ by way of the equations in
Figure \ref{wltsexp}. 
\hide{
We are now ready to give the specific WLTS we need,
namely the quadruple $\ltsbang =
(\ltsstates, \ltslabels, \ltstrans,w)$, where:
\begin{varitemize}
\item
  The set of states is
  $\ltsstates=\bigcup_{\typtuplone \in \typtuples}\left( \distrs{\{\tuplone \mid \wfjt
    {\,}{\tuplone}{\typtuplone}\}}\times \{ \typtuplone\}\right)$;
\item
  The set of actions is $\ltslabels = \labelsexp
  \cup\{\typtuplone \mid \typtuplone \in \typtuples \}$;
\item
  The transition function 
  $\ltsact{}{\cdot}{}$ is defined as follows:
  \begin{align*}
    &\ltsact {(\distrone, \typtuplone)}{\typtuplone}{(\distrone,
      \typtuplone)} \text{ for }\typtuplone \in \typtuples;\\
    &\ltsact{(\distrone, \typtuplone)}{\actone}{\sum_{\tuplone \in \tuples}
      \distrone(\tuplone)\cdot \probexp((\tuplone, \typtuplone))(\actone)} \text{ for }\actone \in \labelsexp.
  \end{align*}
\item
  The weight $w$ is such that $w(\distrone, \typtuplone) = \sumdistr
  \distrone$.
\end{varitemize}
}
\begin{figure}
\begin{center}
    \fbox{
      \begin{minipage}{\condwidth}
        \begin{center}
\scalebox{\condscale}{         
$
  \ltsstates=\bigcup_{\typtuplone \in \typtuples}\left( \distrs{\{\tuplone \mid \wfjt{\,}{\tuplone}{\typtuplone}\}}\times \{ \typtuplone\}\right)$}
\scalebox{\condscale}{$
\ltslabels = \labelsexp
\cup\{\typtuplone \mid \typtuplone \in \typtuples \} 
\qquad w(\distrone, \typtuplone) = \sumdistr \distrone
 $
}
\ \\
\scalebox{\condscale}{
$
\ltsact {(\distrone, \typtuplone)}{\typtuplone}{(\distrone,
      \typtuplone)} \text{ for }\typtuplone \in \typtuples \qquad
    \ltsact{(\distrone, \typtuplone)}{\actone}{\sum_{\tuplone \in \tuples}
      \distrone(\tuplone)\cdot \probexp((\tuplone, \typtuplone))(\actone)} \text{ for }\actone \in \labelsexp.
$}
\end{center}

\end{minipage}
}
\end{center}
\caption{The WLTS $\ltsbang =(\ltsstates, \ltslabels, \ltstrans,w) $}\label{wltsexp}
\end{figure}

If $\stateone = (\distrone, \typtuplone)$ is in $\ltsstates$, we say
that $\stateone$ is a $\typtuplone$-state.  It is easy to check that
$\ltsbang$ is nothing more than the natural way to turn $\probexp$
into a deterministic transition system.  We illustrate this idea in
Figure~\ref{fig:fragmentWLTS}, by giving a fragment of $\ltsbang$
corresponding to (part of) the fragment of $\lmcbang$ given in Example
\ref{example2}.

\begin{figure}
  \begin{center}
    \fbox{
      \begin{minipage}{\condwidth}
        \begin{center}
      \begin{tikzpicture}[scale = 0.5]
        \node [draw, rectangle, rounded corners] (s1) at (0,0) {\scriptsize$
          (\dirac{\dot{\termone_\varepsilon}}, \dot{\bang{{\typone}}})
          $};
        \node [draw, rectangle, rounded corners] (s2) at (6,0) {\scriptsize $
\begin{array}{c}
  ( \dirac{(\tuplonea{{\psumindex {\Omega_!} \varepsilon I}}, \tuplonea{})},  \\(\tuplonea{{\typone}}, \tuplonea{} ))
\end{array}
$};
        \node [draw, rectangle, rounded corners] (s3) at (13,0) {\scriptsize $
          \begin{array}{c}
            (\diracpar{(\tuplonea{{\psumindex {\Omega_!} \varepsilon I}}, \tuplonea{I})}{\varepsilon},\\(\tuplonea{{\typone}}, \tuplonea{\typone}))
          \end{array}
          $};
        \node [draw, rectangle, rounded corners, inner sep=1pt] (s4) at (18,2.6) {\scriptsize $
\begin{array}{c}
  ( \diracpar{(\tuplonea{{\psumindex {\Omega_!} \varepsilon I}}, \tuplonea{I,I})}{\varepsilon^2},\\ (\tuplonea{{\typone}}, \tuplonea{\typone, \typone}))
\end{array}
$};
        \node [draw, rectangle, rounded corners, inner sep=1pt] (s5) at (18,-2.6) {\scriptsize $
          \begin{array}{c}
            (\diracpar{(\tuplonea{{\psumindex {\Omega_!} \varepsilon I}}, \tuplonea{\valone})}{\varepsilon},\\(\tuplonea{{\typone}}, \tuplonea{ \typtwo} ))\\
            \text{with }\wfjt{}{\valone}{\typtwo}
          \end{array}
          $};
        \draw[->](s1) to node [below]{$\evalbang 1 $} (s2);
        \draw[->](s2) to node[below]{$\actbang 1 $} (s3);
        \draw[->](s3) to node[above, near start, xshift=-0.1cm, yshift=-0.1cm]{$\actbang 1 $} (s4);
        \draw[->](s3) to node[above, right, xshift=0.0cm, yshift=0.1cm]{$
          \actappltupl {( \emptyseq, \emptyseq,\tuplonea{{\typone}},\tuplonea{ {\typone}} , \valone, \typtwo)
          } 1 
          $} (s5);
      \end{tikzpicture}
      \end{center}
    \end{minipage}}
  \end{center}
  \caption{A Fragment of $\ltsbang$}\label{fig:fragmentWLTS}
  \end{figure}
\subsection{A Coinductively Defined Metric}\label{subsect:tr}
Following Desharnais et al. \cite{DLT2008}, we use a quantitative
notion of bisimulation on $\ltsbang$ to define a distance between
terms. The idea is to stipulate that, for any $\varepsilon \in [0,1]$,
a relation $\relone$ is an $\varepsilon$-bisimulation if it is,
somehow, \emph{$\varepsilon$-close to a bisimulation}. The distance
between two states $\stateone$ and $\statetwo$ is just the smallest
$\varepsilon$ such that $\stateone$ and $\statetwo$ are
$\varepsilon$-bisimilar.  However, while in \cite{DLT2008} the notion
of $\varepsilon$-bisimulation is \emph{local}, we want it to be more
restricted by the \emph{global} deviation we may accept considering
arbitrary sequences of actions.

\begin{definition}
\label{dista}
Let $\ltsone = (\mcstates, \mclabels, \ltsact{}{\cdot}{},w)$ be a WLTS.
Let $\relone$ be a symmetric and reflexive relation on $\mcstates $,
and $\varepsilon \in [0,1]$. $\relone$ is a $\varepsilon$-bisimulation
whenever the following two conditions hold:
\begin{varitemize}
\item if $\relate {\relone} \stateone \statetwo$, and
  $\ltsact \stateone \actone \statethree$, then there exists $\statefour$ such that
  $\ltsact \statetwo \actone \statefour$, and
  it holds that $\relate {\relone} \statethree
  \statefour$. 
\item if $\relate {\relone} {\stateone} {\statetwo}$, then
  $\lvert\weight \stateone - \weight \statetwo\rvert \leq \varepsilon$.
\end{varitemize}
For every $\varepsilon \in [0,1]$, there exists a largest
$\varepsilon$-bisimulation, that we indicate as
$\relone^\varepsilon$. Please observe that it is not an equivalence
relation (since it is not transitive). We can now define a metric on
$\mcstates$ : $\appl {\metrtr \ltsone} \stateone \statetwo =
\inf{\left\{\varepsilon \mid \relate
  {\relone^\varepsilon}{\stateone}{\statetwo}\right\}}$.

\end{definition}
The
  greatest lower bound is in fact reached as a
  $\appl {\metrtr \ltsone} \stateone \statetwo$-bisimulation: that's the sens of the following lemma.
\begin{lemma}\label{largestbisim}
 Let $\stateone, \statetwo \in \ltsstates$. Then
 $\appl{\metrtr \ltsone}\stateone \statetwo \leq \varepsilon$ if and only if
 $\relate{\relone^\varepsilon}\stateone \statetwo$.
\end{lemma}
\begin{proof}
  The right-to-left implication is given by the definition of
  $\metrtr \ltsone$. The other one comes from the fact that the relation
  $\relone$ given by $\relate \relone \stateone \statetwo$ if $\appl
  {\metrtr \ltsone} \stateone \statetwo \leq \varepsilon$ is a
  $\varepsilon$-bisimulation.
\end{proof}
As a corollary, it is easy to see that the infinum in Definition
\ref{dista} is in fact reached.

How can we turn $\metrtr\ltsone$ into a metric on \emph{terms}?  The
idea is to consider the distributions on \emph{tuples} one naturally
gets when evaluating the term.  To every term $\termone$ of type
$\typone$, we define $\stateterm \typone \termone \in \ltsstates$
as 
$ (\distrpar{\dot \valone}{\sem \termone(\valone)}{\valone \in \valset}, \dot \typone).
$
\begin{definition}
  For every terms $\termone$ and $\termtwo$ such that $\wfjt
  {}\termone \typone$ and $\wfjt {}{\termtwo}{\typone}$, we set $\appl
  {\mtbangt \typone} \termone \termtwo = \appl {\metrtr \ltsbang}
  {\stateterm \typone \termone}{\stateterm \typone \termtwo} $.
\end{definition}

\begin{example} \label{extrace1}
Consider again the terms $\termone_\varepsilon$ from Example
\ref{example1}. We fix a type $\typtwo$, and define $\typone = \tarr
\typtwo \typtwo$. As mentioned in Example \ref{example1}, it holds
that $\wfjt{}{\termone_\varepsilon}{\bang \typone}$.  Let now
$\varepsilon, \mu, \alpha$ be in $[0,1]$, and let $\relone$ be any
$\alpha$-bisimulation, such that $\relate \relone{\stateterm{\bang
    \typone}{\termone_{\varepsilon}}}{\stateterm{\bang
    \typone}{\termone_{\mu}}} $. Let $\{\stateone_i\}_{i\in\NN}$,
and $\{\statetwo_i\}_{i\in\NN}$ be families from $\ltsstates$
such that
$ \ltsact {\stateterm{\bang
    \typone}{\termone_{\varepsilon}}} {\evalbang 1} {\ltsact
  {\stateone_0} {\actbang 1} {\ltsact \ldots {\actbang
      1}{\stateone_i \ldots}}}$ and $ \ltsact {\stateterm{\bang
    \typone}{\termone_{\mu}}} {\evalbang 1} {\ltsact
  {\statetwo_0} {\actbang 1} {\ltsact \ldots {\actbang
      1}{\statetwo_i \ldots}}} $.  Since $\relone$ is an $\alpha$-bisimulation,
for every $i$, it holds that $\relate \relone
{\stateone_i}{\statetwo_i}$. Looking at the definition of $\ltsbang$,
it is easy to realise that:
\begin{align*}
\stateone_i &=  \distrpar{(\tuplonea{\psumindex
    {\Omega_!} {\varepsilon} I}, \tuplonea{I,\ldots, I})
  ,(\tuplonea{\typone}, \tuplonea{\typone, \ldots,\typone})}{\varepsilon^i}{i\in\NN};
\\ 
\statetwo_i &=  \distrpar{(\tuplonea{\psumindex {\Omega_!}
    {\mu} I}, \tuplonea{I,\ldots, I}) ,(\tuplonea{\typone},
  \tuplonea{\typone, \ldots,\typone})}{\mu^i}{i\in\NN}.
\end{align*}
By the definition of a $\alpha$-bisimulation, we see that this
implies that $\alpha \geq \lvert \varepsilon^i - \mu^i \rvert$. Since
this reasoning can be done for every $\alpha$ such that
$\termone_{\varepsilon}$ and $\termone_{\mu}$ are $\alpha$-bisimilar,
it means that: $\appl {\mtbangt{\bang
    \typone}}{\termone_{\varepsilon}}{\termone_{\mu}} \geq \sup_{i \in
  \NN} \lvert \varepsilon^i - \mu^i \rvert$. Moreover, if we
consider the special case where $\varepsilon = 0$, we can actually
construct a $\mu$-bisimulation by taking $$\relone = (\stateterm{\bang
  \typone}{\termone_{0}}, \stateterm{\bang
  \typone}{\termone_{\mu}})\cup \{(\stateone_0 ,\statetwo_0 ) \} \cup
\{((0, \typtuplone), (\distrone, \typtuplone) \mid \sumdistr \distrone
\leq \mu \}.$$ We can easily check that $\relone$ is indeed a
$\mu$-bisimulation, which tells us that $\appl {\mtbangt{\bang
    \typone}}{\termone_{0}}{\termone_{\mu}} = \mu$. 
%
\end{example}

In Section \ref{fullabs} below, we will prove that $\mtbangt\typone$
coincides with $\mctxbang\typone$. While, as we will see soon, this
helps a lot when precisely evaluating the distance between terms,
sometime a characterization of $\mtbangt\typone$ by traces is
very effective when deriving lower bounds on the distance between
terms (see, e.g., Section \ref{tracecars}).

\subsection{A Trace Characterization}\label{tracecars}
\newcommand{\admtraces}[1]{\mathcal{A}(#1)} In this section, we will
characterize the metric defined in the previous section in an
\emph{inductive} way, by way of traces. This will be useful when
proving full abstraction. A \emph{trace} is a (possibly empty)
sequence of actions in $\ltslabels$. Formally, the set of traces is
generated by the following grammar:
$$
\traceone \in \traces \bnf \emptytr \midd \concat \actone
\traceone \qquad \text{ where } \actone \in \ltslabels.
$$
If $\traceone \in \traces$, and $\stateone \in \ltsstates$, we write
$\ltsact \stateone \traceone \statetwo$ if $\traceone = \actone_1,
\ldots, \actone_n$ and there exists a sequence of states
$\stateone_0,\ldots,\stateone_n$ where $\stateone_0 = \stateone$,
$\stateone_n = \statetwo$, and moreover for every $1\leq i <n$, it
holds that $\ltsact {\stateone_{i-1}}{\actone_i}{\stateone_i}$. A
trace $\traceone$ is said to be \emph{ admissible } for $\stateone$ if
there exists $\statetwo$ such that $\ltsact \stateone \traceone
\statetwo$. The set of admissible traces for $\stateone$
is indicated as $\admtraces{\stateone}$. The following is
a well-posed definition, because the underlying WLTS $\ltsbang$
is by definition deterministic.
\begin{definition}
  Let be $\stateone \in \ltsstates$, and
  $\traceone\in\admtraces{\stateone}$. Then there is a unique
  $\statethree$ such that $\ltsact {\stateone} \traceone \statethree$,
  and we define the \emph{success probability of $\traceone$ starting from $\stateone$}, that we note $\probt\traceone \stateone$, as $\weight \statethree$.
\end{definition}
We can express the metric $\metrtr\ltsbang$ as the maximum separation
that a trace can induce: 
\begin{proposition}\label{propcartr}
Let $\stateone,\statetwo\in\ltsstates$. Then:
$$
\appl{\metrtr\ltsbang}{\stateone}{\statetwo} =
\left\{
\begin{array}{ll} 
  \sup_{\traceone \in \traces} {\lvert \probt \traceone \stateone - \probt \traceone \statetwo \rvert}&\mbox{if $\admtraces{\stateone}=\admtraces{\statetwo}$};\\
  1&\mbox{otherwise}.
\end{array}
\right.
$$
\end{proposition}
\begin{proof}
We show separately the two inequalities:
\begin{varitemize}
\item
   First, if two states $\stateone$ and $\statetwo$ are related by an
   $\varepsilon$-bisimulation, we can see that $\lvert \probt
   \traceone \stateone - \probt \traceone \statetwo \rvert \leq
   \varepsilon$ (since the states obtained after having done every
   action in $\traceone$ are still $\varepsilon$-bisimilar).
 \item
   To show the other implication, it is sufficient to define, for every
   $\varepsilon$, a relation imposing a bound $\varepsilon$ on the success
   probabilities of traces, and to show that it is a
   $\varepsilon$-bisimulation.
\end{varitemize}
\end{proof}
Please observe that Proposition \ref{propcartr} gives an inductive
definition of $\mtbang$, while the notion of $\varepsilon$-bisimulation is defined
coinductively.

We conclude this section by stating a corollary of Proposition
\ref{propcartr} which will be useful in the following.
\begin{corollary}
\label{cor1}
Let be $\termone$ and $\termtwo$ of type $\typone$, such that $\stateterm \typone \termone$ and $\stateterm \typone \termtwo$ are \emph{not} $\varepsilon$-bisimilar. Then there exists a trace $\traceone$ such that 
$\lvert \probt \traceone {\stateterm \typone \termone} - \probt \traceone {\stateterm \typone \termtwo} \rvert
 \geq \varepsilon$.
\end{corollary}
\subsection{Full Abstraction}\label{fullabs}
In this section, we prove that $\mtbangt\typone$
coincides with $\mctxbang\typone$.
\subsubsection{Soundness}

We first of all show that the metric $\mtbangt\typone$ is \emph{sound} with
respect to $\mctxbang\typone$, i.e. that $\mtbangt\typone$ discriminates at least
as much as $\mctxbang\typone$:

\begin{theorem}[Soundness]
\label{soundness}
For any terms $\termone$ and $\termtwo$ of $\LBANG$, such that $\wfjt
{}\termone \typone$ and $\wfjt {}{\termtwo}{\typone}$, it holds that $
\appl {\mctxbang \typone} \termone \termtwo \leq \appl
      {\mtbangt \typone} \termone \termtwo $.
\end{theorem}
The rest of this section is devoted to the proof of Theorem \ref{soundness}. 
Please remember that our definition of the tuple distance is based on
the notion of $\varepsilon$-bisimulation. Proving the soundness
theorem, thus, requires us to show that for any terms $\termone$ and $\termtwo$
of type $\typone$ such that $\stateterm \typone \termone$ and
$\stateterm \typone \termtwo$ are $\varepsilon$-bisimilar, and for any $\typtwo \in \types,\contone  \in \ctxs \typone \typtwo$, it holds that $\mid\sumdistr
{\sem{\fillc \contone \termone}} - \sumdistr {\sem{\fillc \contone
    \termtwo}}\mid \leq \varepsilon$.

Our proof strategy is based on the fact that we can decompose every
evaluation path of a term in the form $\fillc \contone \termthree$
into \emph{external} reduction steps (that is, steps that do
\emph{not} affect $\termthree$), and \emph{internal} reduction steps
(that is, reduction steps affecting $\termthree$, but which can be
shown to correspond \emph{only} to actions from $\ltsbang$). 
Intuitively, if we reduce in parallel $\fillc \contone \termone$ and
$\fillc \contone \termtwo$, we are going to have steps where only the
context is modified (and the modification does not depend on whether
we are considering the first program or the second), and steps where
the internal part is modified, but these steps cannot induce too much
of a difference between the two programs, since the two internal
terms are $\varepsilon$-bisimilar.

We first of all need to generalize the notion of a context to one
dealing with tuples rather than terms. We in particular need contexts
with multiple holes having types which match those of the tuple (or,
more precisely, the $typtuplone$-state) they are meant to be paired
with. More formally:
\begin{definition}[Tuple Contexts]\label{validcont}
  \emph{Tuple contexts} are triples of the form $(\contone,
  \typtuplone, \typthree )$, where $\contone$ is an open term,
  $\typtuplone = (\ms \typone, \ms \typtwo)$ is a $(n,m)$ tuple type,
  and $\typthree$ is a type such that $\wfjt{\bang{\ms \varone_{\{ 1,
        \ldots,n\}}} : \bang{\ms \typone}, \ms \vartwo_{\{1, \ldots, m
      \}} : {\ms \typtwo}}{\contone}{\typthree} $. We note
  $\contextst$ the set of tuple contexts.  A tuple context $(\contone,
  \typtuplone, \typthree)$ is said to be an \emph{open value} if
  $\contone$ is of one of the following four forms: $\abstr \varone
  \termone$, $\abstrexp \varone \termone$, $\bang \termone$,
  $\vartwo_i$ (where $i \in \NN$).
\end{definition}
We now want to define \emph{when} a tuple context and an
$\typtuplone$-state) can be paired together, and the operational
semantics of such an object, which will be derived from that of
$\LBANG$ terms. This is the purpose of the following definition:
\begin{definition}[Tuple Context Pairs]\label{vct}
  We say that a pair $\statethree = (\contone, \stateone)$ is a
  \emph{tuple context pair} iff $\stateone = (\typtuplone, \distrone)$
  is an $\typtuplone$-state, and $\exists\typthree \in \types,
  \,(\contone, \typtuplone, \typthree)\in \contextst$. We indicate as
  $\ctsinfd$ the set of tuple context pairs.  Moreover, given such a
  $\statethree = (\contone,(\typtuplone, \distrone))$, we define
  $\forget{\statethree}$ as the (potentially infinite) distribution
  over $\tms$ given by:
  $$
  \forget{\statethree} = 
  \distrpar{\subst{\subst \contone {\ms \varone} {\ms \termone}}{\ms \vartwo}{\ms \termtwo}}{\distrone(\ms\termone,\ms\termtwo)}{(\ms \termone, \ms \termtwo) \in \supp (\distrone)}.
  $$
\end{definition}
Giving a notion of context distance for $\typtuplone$-states is now
quite easy and natural, since we know how contexts for such objects
look like. For the sake of being as general as possible, this notion
of a distance is parametric on a set of tuple contexts
$\contsone\subseteq\contextst$.
\begin{definition}\label{def:metr_ctx_tupl}
  Let be $\contsone\subseteq\contextst$ ,  $\typtuplone \in \typtuples$, and  
  $\stateone, \statetwo$ two $\typtuplone$-states. We define:
  $$
  \appl{\metrctx{\contsone}}\stateone \statetwo = \sup_{(\contone, \typtuplone, \typone) \in
            \contsone}\left\lvert
  \sumdistr {\sem{\forget{\contone, \stateone}}}- \sumdistr
            {\sem{\forget{\contone, \statetwo}}} \right\rvert 
  $$
\end{definition}
Unsurprisingly, the context distance between terms equals
$\metrctx{\contextst}$ when applied to $\typtuplone$-states
obtained through $\stateterm \typone\cdot$:
\begin{proposition}
If $\wfjt{}{\termone,
    \termtwo}{\typone}$, then $\appl{\mctxbang{\typone}}{\termone}{\termtwo} = \appl
  {\metrctx{\contextst}} {\stateterm \typone \termone}{\stateterm
    \typone \termtwo}$.
\end{proposition}

But why did we introduce $\ctsinfd$? Actually, these pairs allow for a
fine analysis of how tuples behave when put in a context, which in
turn is precisely what we need to prove Theorem~\ref{soundness}. This
analysis, however, is not possible without endowing $\ctsinfd$ itself
with an operational semantics, which is precisely what we are going to
do in the next paragraphs.

\paragraph{Semantics for $\ctsinfd$:\\}
In this paragraph, we define, for every $\stateonetd \in \ctsinfd$,
its semantics $\semct \stateonetd$, which is a (potentially infinite)
distribution over $\{\statetwotd \in \ctsinfd \mid \forget \statetwotd
\in \distrs\valset \}$.

Two relations need to be defined. On on the one hand, we need
a one-step \emph{labelled} transition relation $\osctnp$ which
turn an element of $\ctsinfd$ into a distribution over $\ctsinfd$
by perfoming an action. Intuitively, one step of reduction in $\osctnp$
corresponds to \emph{at most} one step of reduction in $\ltsbang$.
If that step exists, (i.e. if the \emph{term} is reduced) 
then the label is the same, and otherwise (i.e., if only the \emph{context} is reduced), 
the label is just $\tau$. We also need a multi-step approximation 
semantics $\wtrcctnp$ between elements of $\ctsinfd$ and subdistributions
over the same set. The latter is based on the former, and
both are formally defined in Figure \ref{osctl},
where 
\begin{varitemize}
\item
  $\contevone$ is an evaluation context; 
\item
  $\stateone$ is an $(n,m)$-state from $\ltsstates$;
\item
  $\stateonetd$ is a tuple-context pair from $\ctsinfd$;
\item
  For every context $\contone$, $\contone_{\remove{\isone}}$
  stands for the context
  $$
  \contone\{\vartwo_1/\vartwo_{1-\#\{j \mid j \in \isone \wedge j <1\}}\}
  \ldots\{\vartwo_n/\vartwo_{n-\#\{j \mid j \in \isone \wedge j <n\}}\}
  $$
\end{varitemize}

\begin{definition}
  We define a \emph{one-step labelled transition relation for
    $\ctsinfd$}, that we note $\osctl \stateonetd \actone \distrone$
  with $\stateonetd \in \ctsinfd$, $\actone \in \labelsexp \cup \{
  \tau \}$, and $\distrone$ a finite distribution over $\ctsinfd$, as
  specified in Figure~\ref{osctl}.\\ We call \emph{normal form} those
  elements in $\ctsinfd$ that cannot be reduced by $\osctl{}{}{}$, and
  we note $\nforms \ctsinfd$ the set of normal forms.
\end{definition}
Observe that we can actually give a characterization of normal forms:
\begin{equation}\label{carnf}
  \nforms \ctsinfd = \{(\contone, \stateone) \in \ctsinfd \mid
  \contone \text{ is an open value }\}.
\end{equation}
If $\distrone \in \distrs\ctsinfd$, we note $\valct \distrone =
\sum_{\stateonetd \in \nforms \ctsinfd} \distrone(\stateonetd) \cdot
\dirac {\stateonetd}$.

\begin{figure*}
\begin{center}
\fbox{
\begin{minipage}{\condwidth}
\begin{center}

\begin{center}
\scalebox{\condscale}{
 \AxiomC{\strut}
  \UnaryInfC{ $
\osctl{(\fillc \contevone{(\abstr \varthree \termtwo )\termone}, 
\stateone)}{\tau}
{\dirac{(\fillc \contevone {\subst \termtwo \varthree \termone},
\stateone)}}
$} \DisplayProof \quad
 \AxiomC{ \strut}
  \UnaryInfC{$
\osctl{(\fillc \contevone
    {(\abstrexp \varthree \termtwo ){\bang\termone}} , \stateone)}{\tau}{\dirac
{({\fillc \contevone {\subst \termtwo \varthree\termone}},
\stateone)}}$}
  \DisplayProof}
\end{center}
\begin{center}
\scalebox{\condscale}{
$$ \AxiomC{\strut}
 \UnaryInfC{ $
\osctl{(\fillc \contevone
      {\termone \oplus \termtwo} , \stateone)}{\tau}{
\diracpar{(\fillc \contevone \termone, \stateone )}{1/2} + 
\diracpar{(\fillc \contevone \termtwo , \stateone)}{1/2}
}
$} \DisplayProof$$
}
\end{center}
\begin{center}
\scalebox{\condscale}{
$$
\AxiomC{ $\ltsact{\stateone}{\actbang j}{\statetwo}$} 
 \UnaryInfC{ $
\osctl {(\fillc \contevone {\varone_j}, \stateone)}{\actbang j}{
\dirac{ (\fillc \contevone {\vartwo_{m+1}}, \statetwo)}}
 $} \DisplayProof
\qquad\quad
  \AxiomC{ $\ltsact{\stateone}{\evalbang j}{\statetwo} $}
\AxiomC{$\contone = \fillc
            \contevone {(\abstrexp \varthree
              \termtwo)\bang{\varone_{n+1}}}_{\remove{\{j\}}}  $} 
\BinaryInfC{ $
\osctl{(\fillc \contevone
      {(\abstrexp \varthree \termtwo){\vartwo_j}}, \stateone)}{\evalbang j}{\dirac{(\contone,\statetwo)}}$}
  \DisplayProof  
$$
}
\end{center}
\begin{center}
\scalebox{\condscale}{
$$ \AxiomC{ $
\begin{array}{c}
\strut \\ \strut \\
 \ltsact{\stateone}{\actappltupl \judgone j}{\statetwo}
\end{array}$
}
  \AxiomC{  $
\begin{array}{c}
\strut \\
\stateone = (\distrone, \typtuplone) \wedge 
{\typtuplone = \ms \typone, \ms \typtwo} \\ 
\ms{\typtwo}_j = \tarr \typfour \typfive   
\end{array}
$}
 \AxiomC{ $
\begin{array}{c}\strut \\
\judgone = (\{1, \ldots,n\}, \istwo, \ms{\typone}, \ms{\typtwo}, \valone, \typfour ) \in \vvalidjudgss \\ j \not \in \istwo
\end{array}
$} \TrinaryInfC{ $
\osctl{(\fillc \contevone {{\vartwo_j} \valone}, \stateone)}{\actappltupl \judgone j}{\dirac{(
(\fillc \contevone {\vartwo_{m+1}})_{\remove{ \istwo \cup \{j \}}}, \statetwo)}}
$} \DisplayProof$$ }
\end{center}
\begin{center}
\scalebox{\condscale}{
\AxiomC{ $\stateonetd \text{ in normal form for  } \osctl{}{\cdot}{}$}
\UnaryInfC{ $\wtrcct {\stateonetd} {\dirac \stateone} $} 
\DisplayProof}
\scalebox{\condscale}{
 \AxiomC{\strut}
\UnaryInfC{ $\wtrcct {\stateonetd} {0} $} 
\DisplayProof
\qquad
\AxiomC{ $\osctl \stateonetd \actone \distrone$}
\AxiomC{ $\wtrcct \statetwotd {\distrtwo_\statetwotd}$}
\BinaryInfC{ $\wtrcct {\stateonetd} {\sum_{\statetwotd \in \supp( \distrone)} \distrone(\statetwotd)\cdot \distrtwo_\statetwotd} $} 
\DisplayProof }
\vspace{10pt}
\end{center}
\end{center}
\end{minipage}}
\end{center}
\caption{Rules for $\wtrcct{}{}$} \label{osctl}
\end{figure*}

\begin{definition}
We define an \emph{approximation semantics relation for $\ctsinfd$}, that we note $\wtrcct
{\stateonetd}{\distrone}$, where $\stateonetd$ is an element of 
$\ctsinfd$, and $\distrone$ is a finite distribution over
$\nforms\ctsinfd$, as specified in Figure \ref{osctl}.\\
We define a \emph{semantics for $\ctsinfd$}, that we denote $\semct \stateonetd$
for $\stateonetd \in \ctsinfd$, as: $$\semct \stateonetd = \sup\{\distrone \mid \wtrcct \stateonetd \distrone \}.$$
  \end{definition}
  
\paragraph{Relating $\semct{}$ and $\sem{}$:\\}

We first show that this definition can indeed be related
to the usual semantics for terms. This takes the form of the following lemma: 
\begin{lemma}
\label{ctsem}
Let be $\statethree \in \ctsinfd$. Then:
\begin{varitemize}
\item $\{\distrone \mid \wtrcct \statethree \distrone\}$ is a directed
  set. We define $\semct \statethree$ as its least upper bound; 
\item $\forget{\cdot}:\distrs{\ctsinfd} \rightarrow \distrs
 \tms$ is continuous;
\item $\sem {\forget \statethree} = \forget {\semct \statethree}$.
\end{varitemize}
\end{lemma}
  \begin{proof}
    The proof of the first point is exactly the same as the one given in \cite{DalLagoZorzi} for the operational semantics of a probabilistic $\lambda$-calculus. The second and third ones are more involved.
    \paragraph{Continuity of $\forget{}$:}
    We have to show that, for any countable directed subset $I$ of $\distrs \ctsinfd$:
$$\sup \{\forget \distrone \mid \distrone \in I\} = \forget{\sup \{\distrone \mid \distrone \in I \}},$$
where we have extended in a natural way the definition of $\forget{\stateonetd}$ to $\forget \distrone$, with $\distrone \in \distrs \ctsinfd$.\\
 
We can first observe that, since $\forget {}$ is monotonous, it holds that:
$\forget{\sup \{\distrone \mid \distrone \in I \}} \geq \sup \{\forget \distrone \mid \distrone \in I\} $.
The other inequality is more involved. Let be $\distrtwo = \sup\{ \distrone \mid \distrone \in I\}$. Since $I$ is a countable directed set, there exists an increasing sequence $(\distrthree_n)_{n \in \NN}$ in $I$ such that $\distrtwo = \sup\{\distrthree_n \mid n \in \NN \} $. For $\termone \in \tms$, we note $\alpha_{\termone,(\contone, \stateone)}  = \sum_{\{\tuplone \mid \fillc \contone \tuplone = \termone\}} \stateone(\tuplone)$. We use it to express $\forget \distrtwo$:
$$\forget \distrtwo (\termone) = \sum_{(\contone, \stateone)\in \ctsinfd} \lim_{n \rightarrow \infty} \distrthree_n(\contone, \stateone) \cdot \alpha_{\termone,(\contone, \stateone) }.$$
Moreover, please observe that we can restrict ourselves to consider distributions over a countable subset of $\ctsinfd$ (which consists in the element of $\ctsinfd$ that can be generated by our language). Since $\distrtwo$ is a distribution on a countable set, we know that for every $\varepsilon$ there exists a finite subset $J_\varepsilon$ of $I$, such that $\sum_{(\contone, \stateone) \not \in J_\varepsilon} \distrtwo(\contone, \stateone) \leq \varepsilon$. So for every $\varepsilon$, it holds that:
$$\forget \distrtwo (\termone) = \sum_{(\contone, \stateone) \in J_\varepsilon} \lim_{n \rightarrow \infty} \distrthree_n(\contone, \stateone) \cdot \alpha_{\termone,(\contone, \stateone) } + \varepsilon.$$
Since we now have to consider only a finite sum, we can exchange sum and limit:
$$\forget \distrtwo (\termone) =  \lim_{n \rightarrow \infty} \sum_{(\contone, \stateone) \in J_\varepsilon} \distrthree_n(\contone, \stateone) \cdot \alpha_{\termone,(\contone, \stateone) } + \varepsilon.$$
And since that's true for every $\varepsilon$, we have the result.

\paragraph{$\sem {\forget \cdot} = \forget {\semct{\cdot}}$:}
We show separately the two inequalities. First, we show that, for every $\stateonetd \in \ctsinfd$, $\forget{\semct \stateonetd} \leq \sem{\forget {\stateonetd}} $.
The proof is based on the fact that Diagram~\eqref{diag:semantics} commutes:
\begin{equation}\label{diag:semantics}
\includegraphics[scale=1]{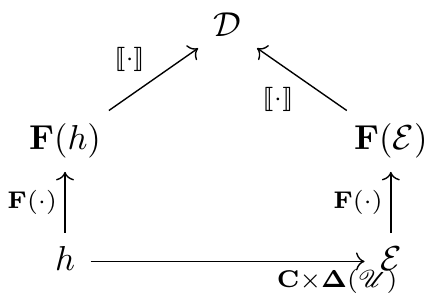}
\end{equation}
The result stated by~\eqref{diag:semantics} is formally given by the following lemma.
\begin{lemma}
\label{aux1}
For every $\stateonetd \in \ctsinfd$, it holds that
if $\osct \stateonetd \distrone$, then $\sem {\forget \stateonetd} = \sum_{\statetwotd \in \ctsinfd} \distrone(\statetwotd)\cdot\sem {\forget \statetwotd}$
\end{lemma}
\begin{proof}
The proof consists in considering every case in the definition of $\osct{}{}$, and use continuity result for $\sem{\cdot}$.
\end{proof}
We are now able to show that $\forget{\semct{\cdot}} \leq \sem{\forget{\cdot}}$:
\begin{proposition}
  \label{dir1}
  For every $\stateonetd \in \ctsinfd$, it holds that
$\forget{\semct{\stateonetd}} \leq \sem{\forget{\stateonetd}}$.
\end{proposition}
\begin{proof}
Observe that it is sufficient to show that
 $\wtrcct \stateonetd \distrone$ implies $\forget \distrone \leq \sem{\forget \stateonetd}$. We prove it by induction on the derivation of $\wtrcct \stateonetd \distrone$. 
\begin{varitemize}
\item If this derivation is  :
\AxiomC{}
\UnaryInfC{$\wtrcct {\stateonetd} {\emptyset} $} 
\DisplayProof , then $\distrone = \emptyset$, and the result holds.
\item If this derivation is 
 \AxiomC{$\stateonetd \text{ in normal form }$}
\UnaryInfC{$\wtrcct {\stateonetd} {\dirac \stateonetd} $} 
\DisplayProof, then $\forget \stateonetd$ is a distribution over values (it can be checked easily by considering the carcaterisation of normal forms given in Equation~\eqref{carnf}), and so $\sem {\forget \stateonetd} = \forget \stateonetd$, and the result holds.
\item If this derivation is:
\AxiomC{$\osct \stateonetd \distrthree$}
\AxiomC{$(\wtrcct \statetwotd {\distrtwo_\statetwotd})_{\statetwotd \in \supp( \distrthree)}$}
\BinaryInfC{$\wtrcct {\stateonetd} {\sum_{\statetwotd \in \supp( \distrthree)} \distrthree(\statetwotd)\cdot \distrtwo_\statetwotd} $} 
\DisplayProof, we apply Lemma \ref{aux1} : $\sem {\forget \stateonetd} = \sum_{\statetwotd \in \supp(\distrthree)} \distrthree(\statetwo)\cdot \sem{\forget \statetwotd} $. Moreover, we can apply the induction hypothesis to $\wtrcct \statetwo {\distrtwo_\statetwo}$ for every $\statetwotd \in \supp \distrthree$, which leads to: $\forget{\distrtwo_\statetwotd}\leq \sem {\forget \statetwotd}$. Using it, we obtain: 
\begin{align*}
\forget \distrone &= \sum_\statetwotd \distrthree(\statetwotd) \cdot \forget {\distrtwo_\statetwotd} \\
&\leq \sum_\statetwotd \distrthree(\statetwotd) \cdot \sem{\forget \statetwotd}\\
& = \sem{\forget \stateonetd}
\end{align*}, which is the result.
\end{varitemize}
\end{proof}
We are now going to show the other inequality, namely that, for any $\stateonetd \in \ctsinfd$, $\sem{\forget \stateonetd} \leq \forget{\semct \stateonetd}$. We first define a particular subset of $\ctsinfd$, namely those $\stateonetd$, such that $\forget \stateonetd$ can be seen as a term (remember that in the general case, it is a distribution over terms).
\begin{definition}
  We define a subset $\unary$ of $\ctsinfd$, as:
  $\unary =\{  (\contone, (\distrone, \typtuplone)) \in \ctsinfd \mid \distrone \text{ is a Dirac distribution}\}$.
  We define an operator $\unaryd{(\cdot)}$:
\begin{align*}
  \unaryd{(\cdot)} :\, &\distrs \ctsinfd \rightarrow \distrs \unary \\
   &\distrone \rightarrow \sum_{\stateonetd = (\contone, (\distrtwo, \typtuplone)) \in \ctsinfd} \distrone(\stateonetd) \cdot \sum_{\tuplone \in {\supp(\distrtwo)}} \distrtwo(\tuplone) \cdot \dirac{(\contone, (\dirac{\tuplone},\typtuplone ))}.
\end{align*}
\end{definition}
Observe that if
$\stateonetd \in \unary$, $\forget{\stateonetd}$ is a Dirac
distribution on terms. By abuse of notation, we'll use $\forget
\stateonetd$ to denote the term $\termone$ such that $\forget
\stateonetd = \dirac \termone$. The operator $\unaryd{(\cdot)}$ preserves $\semct{}$, as formally stated in Lemma~\ref{aux2} below:
\begin{lemma}
\label{aux2}
For every $\distrone \in \distrs \ctsinfd$, it holds that:
$\forget{\semct{\distrone}} = \forget{\semct{\unaryd \distrone}}$.
\end{lemma}
\begin{proof}
It is a consequence of the fact that $\ltsbang$ has the corresponding property: if $\stateone=(\distrone, \typtuplone) \in \ltsstates$, $\ltsact \stateone \actone {(\distrtwo, \typtupltwo)}$, then for every $\tuplone \in \supp (\distrone)$ there exists $\distrthree_\tuplone$ such that $\ltsact {(\dirac \tuplone, \typtuplone)}\actone {(\distrthree_\tuplone, \typtupltwo)}$, and moreover 
$\distrtwo = \sum_{\tuplone \in \supp(\distrone)} \distrone(\tuplone) \cdot \distrthree_\tuplone$.
\end{proof}

We introduce some notations on $\sssp{}{}$: we write $\ssspn \termone \distrone n$ with $n \in \NN$, if there exists a derivation $\derivone : \sssp \termone \distrone$ of size at most $n$. 
Moreover, we write $\ssspn \distrone \distrtwo n$ if there exists a finite set $I$ such that $\distrone = \sum_{i \in I} \alpha_i \cdot \dirac {\termone_i}$, and such that moreover $\ssspn{\termone_i}{\distrtwo_i}{n}$, and $\distrtwo = \sum_{i \in I}\alpha_i \cdot \distrtwo_i$. 
Please observe that this definition imply that $\distrone$ is always a \emph{finite} distribution over terms. We also introduce auxiliary  notations for $\osctl{}{}{}$, as given in Figure~\ref{osctplus}.

 Moreover, if $\distrone$, $\distrtwo$ are distributions over a set, we say that $\distrone$ is a \emph{finite approximation} of $\distrtwo$ if $\distrone$ is finite, and moreover $\distrone \leq \distrtwo$. We denote it by $\distrone \lfin  \distrtwo$.

\begin{figure*}
\begin{center}
\fbox{
\begin{minipage}{0.95 \textwidth}
\small{
$$ 
\AxiomC{}
\UnaryInfC{$\osctrefl\stateonetd {\dirac{\stateonetd}} $}
\DisplayProof \qquad
\AxiomC{$\osct \stateonetd \sum_{i \in I}{\alpha_i \cdot \dirac{\statetwotd_i}} $}
\AxiomC{$\osctrefl {\stateonetd_i} {\distrtwo_i} $}
 \BinaryInfC{$\osctplus \stateonetd { \sum_{i \in I} \alpha_i \cdot \distrtwo_i} $} \DisplayProof
\qquad
\AxiomC{$\osctplus \stateonetd \distrone $}
 \UnaryInfC{$\osctrefl \stateonetd {\distrone} $} \DisplayProof
$$
}
\end{minipage}}
\end{center}
\caption{Definition of $\osctplus{\stateonetd}{\distrone}$} \label{osctplus}
\end{figure*}

We first consider the case where $\stateonetd \in \unary$ is such that $\forget \stateonetd$ is a value. We can see that the following diagram commutes
\begin{equation}
\includegraphics[scale=1]{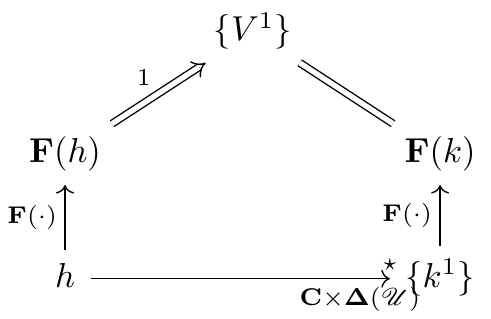}
\end{equation}, as formally stated in Lemma~\ref{lem:unaryval} below.

\begin{lemma}\label{lem:unaryval}
 Let be $\stateonetd \in \unary$, such that $\forget \stateonetd$ is a value. There exists $\statetwotd \in \unary$ such that $\osctrefl \stateonetd {\dirac{\statetwotd}} $ and $\statetwotd \in \nforms \ctsinfd$ verifies $\forget \statetwotd = \forget \stateonetd$.
\end{lemma}
\begin{proof}
The proof consists in showing that, if $\forget \stateonetd$ is a value, then there is a finite sequence $\stateonetd = \statetwotd_0 \ldots, \statetwotd_n$ with $\osct{\statetwotd_i}{\dirac{\statetwotd_{i+1}}}$, every $\statetwotd_i \in \unary$, $\statetwotd_n \in \nforms \ctsinfd$, and moreover $\forget{\statetwotd_i} = \forget{\stateonetd}$.
  \end{proof}

We now consider the case where $\forget \stateonetd$ is not a value: it means that we have to consider $\distrone \in \distrs \tms$, such that $\ssspn {\forget \stateonetd} \distrone n$, with $n>1$. Our aim is to show that $\distrone \leq \semct{\stateonetd}$. In that aim, we show first that Diagram~\ref{diag:variant} below commutes,

\begin{equation}\label{diag:variant}
  \includegraphics[scale=1]{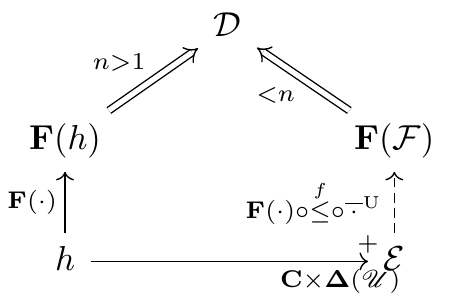}
\end{equation}
as formally stated by the following Lemma:
\begin{lemma}
\label{lemauxplus}
For any $\stateonetd \in \unary$, such that $\forget \stateonetd$ is not a value, there exists $\distrtwo \in \distrs \ctsinfd$ such that $\osctplus \stateonetd \distrtwo$, and for every $\distrone \in \distrs \tms$, and $n\geq 1$ such that $\ssspn {\forget \stateonetd} \distrone n$ ,  there exists $\distrthree$ such that $\distrthree \lfin \unaryd\distrtwo$, and moreover $\ssspn {\forget{\distrthree}}{\distrone} m$ with $m < n$.
  \end{lemma}

 Before doing the proof, we give some intuitions about this diagram:
we first consider these steps of $\osct{}{}$ that reduce more that one redex. It forces us to use not the one step transition $\onestepr{}{}$ to talk about the semantics for terms, but directly the approximation semantics $\sssp {}{}$. Let us consider now the steps of $\osct{}{}$ where no redex is reduced: more precisely, we can characterize them as the ones corresponding to actions passing a term that is already a value from the exponential part to the linear part of the inside tuple. The proof is made possible by the fact that we can show that there can only be a finite number of reduction step $\osct{}{}$ that don't reduce the underlying term, before there is a step reducing at least a redex in it.

\begin{proof}
We are first going to state a lemma about the approximation semantics $\sssp{}{}$, which we'll need later in the proof of Lemma \ref{lemauxplus}.
Indeed, in order to manage the problem of a step of $\osct{}{}$ corresponding to strictly more than one step of $\onestepr{}{}$, we need to transform a derivation $\derivone : \sssp \termone \distrone$ in derivations corresponding to reduction processes done after having completely reduced part of $\termone$. More precisely, 
the Lemma \ref{aux3} below gives us a way to extract of a derivation tree the subtrees corresponding to the evaluation after having reduced all the redex in a subterm in evaluation position. Please observe that it talks only about semantics of term, and express a form of compositionnality in the derivation. It will be used to treat the case of applicative actions, where we reduce many redex in one step of reduction for $\osct{}{}$.
\begin{lemma}
\label{aux3}
Let be $\termone$ a closed term which isn't a value, with $\ssspn{\fillc \contevone \termone}\distrone n$ a valid derivation of the approximation semantics. Then there exists a finite set $I \subseteq \NN$, and for every $i \in I$ a value $\valone_i$, a coefficient $\alpha_i$, and a valid derivation $\derivtwo_i : \ssspn {\fillc \contevone {\valone_i}}{\distrtwo_i} m $ such that  $n < m$,  
$\distrone = \sum_{i \in I} \alpha_i \cdot \distrtwo_i $, and moreover for every $\valone$,  $\sum_{i \mid \valone_i = \valone} \alpha_i \leq \sem \termone(\valone)$ 
\end{lemma}
\begin{proof}
The proof is by induction on $\derivone$.
\end{proof}
Let us now show the first part of Lemma \ref{lemauxplus}.
 Let $\stateonetd \in \unary$ such that $\forget \stateonetd$ is not a value. By definition of $\unary$, $\stateonetd$ is of the form:
$\stateonetd = (\contone,(\dirac{\tuplone}, \typtuplone)) $. We are going to construct the distribution $\distrtwo$ by case analysis on $\contone$ . Please observe that there is at most one possible action $\actone$ such that $\osctl{(\contone, (\dirac \tuplone, \typtuplone))}{\actone}{}$, and moreover it depends only on in $\contone$ and $\typtuplone$.  Consequently, we do a case analysis on the action $\actone$.
\begin{varitemize}
\item If $\actone = \tau$: it means that the reduced redex doesn't affect the inside tuple. 
  We take $\distrtwo$ such that  $\osctl {\stateonetd} \tau {\distrtwo = \sum_{1\leq i \leq n}\frac 1 n \dirac{(\contone_i, {(\tuplone, \typtuplone)})}}$. Observe that $\distrtwo$ is actually a finite distribution. Let be $\distrone \in \distrs \valset$, and $n \geq 1$, such that  $\ssspn {\forget \stateonetd} \distrone n$. 
Since $\forget \stateonetd \not \in \valset$, and $n \geq 1$,  $\ssspn {\forget \stateonetd} \distrone n$ is of the form:
$$
 \AxiomC{$\onestepr{\forget \stateonetd}{\termtwo_1, \ldots, \termtwo_p}$}
  \AxiomC{$(\ssspn{\termtwo_i}{\distrone_i}{m_i})_{1\leq i \leq p}$}
  \BinaryInfC{$\sssp{\forget \stateonetd}{\distrone = \sum {\frac 1 n} \cdot \distrone_i }$}
  \DisplayProof$$
  , where for every $i$, $m_i < n$.
  We take $\distrthree = \distrtwo$, $m = \max \{n_i \mid 1\leq i \leq p \} < n$ and the result is obtained by seeing that $\ssspn{\forget \distrthree = \sum_{1 \leq i \leq p}\frac 1 n \cdot \dirac{\termtwo_i}}{\distrone}{m}$.
 
\item If $\actone = \evalbang j $:  then $\contone = \fillc \contevone {(\abstrexp \varthree \termone){\vartwo_j}}$.

  It means that the only rule than can be used to show $\osct{\stateonetd}{\cdot}$ is:

$$  \AxiomC{\scriptsize $\ltsact{\stateone}{\evalbang j}{\statetwo} $}
 \UnaryInfC{\scriptsize $
\osctl{\stateonetd = (\fillc \contevone
  {(\abstrexp \varthree \termone){\vartwo_j}}, \stateone)}{\evalbang j}{\dirac{\statetwotd}}$}
 \DisplayProof$$
 where $\statetwotd =  ({\fillc
            \contevone {(\abstrexp \varthree
              \termone)\bang{\varone_{n+1}}}}_{\remove{\{j\}}},\statetwo)$. 

 Observe that $\statetwotd \in \{\statethreetd \in \unary \mid
 \osctl \statethreetd \actone \cdot \,  \text{ implies } \actone = \tau \}$. 

 Using the previous case, we know that there exist $\distrtwo_\statetwotd$ verifying the condition of Lemma \ref{lemauxplus} with respect to $\statetwotd$. Now, we are going to show that we can reuse it with respect to $\stateonetd$. First, notice that $\osctplus \stateonetd \distrtwo$ (since $\osctplus{}{}$ is transitive). So we can take $\distrtwo = \distrtwo_\statetwotd$. Let be $n \geq 1$, and $\distrone$ such that $\ssspn{\forget \stateonetd}{\distrone}{n} $. Since $\forget \stateonetd = \forget \statetwotd$, it holds that $\ssspn{\forget \statetwotd}{\distrone}{n}$. We can conclude by using the properties of $\distrtwo_\statetwotd$ with respect to $\statetwotd$.
 
\item  If  $ \actone =\actappltupl \judgone j$: it means that $\contone = \fillc \contevone{\vartwo_j \valone}$, and the rule used to reduce $\stateonetd$ by $\osctl{}{}{}$ is the applicative one. So we can take $\distrtwo$ defined as:
$$ \AxiomC{\scriptsize $
 \ltsact{\stateone}{\actappltupl \judgone j}{\statetwo}$}
  \AxiomC{\scriptsize $
\begin{array}{c}
\stateone = (\distrone, \typtuplone) \wedge 
{\typtuplone = \ms \typone, \ms \typtwo} \\ 
\ms{\typtwo}_j = \tarr \typfour \typfive   
\end{array}
$}
 \AxiomC{\scriptsize $
\begin{array}{c}
\judgone = (\{1, \ldots,n\}, \istwo, \ms{\typone}, \ms{\typtwo}, \valone, \typfour ) \in \vvalidjudgss \\ j \not \in \istwo
\end{array}
$} \TrinaryInfC{\scriptsize $
\osctl{(\fillc \contevone {{\vartwo_j} \valone}, \stateone)}{\actappltupl \judgone j}{\distrtwo = \dirac{(
(\fillc \contevone {\vartwo_{m+1}})_{\remove{ \istwo \cup \{j \}}}, \statetwo)}}
$} \DisplayProof$$
 We see that $\osctplus{\stateonetd}{\distrtwo}$. Let be $\distrone$, and $n \geq 1$, such that $\ssspn {\forget{\stateonetd}} \distrone n$. We need to define $\distrthree$ as specified in Lemma \ref{lemauxplus}. Notice that $\distrtwo \not \in \distrs \unary$, but we are able to express $\unaryd \distrtwo$ as:
 $$\unaryd \distrtwo = \sum_{\valtwo \in \valset} {\sem{\forget{\vartwo_j \valone,\stateone} }(\valtwo)} \cdot \dirac{\statetwotd_\valtwo},$$ where $\forget {\statetwotd_\valtwo} = \fillc{(\forget{\contevone, \stateone})}{\valtwo} $, and moreover $\statetwotd_\valtwo \in \unary$. 
We are going to apply Lemma \ref{aux3} to 
$\fillc{(\forget{\contevone, \stateone})}{\forget{\vartwo_j \valone,\stateone}} $
: there exist a finite set $I$, and a family $(\valtwo_i)_{i \in I}$ of $\valset$, such that 
$\ssspn{\sum_{i \in I} \alpha_i \cdot \dirac{\fillc{(\forget{\contevone, \stateone})}{\valtwo_i} }}{ \distrone} m$,
 with $m<n$, and for every $\valtwo \in \valset$, it holds that $\sum_{i \mid \valtwo_i = \valtwo} \alpha_i \leq \sem{\forget{\vartwo_j \valone,\stateone}}(\valtwo)$. Now we define $\distrthree = \sum \alpha_i \cdot \dirac{\statetwotd_{\valtwo_i}}$, and we can see that $\distrthree \lfin {\unaryd \distrtwo}$ and $\ssspn {\forget \distrthree}{\distrone} m$, and so we have the result.
\item If $\actone =\actbang j $: it means that $\contone = \fillc \contevone {\varone_j}$, and that the rule used is:
$$ 
 \AxiomC{\scriptsize $\ltsact{\stateone}{\actbang j}{\statetwo}$} 
 \UnaryInfC{\scriptsize $
\osctl {(\fillc \contevone {\varone_j}, \stateone)}{\actbang j}{
\dirac{ (\fillc \contevone {\vartwo_{m+1}}, \statetwo)}}
$} \DisplayProof $$
We define $\ms \termone$ and $\ms \termtwo$ such that $\tuplone = (\ms \termone, \ms \termtwo)$.
We can distinguish two cases, depending on $\ms \termone_j$:
\begin{varitemize}
\item If $\ms \termone_j$ is not a value. 
Then we are going to apply Lemma \ref{aux3}. Indeed, $\forget {\contone, \stateone} = \fillc{\forget {\contevone,\stateone}} {\ms \termone_j}$. The proof is essentially the same that for the applicative case.
 \item If $\ms \termone_j$ is a value. Then we can define $\statetwotd \in \unary$ by $\statetwotd = (\fillc \contevone {\vartwo_{m+1}}, \statetwo) $. It is possible that the only action able to reduce $\statetwotd$ is of the form $\actbang k$, However,
 we can iterate the previous reasonning (which intuitively corresponds to do reduction on $\ctsinfd$, but not in the actual term) only a finite number of times (because we can have only a finite number of occurrences of values that are in an evaluation position in $\termone = \forget {\contone, \stateone}$). So after a finite number of steps like that, we arrive to reduce this case to one we have already considered.

\end{varitemize}
\end{varitemize}

\end{proof}

\begin{lemma}\label{propaux}
For every $\stateonetd \in \unary$, if $\sssp {\forget \stateonetd} \distrone$, then $\distrone \leq {\forget{\semct{\stateonetd}}}$.
\end{lemma}
\begin{proof}
The proof is by induction on the size $n$ of the derivation $\sssp{\forget \stateonetd}{\distrone}$:
\begin{itemize}
\item If $n = 0$, then $\distrone = \emptyset$, and the result is true.
\item If $n = 1$, or $\distrone = \emptyset$, or $\distrone = \dirac{\forget \stateonetd}$, and then $\forget \stateonetd$ is a value, and we can apply the second point of Lemma \ref{lemauxplus}.
\item If $n>1$, then it implies that $\forget \stateonetd$ is not a value, and so we can apply the second point of Lemma \ref{lemauxplus}, and apply the induction hypothesis on the strictly smaller derivations it gives us.
\end{itemize}
\end{proof}

As a corollary, we can see that that $\sem{\forget{\stateonetd}} \leq \forget{\semct{\stateonetd}}$ holds in fact for every $\stateonetd \in \ctsinfd$, and not only for those in $\unary$:
\begin{proposition}
  \label{dir2}
  For any $\stateonetd \in \ctsinfd$, it holds that
$\sem{\forget{\stateonetd}} \leq \forget{\semct{\stateonetd}}$
\end{proposition}
\begin{proof}
First, we can see that Lemma \ref{propaux} imply that
$\sem{\forget{\stateonetd}} \leq \forget{\semct{\stateonetd}}$ for
every $\stateonetd \in \unary$. Lemma \ref{aux2} allows to generalize
it to the whole if $\ctsinfd$.
\end{proof}

We have now only to sum up Proposition~\ref{dir1} and \ref{dir2} to obtain the correspondence between $\semct{\cdot}$ and  $\sem{\cdot}$.
\end{proof}

\paragraph{Equivalence relations on $\ctsinfd$:}

Before proceeding, we need to understand how any reflexive and symmetric
relation
on $\ctsinfd$ can be turned into a relation on \emph{distributions} on
$\ctsinfd$.  If $\relone$ is a reflexive and symmetric relation on $\ctsinfd$, we
lift it to distributions over $\ctsinfd$ by stipulating that
$\distrone \relone \distrtwo$ whenever there exists a countable set
$I$, a family $(p_i)_{i \in I}$ of positive reals of sum smaller than
1, and families $(\stateonetd_i)_{i \in I}, (\statetwotd_i)_{i \in I}$
in $\ctsinfd$, such that $\distrone =\distrpar{\stateonetd_i}{p_i}{i
  \in I}$, $\distrtwo = \distrpar{\statetwotd_i}{p_i}{i \in I}$, and
moreover $\stateonetd_i \relone \statetwotd_i$ for every $i \in I$.

We now want to precisely capture \emph{when} a relation on
$\ctsinfd$ can used to evaluate the distance between tuple-context
pairs. 
\begin{definition}\label{osctlpr}
Let $\relone$ be a reflexive and symmetric relation on $\ctsinfd$. 
\begin{varitemize}
\item 
  We say that $\relone$ is \emph{ preserved by $\osctnp$}
  if, for any $\stateonetd, \statetwotd \in \ctsinfd$ such that
  $\stateonetd \relone \statetwotd$, if $\osctl \stateonetd \actone
  \distrone$, then there exists $\distrtwo$ such that $\osctl
  \statetwotd \actone \distrtwo$, and $\distrone \relone \distrtwo$.
\item 
  We say that $\relone$ is \emph{$\varepsilon$-bounding} if
  $\stateonetd \relone \statetwotd$, implies $\lvert
  \sumdistr{\forget{\stateonetd}}-
  \sumdistr{\forget{\statetwotd}}\rvert \leq \varepsilon$.
\item
  Let $\contsone$ be a set of tuple contexts, and $\stateone,
  \statetwo \in \ltsstates$ be two $\typtuplone$-states. 
  We say that $\relone$ is
  \emph{$\contsone$-closed with respect to $\stateone$ and
    $\statetwo$}, if for every $\contone$ and $\typthree$ such that
  $(\contone, \typtuplone, \typthree ) \in \contsone$, it holds that $
  (\contone, \stateone)\relone (\contone, \statetwo)$.
\end{varitemize}
\end{definition}
Please observe how any relation preserving $\osctnp$ and
being $\varepsilon$-bounding can be seen somehow as an
$\varepsilon$-bisimulation, but on tuple-context pairs.
The way we defined the lifting, however, makes it even
a \emph{stronger} notion, i.e., the ideal candidate for
an intermediate step towards Soundness.

Preservation by $\osctl{}{}{}$ leads us to preservation (in a weaker sense) by $\wtrcct{}{}$. That's the sense of the following lemma.
\begin{lemma}
Let $\relone$ be a reflexive and symmetric relation on $\ctsinfd$ preserved by $\osctl{}{}{}$. Let $\stateonetd, \statetwotd \in \ctsinfd$ be such that $\stateonetd { \relone} \statetwotd$. Let $\distrone$ be such that $\wtrcct\stateonetd \distrone$. Then there exists $\distrtwo$ such that $\wtrcct \statetwotd \distrtwo$, and $\distrone{ \relone} \distrtwo$.
\end{lemma}
\begin{proof}
The proof is by induction on the derivation of $\wtrcct \stateonetd \distrone$.
\begin{varitemize}
\item if the derivation is $\AxiomC{}\UnaryInfC{$\wtrcct{\stateonetd}{\emptyset} $}\DisplayProof$, the result holds.
\end{varitemize}
\end{proof}
\begin{lemma}\label{prboundrel}
If $\stateonetd, \statetwotd \in \ltsstates$ are such that there
exists a reflexive and symmetric relation $\relone$ preserved by
$\osctl{}{}{}$, $\varepsilon$-bounding, and containing $(\stateonetd,
\statetwotd)$, then it holds that ${\lvert \sumdistr {\forget{\semct
      \stateonetd}} - \sumdistr{\forget{\semct \statetwotd}}\rvert
  \leq \varepsilon}$.
\end{lemma}
We use Lemma \ref{prboundrel} to show the following proposition.
\begin{proposition}\label{keypropsound}
  Let $\contsone $ be a set of tuple contexts, $\stateone, \statetwo$
  two $\typtuplone$-states and $\relone$ a reflexive and transitive
  relation preserved by $\osctl{}{}{}$, $\varepsilon$-bounding, and
  $\contsone$-closed with respect to $\stateone$ and $\statetwo$. Then
  it holds that $\appl {\mctxbang \contsone} \stateone \statetwo \leq
  \varepsilon$.
\end{proposition}
\begin{proof}
It's a direct consequence of Lemma \ref{prboundrel}.
\end{proof}
Moreover, we see that the conditions from Definition~\ref{osctlpr} are enough
to guarantee that two terms are at context distance at most
$\varepsilon$.
\begin{proposition}\label{corsound}
  Let $\termone, \termtwo$ be two terms of type $\typone$. Suppose
  there exists a reflexive and symmetric relation $\relone$ on $\ctsinfd$, which
  is preserved by $\osctl{}{}{}$, $\varepsilon$-bounding, and
  $\contextst$-closed with respect to $\stateterm \typone \termone$
  and $\stateterm \typone \termtwo$. Then $\appl{\mctxbang \typone}
  \termone \termtwo \leq \varepsilon $.
\end{proposition}
\paragraph{$\varepsilon$-bisimulation and $\wtrcct{}{}$}:
What remains to be done, then, is to show that if two terms are
related by $\relone^\varepsilon$, then they themselves satisfy
Definition~\ref{osctlpr}. Compulsory to that is showing that
any $\varepsilon$-bisimulation can at least be turned into a relation
on $\ctsinfd$.

We need to do that, in particular, in a way guaranteeing the $\contsone$-closure 
of the resulting relation, and thus considering all possible
tuple contexts from $\contsone$:
\begin{definition}
Let $\relone$ be a reflexive and symmetric relation on $\ltsstates$. Let be
$\contsone$ a set of tuple contexts. We define its \emph{contextual lifting to
  $\ctsinfd$ with respect to $\contsone$} as the following binary
relation on $\ctsinfd$:
$$
{\widehat \relone^{\contsone}}_\typtuplone  = \bigcup_{(\contone,\typtuplone, \typone) \in \contsone} \{ ((\contone, \stateone), (\contone, \statetwo)) \mid \stateone, \statetwo \, \typtuplone\text{-states}, \,\stateone \relone \statetwo\}; \qquad
\widehat \relone^{\contsone}  = \bigcup_{\typtuplone \in \typtuples} {\widehat \relone^{\contsone}}_\typtuplone.
$$

\end{definition}
The following result tells us that, indeed, any
$\varepsilon$-bisimulation can be turned into a relation
satisfying Definition~\ref{osctlpr}: 
\begin{proposition}\label{proplifting}
Let $\relone$ be a $\varepsilon$-bisimulation. Then $\lifting
\relone^{\contextst}$ is preserved by $\osctl{}{}{}$ and
$\varepsilon$-bounding, and $\contextst$-closed with respect to every
$\stateone, \statetwo$ such that $\stateone \relone \statetwo$.
\end{proposition}

\begin{proof}
Let $\relone$ be an $\varepsilon$-bisimulation.
We obtain that $\lifting \relone^{\contextst}$ is $\contextst$-closed as a direct consequence of the lifting definition. Similarly, we see that $\lifting \relone^{\contextst}$ is $\varepsilon$-bounding as a direct consequence of the definition of a $\varepsilon$-bisimulation.
Now, let us show that  $\lifting \relone^{\contextst}$ is preserved by $\osctl{}{}{}$.
Let be $\typtuplone$ such that $(\stateonetd,\statetwotd) \in {\widehat \relone^{\contextst}}_\typtuplone$.
The proof is by induction on the derivation of $\wtrcct \stateonetd \distrone$. It is based on the following idea: for every $\contone$ such that there exist $\typone$ for which $(\contone, \typtuplone, \typone)$ is a tuple context:
\begin{varitemize}
\item or for every $\typtuplone$-state $\stateone$, $(\contone, \stateone)$ is in normal form. 
\item or there exists a family $(\contone_i)_{1 \leq i \leq n}$ such that $(\contone_i, \typtuplone,\typone) \in \contextst$, and moreover for every $\typtuplone$-state $\stateone$, $\osctl{(\contone, \stateone)}{\tau}{\sum_{1 \leq i \leq n} \dirac{(\contone_i, \stateone)}}$
\item or there exists an action $\actone$ of $\ltsbang$, and $\conttwo,\typtupltwo$ such that $(\conttwo, \typtupltwo, \typone )$, such that  for every $\typtuplone$-state $\stateone$, $\osctl{(\contone, \stateone)}{\actone}{\dirac{(\conttwo, \statetwo)}}$,  and $\ltsact{\stateone}{\actone}{\statetwo}$.
\end{varitemize}

\end{proof}

We are finally ready to give a proof of soundness:
\begin{proof}[of Theorem~\ref{soundness}]
  Consider two terms $\termone$ and $\termtwo$ of type $\typone$. Let
  $\varepsilon$ be $\appl{\mtbangt \typone} \termone \termtwo$.  We
  take $\relone^\varepsilon$ (defined in Definition \ref{dista} as the
  largest $\varepsilon$-bisimulation), and we see that $\relate
  {\relone^{\varepsilon}}{\stateterm \typone \termone}{\stateterm
    \typone \termtwo}$.  Proposition \ref{proplifting} tells us that
  we can apply Proposition \ref{corsound} to $\termone$, $\termtwo$,
  and $\lifting{(\relone^\varepsilon)}^{\contextst}$. Doing so we
  obtain that $\appl{\mctxbang \typone} \termone \termtwo \leq
  \varepsilon $, which is the thesis.\qed
  \end{proof}
\subsubsection{Completeness}
We can actually show that
$\mtbangt{\typone}$ is also complete with respect to the contextual
distance: that is the aim of this section.
\begin{theorem}[Full Abstraction]
\label{completness}
For every $\typone$, ${\mctxbang \typone} = \mtbangt \typone$.
\end{theorem}

\begin{proof}
  
    One inequality is given by Theorem \ref{soundness}.

    Please observe that Theorem \ref{soundness} gives us already half of Theorem \ref{completness}: indeed, it says that for every $\termone$ and $\termtwo$ of type $\typone$, it holds that: $\appl{\mctxbang \typone} \termone \termtwo \leq \appl{\mtbangt \typone} \termone \termtwo $. Our goal is now to show the other inequality, that is  $\appl{\mctxbang \typone} \termone \termtwo \geq \appl{\mtbangt \typone} \termone \termtwo $ .
    The proof of
    the other one is based on the trace characterization $\metrtr\ltsbang $ given in Section~\ref{tracecars}. More precisely, we use the fact that every trace can be simulated
    by a context: for every $\typtuplone \in \typtuples$, we construct a tuple context
$(\contone, \typtuplone,\typtwo)$, such that for every $\typtuplone$-state $\stateone$, it holds that $\probt\traceone \stateone = \obs{\forget {\contone, \stateone}}.$

  \begin{lemma}
    \label{lemc1}
    For every $\traceone \in \traces$, for every $\typtuplone \in \typtuples$ such that $\traceone$ is an admissible trace for $\typtuplone$, there exists a term $\contone $, and a type $\typthree$ such that:
    \begin{varitemize}
    \item $(\contone, \typtuplone, \typthree)$ is a tuple context.
    \item  for every $\stateone \in \ltsstates$ of the form $\stateone = (\distrone, \typtuplone) $, it holds that 
      $\sumdistr {\sem{\forget{ \contone,\stateone}}} =
      \probt  \traceone {\stateone}. $ 
    \end{varitemize}
  \end{lemma}
  \begin{proof}
    Let be $\typtuplone = (\ms \typone, \ms \typtwo)$ a type for tuple. Let be $(n,m)$ the arity of $\typtuplone$.
    We define a context $\contone$, and a type $\typthree$, by induction on the length of the admissible trace :
    \begin{varitemize}
    \item if $\traceone = \emptytr$: We want a context that terminates with probability $1$, whatever term we fill it with. We take: 
\begin{align*}      
\contone & = \abstr \varthree {\varthree \vartwo_1\ldots \vartwo_m}\\
\typthree & = \tarr{(\tarr{\tarr{\tarr{\ms \typtwo_1}{\ldots}}{\ms \typtwo_m}}{\vartypone} )}{\vartypone}
\end{align*}
    \item if $\traceone = \concat{\evalbang i }\tracetwo$. 
      Let be $\typtupltwo$ such that, whenever we do the action $\evalbang i$ from a $\typtuplone$-state, we obtain a $\typtupltwo$-state.
      Let
      $\conttwo$ and $\typfour$ the open term and type obtained by induction hypothesis applied to $\tracetwo$ and $\typtupltwo$. We take $\typthree = \typfour$, and:
      $$
        \contone  = (\abstrexp \varthree {\subst\conttwo{\varone_{n+1}}{\varthree}}_{\add{\{i\}}})\vartwo_i$$
    \item if $\traceone = \concat{\actbang i}\tracetwo$.
Let be $\typtupltwo$ such that, whenever we do the action $\actbang i$ from a $\typtuplone$-state, we obtain a $\typtupltwo$-state.
      Let
      $\conttwo$ and $\typfour$ be the open term and type obtained by induction hypothesis applied to $\tracetwo$ and $\typtupltwo$. We take $\typthree = \typfour$, and $\contone = \subst \conttwo {\varone_i}{\vartwo_{m+1}} $.
    \item if $\traceone = \concat{\actappltupl \judgone i}{\tracetwo}$. Since $\traceone$ is an admissible trace for $\typtuplone = (\ms \typone, \ms \typtwo)$, $\judgone$ should be of the form: 
      $\judgone = (\isone, \istwo, \ms \typone, \ms \typtwo, \termone, \typfive)$. Let
      $\typtupltwo$ be the corresponding tuple type obtained after the
      action $\actbang i $, and let be $\conttwo$ and $\typfour$ the
      term and the type given by the induction hypothesis
      and $\typtupltwo$. Then we take $\typthree = \typfour$:
      $$\contone = \subst {(\conttwo_{\add{\istwo \cup \{i \}}})}{\vartwo_i}{\vartwo_{n+1}\termone} $$
    \end{varitemize}
  \end{proof}
    When we combine Proposition \ref{lemc1} and Proposition \ref{propcartr}, we obtain the following corollary:
    \begin{corollary}
      \label{corcompletness}
     Let be $\stateone$ and $\statetwo$ two $\typtuplone$-states. 
      $$\appl {\metrtr\ltsbang} \stateone \statetwo \leq \sup \{\lvert \obs{\forget{\contone, \stateone}} - \obs{\forget{\contone, \statetwo}} \rvert \, \text{ s.t. } \exists \typone, \, (\contone,\typtuplone,\typone) \in \contextst \}  $$
    \end{corollary}
    \begin{proof}
      Let be $\stateone$ and $\statetwo$ two $\typtuplone$-states that are \emph{not} $\varepsilon$-bisimilar. Then there exists a trace $\traceone$ that shows it, that is verifying $\lvert {\probt \stateone \traceone} - {\probt \statetwo \traceone}\rvert > \varepsilon$. Now, we just have to take the  open term $\contone$ such that $(\contone,\typtuplone, \cdot )\in \contextst$ , corresponding to this trace $\traceone$. And we have $\lvert { \obs {\forget{ \contone ,\stateone}} - \obs{\forget{\contone, \statetwo}}}\rvert > \varepsilon$, and the result folds.
    \end{proof}
    
    We can now transform Corollary \ref{corcompletness} in a result on terms, which is exactly what we need to end the proof of theorem \ref{completness}.
    \begin{lemma}
      \label{lemcompletness1}
      Let be $\termone$, $\termtwo$ of type $\typone$. Then:
      $$\appl {\mtbangt \typone} \termone \termtwo \leq \sup \{\lvert \obs{\fillc \contone \termone} - \obs{\fillc \contone  \termtwo} \rvert \, \mid \exists \typtwo \text{ s.t. }\wfjt{\text{hole}:\typone}{\contone}\typtwo \}  $$
    \end{lemma}

    \begin{proof}
      Let us unfold the definition of  $\appl {\mtbangt \typone} \termone \termtwo$: please recall that we defined it as $\appl{\metrtr\ltsbang}{\stateterm \typone \termone}{\stateterm \typone \termtwo}$. Let be $\typtuplone$ such that $\stateterm \typone \cdot$ is an $\typtuplone$-term. Please recall that $\typtuplone = (\seq{},\seq{\typone}) $. We can see that for every $\contone, \typtwo$ such that $(\contone, \typtuplone, \typtwo)\in \contextst$, it holds that $\contone$ can in fact be see as a context for terms: more precisely, it holds that $\wfjt{\vartwo_1:\typone}{\contone}{\typtwo}$ and moreover,  it holds that
      $$\obs{\fillc{((\abstr {\vartwo_1} \contone) \text{hole})}{\termone}} = \obs{\forget{\contone, \stateterm \typone \termone}} \qquad \obs{\fillc{((\abstr {\vartwo_1} \contone)\text{hole})}{\termtwo}} = \obs{\forget{\contone, \stateterm \typone \termtwo}}.$$
      Since we can transform every tuple of context for $\typtuplone$ into a context on terms in that way, we obtain the result.
    \end{proof}
   
    Theorem \ref{completness} is a direct consequence of Lemma \ref{lemcompletness1}: indeed, please recall that we had to show that $\appl{\mctxbang \typone} \termone \termtwo \geq \appl{\mtbangt \typone} \termone \termtwo $. Since that's exactly what Lemma \ref{lemcompletness1} says, it ends the proof.

\end{proof}
\subsection{On an Up-to-Context Technique}
\label{subsect:Upto}
\subsubsection{Up-to $\varepsilon$-bisimulation}
As we have just shown, context distance can be characterized as a
coinductively defined metric, which turns out to be useful when
evaluating the distance between terms. In this section, we will go
even further, and show how an \emph{up-to-context}~\cite{Sangiorgi98}
notion of $\varepsilon$-bisimulation is precisely what we need to
handle our running example.

We first of all need to generalize our definition of a tuple: an
\emph{open tuple} is a pair $(\ms \termone, \ms \termtwo )$, where
$\ms \termone$ and $\ms \termtwo$ are sequences of (not necessarily
closed) typable terms.
\begin{definition}
 If $\tuplone = (\ms \termone, \ms \termtwo)$ is an open tuple, and
 $\typtuplone = (\ms \typthree, \ms \typfour)$ is a tuple type, we say
 that $(\ms \typone, \ms \typtwo, \tuplone, \typtuplone)$ is a
 \emph{substitution judgment} iff:
\begin{varitemize}
\item 
  $\wfjt{\bang{\ms \varone}: \bang {\ms \typone}}{{\ms
    \termone_i}}{{\ms \typthree}_i}$;
\item 
  if $n$ and $m$ are such that $\ms \typtwo$ is a $n$-sequence, and
  $\ms \termtwo$ a $m$-sequence, then there exists a partition $\{\isone_1, \ldots
  \isone_m\}$ of $\{1,\ldots, n \}$ such that
  $\wfjt{\ms{\vartwo}_{\isone_j}: \ms{\typtwo}_{\isone_j}}{\termtwo_j}{\typfour_j} $ 
  for every $j \in \{1, \ldots,m\}$.
\end{varitemize}
$\judgsubs$ is the set of all substitution judgments.
\end{definition}
If $\judgone = (\ms \typone, \ms \typtwo, \tuplone, \typtuplone) \in
\judgsubs$, and $\tupltwo \in \tuples$ is of type $(\ms \typone, \ms
\typtwo )$, then there is a natural way to form a tuple $\fillc \judgone
\tupltwo$, namely by substituting the free variables of $\tuplone$ by
the components of $\tupltwo$. In the following, we restrict
$\judgsubs$ to those judgments $\judgone$ such that for every
$\tupltwo$, terms in the linear part of $\fillc \judgone \tupltwo$ are
values. Observe that we always have $\wfjt{}{\fillc \judgone
  \tupltwo}{\typtuplone}$.  We extend the notation $\fillc \judgone
\tupltwo$ to distributions over $\tuples$: if $\distrone$ is a
distribution over tuples of type $(\ms \typone, \ms \typtwo)$, we note
$\fillc \judgone \distrone = \distrpar{\fillc \judgone
  \tupltwo}{\distrone(\tupltwo)}{\tupltwo \in \tuples}$, which is a
distribution over tuples of type $\typtuplone$.  Moreover, we want to
be able to apply our substitution judgments to the states of
$\ltsbang$. If $\stateone = (\distrone, (\ms \typone, \ms
\typtwo))\in\ltsstates$, and $\judgone = (\ms \typone, \ms \typtwo,
\tuplone, \typtuplone)$, the state of $\ltsbang$ defined by $( \fillc
\judgone \distrone, \typtuplone )$ will be often indicated as $\fillc
\judgone \stateone$.

\begin{example}\label{exsubs1}
We illustrate on a simple example the use of substitution judgments.
Let be $\typtwo$ any type.  Consider $\ms \typone = \tuplonea{\tarr
  \typtwo \typtwo}$, and $\ms \typtwo = \tuplonea{}$. Moreover, let
$\tuplone = (\tuplonea{\varone_1}, \tuplonea {I})$ and $\typtuplone =
(\tuplonea{\tarr \typtwo \typtwo}, \tuplonea{\tarr \typtwo
  \typtwo})$. Then $\judgone = (\ms \typone, \ms \typtwo, \tuplone,
\typtuplone)$ is a substitution judgment. We consider now a tuple of
type $(\ms \typone, \ms \typtwo)$. In fact, we take here a tuple that
will be useful in order to analyze our running example: $\tupltwo
=(\tuplonea{\psumindex {\Omega_!} \varepsilon I}, \tuplonea{}) $. By
substituting $\tupltwo$ in $\judgone$, we obtain $\fillc \judgone
\tupltwo = (\tuplonea{\psumindex {\Omega_!} \varepsilon I},
\tuplonea{I})$, and we can see easily that we obtain indeed a tuple of
type $\typtuplone$.
\end{example}

The main idea behind up-to context bisimulation is to allow for the
freedom of discarding any context when proving a relation to be a
bisimulation. This is captured by the following definition:
\begin{definition}
  Let $\relone$ be a relation on $\ltsstates$. $\relone$ is an
  \emph{$\varepsilon$-bisimulation up to context} if for every
  $\stateone$ and $\statetwo$ such that $\stateone \relone \statetwo$,
  the following holds:
  \begin{varitemize}
  \item 
    there exists $\typtuplthree \in \typtuples$ such that $\stateone =
    (\distrone, \typtuplthree) $, $\statetwo = (\distrtwo,
    \typtuplthree)$, and $\lvert \sumdistr \distrone - \sumdistr
    \distrtwo \rvert\leq \varepsilon$.
  \item 
    for any $\actone \in \labelsexp$, if $\ltsact \stateone \actone
    \statethree = (\distrone, \typtuplone)$ and $\ltsact \statetwo
    \actone \statefour = (\distrtwo, \typtuplone) $, then there exists
    a finite set $I \subseteq \NN$ such that:
    \begin{varitemize} 
    \item 
      there is a family of rationals $(p_i)_{i\in I}$ such that $\sum_{i \in I}
      p_i \leq 1$;
    \item 
      there are families $\ms \typone^i$, $\ms \typtwo^i$, and
      $\tuplone^i$, such that $\judgone_i = (\ms \typone^i, \ms
      \typtwo^i, \tuplone^i, \typtuplone)$ is a substitution judgment
      for every $i\in I$;
    \item 
      there are distributions over tuples $\distrone_i$, $\distrtwo_i$
      such that $(\distrone_i, \typtupltwo_i) \relone (\distrtwo_i,
      \typtupltwo_i)$;
    \end{varitemize}
    and moreover
      $\distrone = \sum\nolimits_{i \in I}p_i \cdot \fillc{\judgone_i}{\distrone_i} $, and 
      $\distrtwo = \sum\nolimits_{i \in I} p_i \cdot \fillc{\judgone_i}{\distrtwo_i}.$
  \end{varitemize}
\end{definition} 

The just introduced proof method is indeed quite useful when handling our running
example.
\begin{example}
We show that up-to bisimulations can handle our
running example. Please recall the definition of
$\termone_\varepsilon$ given in Example \ref{example1}.
First, we can see that, for every $a$, for every type $\typtwo$, $\stateterm {\bang{(\tarr
    \typtwo \typtwo)}} {\termone_{a}} = (\dirac{(\tuplonea{},
  \tuplonea{\bang{\psumindex {\Omega_!} {a} I}})},(\tuplonea{},
\tuplonea{\bang{\tarr \typtwo \typtwo}}) )$. We define a
relation $\relone$ on $\ltsstates$ containing $(\stateterm
{\bang{(\tarr \typtwo \typtwo)}} {\termone_{\varepsilon}},
\stateterm {\bang{(\tarr \typtwo \typtwo)}}
{\termone_{\mu}})$, and we show that it is an
$\gamma$-bisimulation up-to context for an appropriate
$\gamma$. In order to simplify the
notations, we define $\typtupltwo =({\tuplonea{\tarr\typtwo
    \typtwo}},\tuplonea{})$, and $\stateone_n, \statetwo_n \in \ltsstates$ as:
$$
\stateone_n =  (\diracpar{({\tuplonea{({\Omega_!} \oplus^{\varepsilon}I)}},\tuplonea{})} {({\varepsilon}^n)}, \typtupltwo), \qquad
\statetwo_n = (\diracpar{({\tuplonea{({\Omega_!} \oplus^{\mu}I)}},\tuplonea{})}{({\mu}^n )}, \typtupltwo). 
$$
Then, we define the relation $\relone$ as
$
\relone = \left\{\left(\stateterm \typone \termone, \stateterm \typone \termtwo \right) \right\} 
\cup \left\{(\stateone_n ,\statetwo_n) \mid n \in \NN \right\}.
$ 
One can check that $\relone$ is indeed a $\gamma$-bisimulation
up-to-context (where
$\gamma=\sup_{n\in\NN}\lvert\varepsilon^n-\mu^n \rvert$) by
carefully analysing every possible action. The proof is based on the following
observations:
\begin{varitemize}
\item 
  The only action starting from $\stateterm \typone\termone$ or
  $\stateterm \typone \termtwo$ is $\actone = \evalbang 1$, passing a
  term to the exponential part of the tuple, then we end up in
  $\stateone_0$ and $\statetwo_0$ respectively.
\item 
  If we start from $\stateone_n$ or $\statetwo_n$, the only relevant
  action is Milner's action $\actone = \actbang 1$,
  consisting in taking a copy of the term in the exponential part,
  evaluating it, and putting the result in the linear part. We can see
  (using the substitution judgment $\judgone$ defined in Example
  \ref{exsubs1}), that $\ltsact {\stateone_n}{\actone}{\fillc
    \judgone {\stateone_{n+1}}}$, and similarly $\ltsact
      {\statetwo_n}{\actone}{\fillc \judgone {\statetwo_{n+1}}}$, and
      the result follows.
\end{varitemize}
\end{example}
\begin{example}
  We give now an example illustrating a not-so-obvious constraint in
  the definition of a substitution judgment. Suppose we take $\tuplone
  = (\tuplonea{\varone_1}, \tuplonea {\varone_1})$, and $\typtuplone =
  (\tuplonea{\typone}, \tuplonea{\typone})$. Then $\judgone =
  (\tuplonea{ \typone}, \tuplonea{}, \tuplone, \typtuplone)$ is a
  not a substitution judgment, since $\varone_1$ appear in the linear
  part of $\tuplone$, and we need non-linear variables to type it. The
  constraint is there to avoid the situation in which,
  $(\dirac{(\tuplonea{\valone}, \tuplonea{})}, (\tuplonea{\typone},
  \tuplonea{}))$ and $(\dirac{(\tuplonea{\valtwo}, \tuplonea{})},
  (\tuplonea{\typone}, \tuplonea{}))$ would be up-to bisimilar
  for every $\valone$ and $\valtwo$ of type $\typone$.
\end{example}
\subsubsection{Soundness of up-to technique}
Bisimulations up-to contexts would be useless without a correctness
result like the following one:
\begin{theorem}
\label{uptos}
  If $\relone$ is an
  $\varepsilon$-bisimulation up-to context, then $\relone\subseteq\relone^\varepsilon$.
\end{theorem}
  The remaining of this section
  consists in the proof of Theorem \ref{uptos}. The proof is an extension of that of Theorem \ref{soundness}
  (although technically more involved).
  
 We first define a relation $\upto{\distrone}{\distrtwo}$ between
 distribution over $\ctsinfd$, which expresses the fact that
 $\distrtwo$ is obtained from $\distrone$ by changing the way we split
 our term into external environment and inside tuple:

\begin{definition}
Let be $\distrone$ and $\distrtwo$ two distributions over $\ctsinfd$. We write $\upto \distrone \distrtwo$ if:
\begin{varitemize}
\item $\distrone = \sum_{i \in I} \alpha_i \dirac{(\contone_i, \stateone_i)}$, $I$ countable set.
\item $\stateone_i$ is a $\typtuplone_i$-state.
\item For every $i$, there exist an open tuple $\tuplone_i $, a tuple type $\typtupltwo_i = (\ms \typone^i, \ms \typtwo^i)$, and $\statetwo_i$ a $\typtupltwo_i$-state, such that:
\begin{varitemize}
\item $(\ms \typone^i,\ms \typtwo^i, \tuplone^i, \typtuplone_i )$ is a substitution judgment,
\item $\stateone_i = \fillc {\judgone_i}{\statetwo_i} $
\item $\distrtwo = \sum_{i \in I} \alpha_i \cdot \dirac{\forget{\contone_i, ( \tuplone_i, \typtuplone_i)}, \statetwo_i}$ 
\end{varitemize}

\end{varitemize}
\end{definition}
Please observe that there are many possible distributions $\distrtwo$
that verify $\upto \distrone \distrtwo$.
We can express the fact that $\upto{}{}$ is a congruence relation that doesn't modify the underlying program,  and that moreover, the relation $\upto{}{}$ is designed to preserve normal forms with respect to $\osct{}{}$, by the following lemma:
\begin{lemma}
Let be $\distrone$, $\distrtwo$ such that $\upto \distrone \distrtwo$.
Then:
\begin{varitemize}
 \item $\forget \distrone = \forget \distrtwo$.
\item If $\supp(\distrone) \subseteq \nforms \ctsinfd$, it holds that $\supp (\distrtwo) \subseteq \nforms \ctsinfd$.
\end{varitemize}
\end{lemma}
\begin{proof}
It is a direct consequence of the unfolding of the definition of $\upto{}{}$.
\end{proof}

Now, we are ready to define
a relation $\uptor{\stateonetd}{\distrone}$, that intuitively
corresponds to do a step $\osct{}{}$, and then potentially change the
way we split programs between internal part and external part.

\begin{definition}
We define a relation $\uptor{\stateonetd}{\distrone}$, where $\stateonetd \in \ctsinfd$, and $\distrone$ a finite distribution over $\ctsinfd$, by the following rule:
$$ \AxiomC{$\osct \stateonetd \distrtwo $}
\AxiomC{$\upto \distrtwo  \distrone $} 
 \BinaryInfC{$\uptor \stateonetd \distrone $} \DisplayProof$$
\end{definition}

We define below the lifting of $\uptor{}{}$ to
distribution over $\ctsinfd$:
\begin{definition}
We define a one-step relation reduction $\uptor \distrone \distrtwo$ on (finite) distributions over $\ctsinfd$ as follows:
$$ 
\AxiomC{$\distrone = \valct \distrone + \sum_{i \in I} \alpha_i \cdot \dirac{\stateonetd_i} $}
\AxiomC{$\uptor {\stateonetd_i} {\distrtwo_i} $}
 \BinaryInfC{$\uptor \distrone {\valct \distrone + \sum_{i \in I} \alpha_i \cdot \distrtwo_i} $} \DisplayProof$$
\end{definition}
Since $\uptor{}{}$ is non-confluent, we cannot use
it to define a semantics for $\ctsinfd$ as we did
for $\osct{}{}$. We are going to define a notion of semantics for $\ctsinfd$ depending on what infinite sequance of reduction we are considering:
 for $\stateonetd \in \ctsinfd$, we call \emph{infinite reduction sequence starting from $\stateonetd$} a sequence $\seqone = (\distrone_n)_{n \in \NN}$ of distributions over $\ctsinfd$, where $\distrone_0 = \dirac{\stateonetd}$, and $\uptor {\distrone_n}{\distrone_{n+1}}$ for every $n \in \NN$.  The idea is that every reduction sequence gave a different semantics for $\stateonetd$, but that we obtain the same distribution over terms for all of them if we use $\forget{\cdot}$ on them. Let us formalize this idea. 
\begin{definition}
\label{semctu}
If $\stateonetd \in \ctsinfd$, and $\seqone = (\distrone_n)_{n \in \NN}$ a reduction sequence starting from $\stateonetd$. We call \emph{semantics of $\stateonetd$ with respect to $\seqone$}, and we note $\semctu \stateonetd \seqone$ the distribution over $\ctsinfd$ given by $\semctu \stateonetd \seqone = \sup\{\valct {\distrone_n} \mid n \in \NN \}$.
\end{definition}
Please observe that Definition \ref{semctu} is well-posed since the  $\valct {\distrone_n}$ are an increasing sequence of distributions over normal forms.

We state below the analogue of the third point of Lemma~\ref{ctsem} for up-to bisimulation:
\begin{proposition}
\label{keylem}
Let be $\stateonetd \in \ctsinfd$, and an infinite reduction sequence $\seqone = \uptor{\distrone_0 = \dirac{\stateonetd}}{\uptor{\distrone_1 \ldots}{\distrone_n \ldots}}\,$. Then $\forget{\semctu {\stateonetd}\seqone} = \sem{\forget {\stateonetd}} $.
\end{proposition}

\begin{proof}
  The proof is similar to the one of Lemma~\ref{ctsem}, and is based on the following two lemmas, that are a variant of Lemma~\ref{aux1} and \ref{lemauxplus}:
the following first lemma is used to show that $\forget{\semctu{\stateonetd}\seqone} \leq \sem{\forget{\stateonetd}}$.

  \begin{lemma}
    Let be $\distrone$, $\distrtwo$ such that $\uptor \distrone \distrtwo$. Then $\sem{\forget  \distrone} = \sem{\forget \distrtwo}$.
    \end{lemma}
  \begin{equation}\label{diag:semanticsupto}
\includegraphics[scale=1]{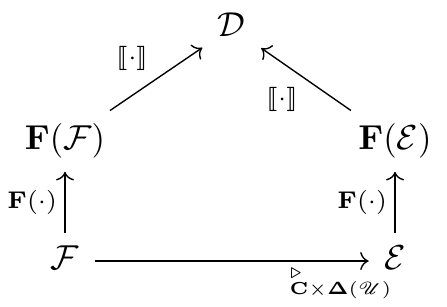}
  \end{equation}
That other lemma is used to show the other direction, namely that $\forget{\semctu{\stateonetd}\seqone} \geq \sem{\forget{\stateonetd}}$.
  \begin{lemma}
    Let be $\stateonetd \in \ctsinfd$ such that $\forget \stateonetd \not \in \valset$, and an infinite reduction sequence $\uptor{\distrone_0 = \dirac{\stateonetd}}{\uptor{\distrone_1 \ldots}{\distrone_m \ldots}}\,$. Then there exists $i \in \NN$, such that for every $\distrtwo$ and $n \in \NN$ with $\ssspn {\forget \stateonetd} \distrtwo n$, there exists a finite $\distrthree \in \distrs \unary$, and an infinite reduction sequence
    $\uptor{\distrthree_i = \distrthree}{\uptor{\distrthree_{i+1} \ldots}{\distrthree_m \ldots}}$, verifying $\forall i, \, \unaryd{\distrthree_i} \leq \unaryd {\distrone_i} $, and such that moreover $\ssspn {\forget \distrthree} \distrtwo m$, with $m <n$.
  \end{lemma}
  It can be expressed by the following diagram:
  \begin{equation}\label{diag:variantupto}
    \includegraphics{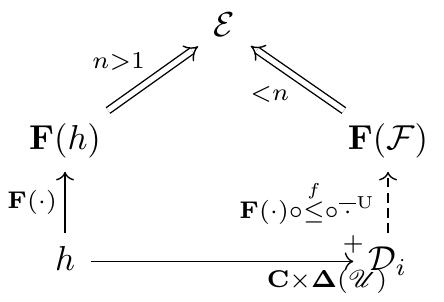}
\end{equation}
\end{proof}
Observe that, even if $\relone$ is an up-to bisimulation, it does not necessarily means that $\lifting \relone^{\contextst}$ preserves $\osctl{}{}{}$. We want to express the fact that it preserves $\uptor{}{}$, but we have to be more careful, since this reduction relation is non-deterministic. To deal with it, we introduce a more general preservation notion:
\begin{definition}
 We say that a relation $\relone$ over $\ctsinfd$ is \emph{weakly preserved by $\uptor{}{}$} if, for any $\stateonetd, \statetwotd \in \ctsinfd$ such that $\stateonetd \relone \statetwotd$, and $\stateonetd \not \in \nforms\ctsinfd$, then there exists $\distrone$ and $\distrtwo$ such that $\uptor \stateonetd \distrone$, $\uptor \statetwotd \distrtwo$, and $\distrone \relone \distrtwo$.
\end{definition}
Our notion of weak preservation is enough to state the following variant of Proposition~\ref{keypropsound}.
\begin{proposition}\label{keypropsoundupto}
  Let $\contsone $ be a set of tuple contexts, $\stateone, \statetwo$ two $\typtuplone$-states and $\relone$ a reflexive and symmetric relation  weakly preserved by $\uptor{}{}{}$, $\varepsilon$-bounding, and $\contsone$-closed with respect to $\stateone$ and $\statetwo$. Then it holds that $\appl {\mctxbang \contsone} \stateone \statetwo  \leq \varepsilon$. 
\end{proposition}
An $\varepsilon$ up-to bisimulation verifies the condistions of Proposition\ref{keypropsoundupto}, as stated by the analogue of Proposition \ref{proplifting} below.
\begin{proposition}\label{propliftingupto}
Let be $\relone$ an up-to bisimulation. Then $\lifting \relone^{\contextst}$ is  weakly preserved by $\uptor{}{}{}$ and $\varepsilon$-bounding, and $\contextst$-closed with respect to every $\stateone, \statetwo$ such that $\stateone \relone \statetwo$.
  \end{proposition}

We are now ready to do 
the proof of Theorem \ref{uptos}. Let be $\relone$ an $\varepsilon$ up-to bisimulation, ans $\stateone, \statetwo \in \ltsstates$ such that $\relate \relone \stateone \statetwo$. Using our completness result stated in Theorem~\ref{completness}, we see that it is enough to show that $\appl {\mctxbang \contextst} \stateone \statetwo  \leq \varepsilon$. Propositions~\ref{keypropsoundupto} and~\ref{propliftingupto} give us immediately the result.

\begin{example}\label{exa:nontrivial}
  We can exploit the soundness of up-to bisimulation to obtain the
  contextual distance for our running example, and conclude that
  $\appl {\mctxbang{\bang{(\tarr \typtwo \typtwo)}}}
  {\termone_{\varepsilon}}{\termone_{\mu}} =
  \sup_{n\in\NN} \lvert \varepsilon^n - \mu^n \rvert$.
  The context distance between $\termone_{\varepsilon}$ and
  $\termone_{\mu}$ is thus \emph{strictly} between $0$ and
  $1$ whenever $\lvert\varepsilon-\mu\rvert$
\end{example}

\section{Probabilistic $\lambda$-Calculi, in Perspective}\label{sect:nonlinear}
The calculus $\LBANGOR$ we analyzed in this paper is, at least
apparently, nonstandard, given the presence of parallel disjunction,
but also because of the linear refinement it is based on. In this
section, we will reconcile what we have done so far with calculi in
the literature, and in particular with untyped probabilistic
$\lambda$-calculi akin to those studied, e.g., in
\cite{DalLagoSangiorgiAlberti2014POPL,CrubilleDalLago2014ESOP}.

We consider a language $\LOPLUS$ defined by the following
grammar:
$$
\termone \in \LOPLUS \bnf \varone \midd \termone \termone  \midd
\abstr \varone \termone  \midd \psum \termone  \termone.
$$ 
\subsection{On Stable Fragments of $\lmcbang$.}
Our objective in this section is to characterize various notions
of context distance for $\LOPLUS$ by way of appropriate embeddings
into $\LBANG$, and thus by the LMC $\lmcbang$. It is quite
convenient, then, to understand when any fragment of
$\lmcbang$ is sufficiently \emph{robust} so as to be somehow
self-contained:
\begin{definition}
  We say that the pair $(\fstone, \flabels)$, where $\fstone \subseteq
  \statesexp$, and $\flabels \subseteq \typtuples \times \labelsexp$ is
  a \emph{stable fragment of $\lmcbang$} iff for every pair
  $(\typtuplone,\actone) \in \flabels$, for every $\typtuplone$-state
  $\stateone$, and for every $\statetwo \in
  \mcstates$ such that $\probexp(\stateone, \actone, \statetwo)>0$, it
  holds that $\statetwo \in \fstone$.
\end{definition}
Using a stable fragment of $\lmcbang$, we can restrict the WLTS
$\ltsbang$ in a meaningful way. The idea is that we only consider some
of the states of $\ltsbang$, and we are able to choose the
possible actions depending on the type of the state of $\ltsbang$ we
consider.
\begin{definition}
If $\flmone = (\fstone, \flabels)$ is a stable fragment of $\lmcbang$,
we define a WLTS by $\ltsf{\flmone} = (\ltsfst{\ltsf{\flmone}},
\ltsflabels{\ltsf{\flmone}}, \ltsftrans{\flmone},w_{{\flmone}})$, as
$$
\ltsfst{\ltsf{\flmone}} = \bigcup\nolimits_{\typtuplone \in \typtuples} \distrs{\{\tuplone \mid (\tuplone, \typtuplone)\in \fstone \}}\times \{ \typtuplone\}; \qquad
\ltsflabels{\ltsf{\flmone}} = \bigcup\nolimits_{(\typtuplone, \actone) \in\flabels}\{\actone\}
\cup\typtuples;$$
$$\ltsactf{}{\cdot}{}{\flmone} = \ltsact{}{\cdot}{}\cap {\{((\distrone, \typtuplone), \actone, \statetwo) \mid \supp(\distrone)\subseteq \fstone, \, (\typtuplone, \actone) \in \flabels  \}};
$$
and $w_{{\flmone}}$ is defined as expected.
\end{definition}
We want to be able to define a notion of distance on a \emph{fragment} of the original
language $\LBANG$, so that it verifies the soundness property for a
\emph{restricted} set of contexts. To do that, we need the restricted set of
contexts $\contsone$ to be preserved by the stable fragment:
\begin{definition}
  Let $\flmone = (\fstone, \flabels)$ be a stable fragment of
  $\lmcbang$. Let $\contsone$ be a set of tuple contexts. We say that
  $\contsone$ \emph{is preserved by $\flmone$}, if the following
  holds: for any $(\contone, \typtuplone,\typthree ) \in \contsone$
  that is not an open value and any $\typtuplone$-state
  $\stateone$ in $\ltsfst{\ltsf{\flmone}}$, there exists $\actone$ such that
  $(\typtuplone,\actone) \in \flabels \bigcup (\typtuples \times \{
  \tau \})$, $\osctl{(\contone, \stateone)}{\actone}{\distrtwo}$, and
  moreover:
  $$ 
  \supp(\distrtwo)\subseteq \bigcup\nolimits_{\typtupltwo \in
    \typtuples} \{(\conttwo,\statetwo) \mid \statetwo \text{ a
  }\typtupltwo \text{-state} \wedge \exists \typfour \text{
    s.t. }(\conttwo, \typtupltwo,\typfour) \in \contsone \} 
  $$
\end{definition}
We are now able to provide guarantees that the contextual distance
$\metrctx{\contsone}$ with respect to our restricted set of
contexts $\contsone $ is smaller that the distance defined on
the LTS $\ltsf{\flmone}$ induced by our stable fragment
$\flmone$. In the following, we assume to have fixed a
  stable fragment $\flmone = (\fstone, \flabels)$ of $\lmcbang$, and
  $\contsone$ a set of tuple contexts preserved by $\flmone$.

\begin{proposition}\label{propofrstclos}
Let be $\relone$ an $\varepsilon$-bisimulation on $\ltsf \flmone$. Then $\lifting \relone^{\contsone}$ is preserved by $\osctl{}{}{}$, $\varepsilon$-bounding, and $\contsone$-closed with respect to every $\stateone, \statetwo$ such that $\relate \relone \stateone \statetwo$. 
\end{proposition}
\begin{proof}
\begin{varitemize}
\item The fact that $\lifting\relone^{\contsone}$ is $\contsone$-closed is a direct consequence of the definition of lifting.
\item The fact that it is $\varepsilon$-bounding is a direct consequence of the definition of a $\varepsilon$-bisimulation.
\item It is preserved by $\osctl{}{}{} $: it is indeed guarantee by our definition of preservation of $\contsone$ by $\flmone$.
\end{varitemize}
\end{proof}

\begin{proposition}\label{frtrs}
 Let $\flmone = (\fstone, \flabels)$ be a stable fragment of
 $\lmcbang$, $\contsone$ a set of tuple contexts preserved by $\flmone$,
and $\stateone, \statetwo \in \ltsfst{\ltsf{\flmone}}$.
Then $
  \appl{\metrctx{\contsone}}{\stateone}{\statetwo}\leq \appl {\metrtr
    {\ltsf{\flmone}}}{\stateone}{\statetwo}. $
\end{proposition}
\begin{proof}
Let be $\varepsilon = \appl {\metrtr {\ltsf{\flmone}}}{\stateone}{\statetwo}$. Let be $\relone$ the biggest $\varepsilon$-bisimulation on $\ltsf \flmone$. By lemma \ref{largestbisim}, it holds that $\relate \relone \stateone \statetwo$. Using Proposition \ref{propofrstclos} and \ref{keypropsound}, we obtain that $\appl {\mctxbang \contsone} \stateone \statetwo  \leq \varepsilon$, which is exactly the result.   
\end{proof}

In the following, we make use of Proposition \ref{frtrs} on stable
fragments corresponding to embeddings of $\LOPLUS$ into $\LBANG$. We
will consider two different encodings depending on the underlying
notion of evaluation.
\subsection{Call-by-Name}
\begin{figure*}
  \begin{center}
  \fbox{\scriptsize
    \begin{minipage}{.95\textwidth}
      \begin{center}
        \scalebox{\condscale}{
$
\contevone \bnf \hole \midd \contevone \termone \qquad 
 $}
\scalebox{\condscale}{$
      \AxiomC{$\strut $}
      \UnaryInfC{$\redonestepr
        {\psum\termone \termtwo }{ \termone ,\termtwo} $} \DisplayProof
      \qquad
      \AxiomC{$\strut $}
      \UnaryInfC{$\redonestepr {(\abstr \varone
          \termone )\termtwo }{\subst\termone \varone \termtwo} $}
      \DisplayProof
      \qquad
      \AxiomC{\strut $\redonestepr \termone {\termtwo_1, \ldots, \termtwo_n}$}
      \UnaryInfC{$\onestepr {\fillc \contevone \termone }{{\fillc \contevone {\termtwo_1}}, \ldots {\fillc \contevone {\termtwo_n}}} $}
      \DisplayProof 
      $}
      \end{center}
  \end{minipage}}
  \end{center}
\caption{One-step CBN Semantics}\label{cbnsem}
\end{figure*}
$\LOPLUS$ can first of all be endowed with call-by-name semantics, as
in Figure \ref{cbnsem}. We use it to define an approximation semantics
exactly in the same way as in Figure \ref{apprsem}, and we take as
usual the semantics of a term to be the least upper bound of its approximated
semantics. Moreover, we denote by $\mctxcbn$ the context distance on
$\LOPLUS$, defined the natural way. We are going, in the remainder of
this section, to use our results about $\LBANG$ to obtain a
characterization of $\mctxcbn$.
\subsubsection{The Call-By-Name Embedding}
Girard's translation \cite{Girard87} gives us an embedding
$\embcbn{\cdot}:\LOPLUS \rightarrow \LBANG$,
defined as follows:
\begin{align*}
\embcbn{\varone} &= \varone &
 \embcbn{\abstr \varone \termone} &=\abstrexp \varone {\embcbn \termone} \\
 \embcbn{\termone \termtwo} &=\embcbn{\termone}\bang{\embcbn \termtwo} &
 \embcbn{\psum\termone\termtwo} &= \psum {\embcbn\termone}{\embcbn \termtwo}
\end{align*}
Please observe that $\embcbn {\cdot}$ respects typing,
in the sense that, when we define $\typecbn = \rectype
{\vartypone}{\tarr{\bang \vartypone} \vartypone}$ , it holds that for
every term $\termone$ of $\LOPLUS$ whose free variables are in $\{
\varone_1, \ldots, \varone_n\}$,  we can show that $\wfjt
       {\bang\varone_1: \bang\typecbn, \ldots, \bang\varone_n : \bang
         \typecbn}{\embcbn \termone} \typecbn. $

This definition allows us to have some useful properties:
\begin{lemma}\label{embcbnprops}
\begin{varitemize}
\item Let be $\termone, \termtwo \in \LOPLUS$, and $\varthree \in \vars$. Then $\embcbn{\subst \termone \varthree \termtwo}= \subst{\embcbn \termone}{\varthree}{\embcbn \termtwo}$;
\item Let be $\termone$ a closed term in $\LOPLUS$. Then $\embcbn{\sem \termone} = \sem{\embcbn \termone}$.
\item Let be $\termone \in {\embcbn \LOPLUS}$, such that there exists an evaluation context $\contevone$, and $\termtwo \in \LOPLUS$, verifying $\embcbn \termone = \fillc \contevone {\embcbn \termtwo}$. There for any term $\termthree \in \LOPLUS$, it holds that $\fillc \contevone{\embcbn \termthree} \in \embcbn \LOPLUS$. 
\end{varitemize}
\end{lemma}

\subsubsection{Metrics for $\LOPLUS$}
It is very tempting to define a metric on $\LOPLUS$ just
as follows:
$ \appl {\mtcbn}{\termone}{\termtwo} = \appl
{\mtbangt {\bang\typecbn}}{\bang {\embcbn \termone}}{\bang{\embcbn
    \termtwo}}$.
We can easily see that it is sound with respect
to the context distance for $\LOPLUS$, since any context of this 
language can be seen, through $\embcbn{\cdot}$, as a context in $\LBANG$.
However, it is not complete, as shown by the following example: 
\begin{example}\label{excbnnc}
We consider $\termone = \psum \Omega {(\abstr \varone {\Omega})}$ and 
  $  \termtwo = {(\abstr \varone {\Omega})}$. We can see that $\appl{\mtbangt {\bang\typecbn}}{ \bang {\embcbn \termone}}{\bang {\embcbn \termtwo}} = 1 $: indeed, when we define a sequence of $\LBANG$-contexts by
$\contone_n = \abstrexp \varone\left((\abstr {\vartwo_1} {\ldots \abstr
  {\vartwo_n} {(\abstr \varthree \varthree \vartwo_1, \ldots
    \vartwo_n)}})\varone \ldots \varone \right)[] $,
 we see that $\obs {\bang{\embcbn \termone}} = 1/{2^n}$ while $\obs{\bang{\embcbn \termtwo}} = 1$. But those contexts $\contone_n$ have more expressive power that any context in $\embcbn{\LOPLUS}$, since they do
something that none of the context from $\LOPLUS$ can do: they evaluate
a copy of the term, and then shift their focus to \emph{another} copy of
the term. It can be seen in the embedding: a term in $\embcbn \LOPLUS$
has never several redexes in linear position. We are now going to explicit this idea, and show that
$\appl{\mctxcbn}{\termone}{\termtwo} = \frac 1 2 <
\appl{\mtcbn}{\termone}{\termtwo} $.
\end{example}

\begin{figure}
\begin{center}
\fbox{
\begin{minipage}{\condwidth}
\begin{center}
\scalebox{\condscale}{
$
\typtuplone = (\tuplonea{}, \tuplonea{\bang {\typecbn}})  \qquad 
\typtupltwo = (\tuplonea{\typecbn},\tuplonea{})\qquad
\typtuplthree_n = (\tuplonea{\typecbn},{\lbrack{\underbracket{\typecbn, \ldots, \typecbn}_n}\rbrack} )
$}
\end{center}
\begin{center}
\scalebox{\condscale}{
\begin{tikzpicture}
\node [draw, rectangle, rounded corners] (M1) at (-0.75,0)
{$\begin{array}{c}(( {\tuplonea{}}, \tuplonea{\bang {\embcbn{\termone}}})\\, \typtuplone) 
\end{array}
$};
\node [draw, rectangle, rounded corners] (N1) at (-0.75,2){$
\begin{array}{c}
(({\tuplonea{}}, \tuplonea{\bang {\embcbn{\termtwo}}})\\, \typtuplone )
\end{array} $};
\node [draw, rectangle, rounded corners] (M2) at (2.35,0)
{$\begin{array}{c}
(( {\tuplonea{\embcbn \termone}}, \tuplonea{})\\, \typtupltwo) 
\end{array}
$};
\node [draw, rectangle, rounded corners] (N2) at (2.35,2){$\begin{array}{c}( ({\tuplonea{\embcbn \termtwo}}, \tuplonea{})\\ , \typtupltwo)\end{array} $};
\node [draw, rectangle, rounded corners] (M3) at (5.5,0)
{$\begin{array}{c}
( (\tuplonea{\embcbn \termone},\\ \tuplonea{\abstrexp \varone {\Omega_!}}), \typtuplthree_1 ) \end{array} $};
\node [draw, rectangle, rounded corners] (N3) at (5.5,2){$\begin{array}{c}
( (\tuplonea{\embcbn \termtwo}, \\ \tuplonea{\abstrexp \varone {\Omega_!}}), \typtuplthree_1) \end{array} $};

\node [draw, rectangle, rounded corners] (M4) at (10,0){$
\begin{array}{c}
  (( \tuplonea{\embcbn \termone}, \\\lbrack{\underbracket{\abstrexp \varone {\Omega_!}, \ldots ,\abstrexp \varone {\Omega_!}}_n}\rbrack) , \typtuplthree_n) \end{array} $};

\node [draw, rectangle, rounded corners] (N4) at (10,2){$\begin{array}{c}
((\tuplonea{\embcbn \termtwo},\\ {\lbrack{\underbracket{\abstrexp \varone {\Omega_!}, \ldots, \abstrexp \varone {\Omega_!}}_n\rbrack}}), \typtuplthree_n) \end{array} $};

\node [] (M5) at (7.5,0){$\ldots $};
\node [] (N5) at (7.5,2){$\ldots $};

\draw[->](M1) to node[above]{$\evalbang 1 $} node[below]{$1$} (M2);
\draw[->](N1) to node[above]{$\evalbang 1 $} node[below]{$1$} (N2);
\draw[->](N2) to node[above]{$\actbang 1 $}node[below]{$1$}(N3);
\draw[->](M2) to node[above]{$\actbang 1 $} node[below]{$\frac 1 2$} (M3);
\draw[->](N3) to node[above]{$\actbang 1 $}node[below]{$1$}(N5);
\draw[->](M3) to node[above]{$\actbang 1 $} node[below]{$\frac 1 2$} (M5);
\end{tikzpicture}}
\end{center}
\end{minipage}}
\end{center}
\caption{A Fragment of $\lmcbang$}\label{examplecbn}
\end{figure}

The way out consists in using the notion of stable fragment to
refine the Markov Chain $\lmcbang$ by keeping only the states and actions we
are interested in.

\begin{definition}
We define a stable fragment $\flmccbn$ as specified in Figure \ref{fig:lmccbn}, and a distance $\metrcbn$ on $\LOPLUS$ as:
$$\appl {\metrcbn}{\termone}{\termtwo} = \appl{\metrtr{\ltsf{\flmccbn}}}
{\statetermcbn \termone}{\statetermcbn \termtwo} ,$$
where $\statetermcbn \termone = (\dirac{(\seq{\embcbn{\termone}}, \seq {})},\typzero )$. 
\end{definition}

\begin{figure*}
\begin{center}
\fbox{
\begin{minipage}{\condwidth}
\begin{center}
\scalebox{\condscale}{
$\typzero = {(\seq{{\typecbn}}, \seq{})} 
\qquad \typun = (\seq{{\typecbn}}, \seq{\typecbn})
\qquad
\tupletermcbn \termone = ((\seq{\embcbn{\termone}}, \seq {}),\typzero )$}
\scalebox{\condscale}{
$ \statesexpcbn = \left(
\left\{ \tupletermcbn \termone \mid \termone \in \LOPLUS\right\} \cup 
\left\{(\seq{\embcbn{ \termone}}, \seq{\embcbn {\valone}}),\typun \mid \termone \in \LOPLUS, \valone \in \LOPLUS \text{ a normal form }\right\} \right) \bigcap \statesexp  $
}

\scalebox{\condscale}{
$\vjcbn = \vvalidjudgss \cap \{(\isone, \istwo, \ms{\typone}, \ms{\typtwo}, \termone, \typthree )\, , \, \termone \in \embcbn{\LOPLUS} \}$}

\scalebox{\condscale}{
$ \labelsexpcbn = \{\typun \} \times \{\actappltupl \judgone 1 \mid \judgone \in \vjcbn \} \cup \{\typzero \} \times \labelsbang  $
}\end{center}
\end{minipage}
}
\end{center}
\caption{The Stable Fragment $\flmccbn = (\statesexpcbn,\labelsexpcbn )$.} \label{fig:lmccbn}
\end{figure*}

We need now to define a set of tuple contexts preserved by $\flmccbn$, the aim of applying Proposition \ref{frtrs}.
\begin{definition}
$\contsonecbn$ is the smallest set of tuple contexts such that:
\begin{varitemize}
\item If $\termone \in \LOPLUS$ with $\fv \termone \subseteq \{\varone_1 \}$, then $(\embcbn \termone, \typzero, \typecbn ) \in \contsonecbn$;
\item If $(\contone, \typzero, \typecbn ) \in \contsonecbn$, and $\contone = \fillc \contevone {\varone_1}$, it holds that $(\fillc \contevone {\vartwo_1} ,\typun, \typecbn) \in \contsonecbn$.
\end{varitemize}
\end{definition}
$ \contsonecbn$ is designed to allow us to link $\mctxcbn$ and $\metrctx{\contsonecbn}$.
\begin{lemma}\label{contextcbnmetr}
For $\termone, \termtwo \in \LOPLUS$ closed terms,
$\appl \mctxcbn \termone \termtwo = \appl {\metrctx{\contsonecbn}}{\statetermcbn \termone}{\statetermcbn \termtwo}.$
\end{lemma}
\begin{proof}
The proof is based on Lemma \ref{embcbnprops}.
\end{proof}

\begin{lemma}\label{preservcbn}
$\contsonecbn$ is preserved by the stable fragment $\flmccbn$.
\end{lemma}
\begin{proof}
Let be $(\contone, \typtuplone, \typthree) \in \contsonecbn$, and $\stateone$ a $\typtuplone$-state in $\ltsf{\flmccbn}$. Using the definition of $\contsone$, we may see that $\typthree = \typecbn$, and moreover we are in one of the following cases:
\begin{varitemize}
\item Or $\typtuplone = \typzero$, $\contone = \embcbn \termone$ with $\fv{\termone}\subseteq\{\varone_1\} $ and $\stateone$ is of the form $(\distrone, \typzero )$, where $\distrone = \sum_{i \in \NN} p_i \cdot \dirac{(\seq{\embcbn \termtwo_i}, \seq{})}$.
Then we still have to consider separately two cases:
\begin{varitemize}
\item or $\embcbn \termone = \fillc \contevone{\varone_1}$. We take $\actone = \evalbang 1 $. We can easily check that $(\typzero, \evalbang 1 ) \in \labelsexpcbn \bigcup (\typtuples \times \{\tau \})$. Moreover, we see that there exists $\distrtwo$ such that $\osctl{(\contone, \stateone)}{\actone}{\distrtwo}$, and more precisely $\distrtwo = \dirac{(\fillc \contevone {\vartwo_1}, \statetwo)}$ where $\statetwo$ is a $\typun$-state.
Since $(\fillc \contevone {\vartwo_1}, \typun, \typecbn) \in \contsonecbn$ by construction, it holds that $\supp(\distrtwo)\subseteq \bigcup_{\typtupltwo \in \typtuples} \{(\conttwo,\statetwo) \mid \statetwo \text{ a }\typtupltwo \text{-state} \wedge \exists \typfour \text{ s.t. }(\conttwo, \typtupltwo,\typfour) \in \contsonecbn \} $, which shows the result.
\item or $\embcbn \termone = \fillc \contevone R$, where $R$ is a redex. We take $\actone = \tau$, and we obtain the result by a similar reasoning as the previous case.
\end{varitemize}
\item Or $\typtuplone = \typun$. Then $\stateone$ is of the form $(\distrone, \typun)$, where $\distrone = \sum_{i \in \NN} p_i \cdot (\seq{\embcbn \termtwo}, \seq{\embcbn \valone})$. Moreover, it implies that $\contone = \fillc \contevone {\vartwo_1}$, and there exists $\termthree$ such that $\fillc \contevone {\varone_1} = \embcbn \termthree$. Looking at the cbn encoding, we see that the only way to have $\fillc \contevone {\varone_1} \in \embcbn \LOPLUS$ is $\contevone = \fillc \contevtwo{\hole \bang{\embcbn \termfour}}$. It means that $\contone = \fillc \contevtwo{\vartwo_1 \bang {\embcbn \termfour} }$. We take $\actone = \actappltupl \judgone 1$ with $\judgone = (\{1\}, \{ \}, \seq \typecbn, \seq \typecbn, \bang {\embcbn \termfour}, \typecbn)$. We can easily check that $(\typun, \actone ) \in \labelsexpcbn \bigcup (\typtuples \times \{\tau \})$. 
 Moreover, we see that there exists $\distrtwo$ and $\statetwo$ a $\typun$-state such that $\osctl{(\contone, \stateone)}{\actone}{\distrtwo}$,
and $\distrtwo = \dirac{(\fillc \contevtwo {\vartwo_1}, \statetwo)}$. Using the third point of Lemma \ref{embcbnprops}, and that $\fillc \contevtwo{\varone_1 \bang {\embcbn \termfour}} \in \embcbn \LOPLUS$, we see that $\fillc \contevtwo {\varone_1} \in \embcbn \LOPLUS$. Looking at the definition of $\contsonecbn$, we see that it means that $(\fillc \contevtwo{\vartwo_1}, \typun, \typecbn) \in \contsonecbn$, and so we have the result.
\end{varitemize}
\end{proof}

\begin{theorem}[Full Abstraction for CBN]\label{th:facbn}
${\mctxcbn}$ and ${\metrcbn}$ coincide.
\end{theorem}
\begin{proof}
  We first show that $\metrcbn$ is at least as discriminating
  $\mctxcbn$. Let be $\termone, \termtwo \in \LOPLUS$.  By definition
  of $\ltsf{\flmccbn}$, we know that $\statetermcbn \termone,
  \statetermcbn \termtwo \in \ltsfst{\ltsf{\flmccbn}}$.  Moreover, we
  know by Lemma \ref{preservcbn} that $\contsonecbn$ is
  preserved by $\flmccbn$. So we can apply Proposition \ref{frtrs},
  and we see that $\appl{\metrctx{\contsonecbn}}{\statetermcbn
    \termone}{\statetermcbn \termtwo} \leq
  \appl{\metrcbn}{\termone}{\termtwo}$. Then, soundness follows using
    Lemma \ref{contextcbnmetr}.
  When proving completeness part, we rely on an ``intrinsic''
  characterization of $\metrcbn$.  The details can be found
  in the next section.\qed
\end{proof}

  \subsubsection{An intrinsic trace characterization of $\metrcbn$}
  Looking at the structure of $\flmccbn$, we see that we can in fact give an intrinsic definition of $\metrcbn$, without considering tuples and exponential constructs. In fact, every sequence of actions in $\ltsf{\flmccbn}$ can be see as an element of the set $\tracescbn$ defined by: 
$$ \traceone \in \tracescbn\bnf \emptytr \midd {\concat\termone\traceone} \text{ ,where }
  \termone \in \LOPLUS \text{ and } \fv \termone \subseteq  \{\varone
    \}.$$
  We can now talk about the probability for a term $\termone$ to effectuate a trace $\traceone$ defined by:
   $$\probtrcbn \termone {\termtwo_1 \ldots \termtwo_m \cdot \emptytr } =  \sumdistr{ \sem{\termone
      \left(\subst{\termtwo_1} \varone \termone\right)\ldots
      \left(\subst{\termtwo_m} \varone \termone\right) }}. $$
and we obtain a simple characterization of $\metrcbn$:
\begin{proposition}
Let be $\termone, \termtwo \in \LOPLUS$. Then:
$$\appl {\metrcbn} \termone \termtwo = \sup  \left\{ \mid{\probtrcbn \termone
  \traceone -  \probtrcbn \termtwo \traceone}\mid \text{ with }  \traceone
\in \tracescbn \right\}.$$
\end{proposition}
Please observe that we can alternatively see traces as contexts. As a consequence, if we define CIU-contexts in $\LOPLUS$ as $\contevone = \left(\abstr \varone \varone \termone_1 \ldots \termone_n \right)[]$. We may express $\metrcbn$ in a purely contextual way:
\begin{proposition}\label{prop:ciu}
Let be $\termone, \termtwo \in \LOPLUS$. Then
$\appl \metrcbn \termone \termtwo = \sup_{\contevone \text{ CIU context} }\lvert \sumdistr {\sem{\fillc\contevone \termone}} - \sumdistr {\sem{\fillc \contevone \termtwo}} \rvert $.
\end{proposition}
Please observe that Proposition \ref{prop:ciu} allows us to see easily that $\appl{\metrcbn}{\termone}{\termtwo} \leq \appl{\mctxcbn}{\termone}{\termtwo}$. As such, it allows us to do the completeness part of the proof Theorem \ref{th:facbn}.

\subsection{Call-by-Value}
In a similar way, we can endow $\LOPLUS$ with a call-by-value
semantics, and embedd it into $\LBANG$. We are then able to define a
suitable fragment of $\lmcbang$, a suitable set of tuple
contexts preserving it, and a caracterisation of a call-by-value context
distance for $\LOPLUS$ follows. While the construction of the stable
fragment (and the set of tuple contexts to consider) are more
involved than in the call-by-name case, we noticed that the caracterisation we
obtain seem to have some similarities with the way environmental
bisimulation for a CBV probabilistic $\lambda$-calculus was defined
in~\cite{SangiorgiV16}.
 In this section, we endow $\LOPLUS$ with a call-by-value
  semantics, as specified in Figure \ref{fig:cbvsem}.
\begin{figure*}
\begin{center}
\fbox{\scriptsize
\begin{minipage}{.95\textwidth}
\begin{center}
\scalebox{\condscale}{
$\valone = \abstr \varone
\termone
\qquad
\contevone \bnf \hole \midd \contevone \valone \midd \termone\contevone. \qquad$
}
\scalebox{\condscale}{
 \AxiomC{$\strut$}  \UnaryInfC{$\redonestepr
    {\psum\termone \termtwo }{ \termone, \termtwo} $} \DisplayProof 
\qquad 
\AxiomC{$\strut $}
 \UnaryInfC{$\redonestepr {(\abstr \varone
      \termone )\valone }{\subst\termone \varone \valone} $}
  \DisplayProof
\qquad
\AxiomC{\strut $\redonestepr \termone {\termtwo_1, \ldots, \termtwo_n}$}
  \UnaryInfC{$\onestepr {\fillc \contevone \termone }{{\fillc \contevone {\termtwo_1}}, \ldots {\fillc \contevone {\termtwo_n}}} $}
  \DisplayProof 
 }
\end{center}
\end{minipage}}
\end{center}
\caption{One-step CBV Semantics}\label{fig:cbvsem}
\end{figure*}
We denote $\mctxcbv$ the contextual metric on $\LOPLUS$ induced by this semantics. We define a new embedding: $\LOPLUS \rightarrow \LBANG$, which respects the CBV semantics, as:
\begin{align*}
\embcbv{\varone} &= \bang \varone &
\embcbv{\abstr \varone \termone}
&= \bang{\abstrexp \varone {\embcbv \termone}} \\ 
\embcbv{\termone\termtwo} &= (\abstrexp\varone \varone\embcbv{\termtwo}){\embcbv
  \termone} & 
\embcbv{\psum\termone \termtwo} &= \psum
       {\embcbv\termone}{\embcbv \termtwo}
\end{align*}
We define $\typecbv$ as the following type : $\typecbv =
\rectype \vartypone {\bang{(\tarr \vartypone \vartypone)}}$. We define $\typtwocbv = \tarr \typecbv \typecbv$: it verifies $\typecbv \equaltypes \bang \typtwocbv $. For any
term $\termone \in \LOPLUS$, if $\fv \termone = \{\varthree_1,
\ldots,\varthree_n\}$ it holds that $\wfjt{(\bang \varthree_i :
  \typecbv)_{1 \leq i\leq n}}{\embcbv \termone}{\typecbv}$.

We will use later the following nice property of the encoding:
\begin{lemma}\label{propembcbv1}
Let be $\termone$ and $\termtwo$ in $\LOPLUS$, and $\varthree$ a free variable of $\termone$. Then $\subst{\embcbv \termone}{\varthree}{\abstrexp \varthree {\embcbv \termtwo}} \in \embcbv \LOPLUS$.
\end{lemma}
\begin{proof}
The proof is by induction on the form of $\termone$.
\end{proof}
\begin{lemma}
Let be $\termone \in \LOPLUS$. Then for any $\valone \in \supp ({{\sem {\embcbv \termone}}})$, it holds that $\valone \in \embcbv \LOPLUS$.
\end{lemma}
\begin{proof}
The proof uses Lemma \ref{propembcbv1}.
\end{proof}

Again, we are now going to define a stable fragment of $\lmcbang$.
However, contrary to the call-by-name case, it does not contain the full call-by-value
encoding $\LOPLUS$. It actually contains only the encoding of \emph{values}
in $\LOPLUS$, which allows us to have a more tractable
characterization. We are able to treat the general case simply by
noting that for any term $\termone$ and $\termtwo$, it holds that
$\appl \mctxcbv \termone \termtwo = \appl \mctxcbv {\abstr \varthree
  \termone}{\abstr \varthree \termtwo}$.

\begin{definition}
We define a fragment $\flmccbv = (\statesexpcbv, \labelscbv)$ as specified in Figure \ref{fig:lmcexpcbv}, and a metric $\metrcbv$ on $\LOPLUS$ specified by:
$$\appl\metrcbv \termone \termtwo =  \appl
{\metrtr{\ltsf \flmccbv}}{\stateterm \typecbv {\bang{\embcbv{\abstr \varone \termone}}}} {\stateterm \typecbv {\bang{\embcbv{\abstr \varone \termtwo}}}} .$$
\end{definition}
\begin{figure*}
\begin{center}
\fbox{
\begin{minipage}{0.95 \textwidth}

\begin{center}
\scalebox{\condscale}{
$ {\typtuplone}_{n,(m,p)}^{\interpone} = (\seq{{(\typtwocbv)}^n}, {\seq{{(\typtwocbv)}^m,{(\bang \typtwocbv)}^p })}_\interpone  
\qquad
\vjcbv = \vvalidjudgss \cap \{(\isone, \istwo, \ms{\typone}, \ms{\typtwo}, \termone, \typecbv )\, , \, \termone \in \embcbv{\LOPLUS} \} $}
\\
\hide{
\scalebox{\condscale}{
$ \statesexpcbv = \left\{(({\seq{\abstrexp \varone \ms{\termone}}}, \seq {\abstrexp \varone \ms{\termtwo}, \bang{(\abstrexp \varone \ms{\termthree})}}_\interpone),\typtuplone_{n,(m,p)}^\interpone); \, \ms \termone, \ms \termtwo, \ms \termthree \, n,m,p\text{-sequences in} \embcbv \LOPLUS, \, \forall i, \termtwo_i \in \ms{\termone}  \right\}  \bigcap \statesexp  $}
\\}
{
\scalebox{\condscale}{
$ \statesexpcbv = \left\{
\begin{array}{l}
(({\seq{\abstrexp \varone \ms{\termone}}}, \seq {\abstrexp \varone \ms{\termtwo}, \bang{(\abstrexp \varone \ms{\termthree})}}_\interpone),\typtuplone_{n,(m,p)}^\interpone) \\ \text{ where }\ms \termone, \ms \termtwo, \ms \termthree \text{ are }n,m,p\text{-sequences in} \embcbv \LOPLUS \text{, and }\forall i, \termtwo_i \in \ms{\termone}
\end{array}
  \right\}  \bigcap \statesexp 
$}\\
 }
\scalebox{\condscale}{
$\labelscbv = \typtuples \times (\left(\{ \actappltupl \judgone i \text{ with } \judgone \in \vjcbv \} \cap \labelsappl\right) \cup \labelsexp \cup \labelsbang)   $}

\end{center}

\end{minipage}
}
\end{center}
\caption{The Stable Fragment $\flmccbv = (\statesexpcbv,\labelscbv )$.} \label{fig:lmcexpcbv}
\end{figure*}
Please observe that contrary to the cbn-case, we do not need to
restrict action by considering the type of states. 
\begin{proposition}
$\flmccbv$ is a stable fragment of $\lmcbang$.
\end{proposition}
\begin{proof}
Let be $(\tuplone, \typtuplone) \in \statesexpcbv$. Since $(\tuplone, \typtuplone)$ is in $\statesexpcbv$, there exist $\interpone$, and $ \ms \termone, \ms \termtwo, \ms \termthree$ respectively $n,m,p$ term sequences, such that $\tuplone = ({{\abstrexp \varone \ms{\termone}}}, {\abstrexp \varone \ms{\termtwo}; \bang{(\abstrexp \varone \ms{\termthree})}}_\interpone)$ and $\typtuplone = \typtuplone_{n,(m,p)}^\interpone$. Let be $\actone $ such that $(\typtuplone, \actone) \in \labelscbv$. Let be $\statetwo = (\tupltwo, \typtupltwo)$ such that  $\probexp(\stateone, \actone, \statetwo)>0$. We want to show that $\statetwo \in  \statesexpcbv$. We distinguish cases depending on $\actone$. 
\begin{itemize}
\item Or $\actone = \actbang i$. It implies that $\tupltwo = (\abstrexp \varone {\ms \termone}, {\abstrexp \varone {\ms \termtwo}; \bang {(\abstrexp \varone \ms{\termthree})}}_\interpone ; \valone)$ with $\wfjt {} \valone \typecbv $ and $ \typtupltwo = (\seq{{(\typtwocbv)}^n}, {\seq{{(\typtwocbv)}^m,{(\bang \typtwocbv)}^p }}_\interpone, \typecbv)  $. Since $\valone$ is obtained by the evaluation of a term in $\embcbv \LOPLUS$, it is also in $\embcbv \LOPLUS$. It allows us to see that $\valone$ is of the form $\bang {\abstrexp \varone \termfour}$. So we see that we can construct a permutation $\interptwo $ from $\{1, \ldots,m+p+1 \}$ such that
$\statetwo = ((\abstrexp \varone {\ms \termone}, {\abstrexp \varone {\ms \termtwo}; \bang {(\abstrexp \varone (\ms{\termthree};\termfour))}}_\interptwo) )$, and $ \typtupltwo = (\seq{{(\typtwocbv)}^n}, {\seq{{(\typtwocbv)}^m,{(\bang \typtwocbv)}^{p+1} }}_\interptwo)  $. As a consequence, $\statetwo \in \statesexpcbv$.
\item Or $\actone = \evalbang i$. It implies that $\tupltwo = (\abstrexp \varone {(\ms \termone ; \termthree_j)}, {\abstrexp \varone {\ms \termtwo}; \bang {(\abstrexp \varone \ms{\termthree}_{\{1, \ldots,p \} \setminus \{j\}})}}_{\interpone_\remove{\{j\}}})$, when $\interpone(m + j) = i$, and $\typtupltwo = (\seq{{(\typtwocbv)}^{n+1}}, {\seq{{(\typtwocbv)}^m,{(\bang \typtwocbv)}^{p-1} }}_{\interpone_\remove{\{j\}}} )  $. So we can see that $\statetwo \in \statesexpcbv$. 
\item Or $\actone = \actappltupl \judgone i$ with $\judgone \in \vjcbv$. Then $\judgone = (\isone, \istwo, \ms{\typone}, \ms{\typtwo}, \termone, \typecbv )$ and  $\termone$ is an open value in $ \embcbv{\LOPLUS}  $. Since there exist $\statetwo$ such that $\probexp(\stateone, \actone, \statetwo)>0$, it means that $\isone \subseteq \{1, \ldots, n\}$, and $\istwo \subseteq \{1, \ldots m+p\}$, and moreover $(\ms \typone, \ms \typtwo)$ is one of the $\typtuplone^{m,n,p}_\interpone$. verifying $1 \leq \interpone(i) \leq m$ (that is, $\ms \typtwo_i = \tarr \typtwocbv \typtwocbv $). By a similar reasoning as previous cases, we can see that $\statetwo \in \statesexpcbv$.
\end{itemize} 
\end{proof}
Please observe that for any term $\termone \in
\LOPLUS$, the associated state $\stateterm \typecbv {\embcbv {\abstr
    \varthree \termone}}$ is in $\statesexpcbv$.

As in the previous section,
we define a suitable set $\contsonecbv$ of tuple contexts preserving $\flmccbv$, in order to apply Proposition \ref{frtrs}.
\begin{definition}
We define a set of tuple contexts $\contsonecbv$ as those of the form $( \contone_1 \ldots \contone_k \conttwo, (\ms \typone, \ms \typtwo), \typecbv)$ such that:
\begin{varitemize}
\item $ \contone_i \in \{ \vartwo_j \mid \typtwo_j = \tarr {\typecbv}{\typecbv} \} \cup \{\abstrexp \varthree {\varthree\embcbv{\termone}} \mid \varthree \not \in \fv \termone\}\cup \abstrexp \varthree {\embcbv{\LOPLUS}}  \cup\{\varone_j \}$
\item $\conttwo \in \embcbv{\LOPLUS} \cup  \{\vartwo_j \mid \typtwo_j = \bang{(\tarr \typecbv \typecbv)} \}\cup \{\varone_j \} $
\end{varitemize}
\end{definition}
\begin{proposition}
$\contsonecbv$ is preserved by $\flmccbv$.
\end{proposition}
\begin{proof}
Let be $(\contone, \typtuplone, \typthree) \in \contsonecbv$. Let be $\actone$ such that $(\typtuplone, \actone) \in \labelscbv$. Let be $\stateone$ a $\typtuplone$-state $\stateone$ in $\ltsfst{\ltsf{\flmccbv}}$.Let be $\distrtwo$ such that $\osctl{(\contone, \stateone)}{\actone}{\distrtwo}$. Let be $\stateonetd \in \supp (\distrtwo)$. We note $\stateonetd = (\contthree, \statetwo)$, and $\typtupltwo$ the tuple type such that $\statetwo$ is a $\typtupltwo$-state. We have to show that $ (\contthree, \typtupltwo,\typecbv) \in \contsonecbv$.

 Please observe that there exists actually only one $\actone$ such that  $\osctl{(\contone, \stateone)}{\actone}{}$, and it depends only on $\contone$. Consequently, we do the proof by considering separately the different form of $\contone$. We know, by definition of $\contsonecbv$, that  $\contone$ is of the form $ \contone_1 \ldots \contone_k \conttwo$.
\begin{itemize}
\item If one of the $\contone_i$ is of the form $\varone_j$, we consider the smallest index $i$ for which it happens. Then $\actone = \actbang i$, and $\contthree = \contone_1, \ldots, \bang{\vartwo_{n+1}}, \ldots \contone_n\conttwo$.
\item Otherwise, $\contone_1 \ldots \contone_k \hole$ is an evaluation context. Then:
\begin{itemize}
\item If $\conttwo = \varone_i$. We have that $\actone = \actbang i $. It follows that $\osctl {\stateone}{\actone}{\statetwo}$, and $\contthree = \contone_1, \ldots, \contone_k \vartwo_{m+p+1} $. So we can see that the result holds.
\item If $\conttwo =  \vartwo_j$, with $\typtwo_j = \bang \typtwocbv$. Then it holds that $\actone = \evalbang i$., and $\contthree = \contone_1, \ldots, \contone_k \bang{\varone_{n+1}}$.
\item If $\conttwo$ is in $\embcbv \LOPLUS$. We show the result by induction on the term $\termthree$ such that $\conttwo = \embcbv{\termthree}$.
\begin{itemize}
\item Or $\conttwo = \embcbv{\psum \termone \termtwo}$. Then $\actone = \tau$, and $\contthree =  \contone_1, \ldots, \contone_k \contfour$, with $\contfour \in \{\embcbv \termone, \embcbv \termtwo \} \subseteq \embcbv \LOPLUS$  and so the result holds.
\item Or $\conttwo = \embcbv{\termone \termtwo}$. Then $\conttwo = (\abstrexp \varthree{\varthree {\embcbv \termtwo}}){\embcbv \termone}$. We can take $\contfour = \embcbv{\termone}$, $\contone_{n+1} = (\abstrexp \varthree{\varthree {\embcbv \termtwo}})$, and we can see that $\contone = \contone_1, \ldots , \contone_{n+1} \contfour$, and so we can apply the induction hypothesis to $\contfour = \embcbv{\termone}$.
\item If $\conttwo$ is an open value. It means that $\conttwo = \bang{\abstrexp \varthree {\embcbv \termone}}$, or $\conttwo = \bang{\varone_i}$. Please observe that $\contone_n$ is always an open value (by construction). It means that the redex $\contone_n \conttwo$ is going to be reduced.
\begin{itemize}
\item If $\contone_n = \vartwo_j$, with $\typtwo_j = \tarr \typecbv \typecbv$, then $\actone = \actappltupl \judgone j$, where $\judgone = (\isone, \istwo, \ms \typone, \ms \typtwo, \conttwo, \typecbv)$, where $\isone \subseteq \{1, \ldots,n\}$, and $\wfjt{\bang {{\ms \varone}_\isone} : \bang {\ms \typone_\isone}, {\ms \vartwo}_\istwo : \ms \typtwo_\istwo }{\conttwo}{\typecbv} $ and we see that indeed $\judgone \in \vjcbv $. Moreover, if we note $q$ the cardinality of $\isone$, we can see that $\contthree = \contone_1, \ldots, \contone_{n-1} \vartwo_{n+1-{q+1}}$. As a consequence, we have the result.
\item If $\contone_n = \abstrexp \varthree{(\varthree\embcbv{\termtwo})}$. Then $\actone = \tau$. Please observe that we can define $\contone_{n+1}$ such that $\bang \contone_{n+1} = \conttwo$. Now, please observe that we should have $\contthree = \contone_1, \ldots, \contone_{n-1} \contone_{n+1} \embcbv{\termtwo} $, and so the result holds.
\item If $\contone_n =  \abstrexp \varthree{(\embcbv{\termtwo})}$. Then if $\conttwo = \bang{\varone_i}$, we have that $\contthree = \contone_1, \ldots, \contone_{n-1} \embcbv{\subst \termone \varthree {\varone_i}} $, and we have the result. Similarly, if $\conttwo =  \bang{\abstrexp \varthree {\embcbv \termone}}$, it holds that $\contthree = \contone_1, \ldots, \contone_{n-1} \embcbv{\subst \termone \varthree {\abstrexp \varthree{\embcbv \termone}}} $. We use Lemma \ref{propembcbv1} to obtain the result. 
\end{itemize}
\end{itemize}
\end{itemize}
\end{itemize}
\end{proof}

\begin{definition}
We define a metric $\metrcbv$ on $\LOPLUS$ by:
$\appl\metrcbv \termone \termtwo =  \appl
{\metrtr{\ltsf \flmccbv}}{\stateterm \typecbv {\bang{\embcbv{\abstr \varone \termone}}}} {\stateterm \typecbv {\bang{\embcbv{\abstr \varone \termtwo}}}}. $
\end{definition}

\begin{theorem}[Full Abstraction for CBV]
 $ \mctxcbv$ and  $\mtcbv$ coincide.
\end{theorem}

The proof can be found in~\cite{longversion}. Again, 
soundness is based on Proposition \ref{frtrs}, while completeness
relies on yet another characterization of the context distance.

\subsubsection{Intrinsic characterization}

We don't have an intrinsic characterization as simpler as in the cbn case. However, we can express the distance $\mctxcbv$ using a LMC
$\lmccbv$, designed to be a simplified and untyped version of $\lmcbang$, and that doesn't use the embedding.
\begin{definition}
We define a labelled Markov chain $\lmccbv = (\stcbv, \actcbv, \probcbv)$, where:
\begin{varitemize}
\item $\stcbv = \{(\ms \valone, n) \mid \ms \valone \text{ a }n\text{-sequence of values}\} $
\item $\actcbv = \{\actappltupl i \termone \mid \exists k \in \NN, \, \termone = x_k \text{ or } \termthree = \abstr \varthree \termone \} $
\item $\probcbv$ is defined in Figure \ref{fig:lmccbv}.
\end{varitemize}
\end{definition}
As for $\LBANG$, we transform the LMC $\lmccbv$ into a purely non-deterministic WLTS. We define 
$\ltscbv =
(\ltsstatescbv, \ltslabelscbv, \ltstrans,w)$, where:
\begin{varitemize}
\item
  The set of states is
  $\ltsstatescbv=\bigcup_{n \in \NN}\left( \distrs{\{\ms \valone \mid \ms \valone \text{ a }n\text{-sequence }\}}\times \{ n\}\right)$
\item
  The set of actions is $\ltslabelscbv = \actcbv\cup \NN$,
\item
  The transition function 
  $\ltsact{}{\cdot}{}$ is defined as follows:
  \begin{align*}
    &\ltsact {(\distrone, n)}{n}{(\distrone,
      n)}\\
    &\ltsact{(\distrone, n)}{\actone}{\sum_{\ms \valone}
    \distrone(\ms \valone)\cdot \probexp((\ms \valone, n))(\actone)} ,
  \end{align*}
\item
  The weight $w$ is such that $w(\distrone, n) = \sumdistr
  \distrone$.
\end{varitemize}
Please remember that we give in Section \ref{subsect:tr} a definition of the trace metric $\metrtr \ltsone$ for any WLTS. We can apply it to $\ltscbv$, and it gives us a characterization of $\metrcbv$. 
\begin{proposition}
Let be $\termone, \termtwo \in \LOPLUS$. Then:
$$\appl{\metrtr \ltsone}{(\seq{\abstr \varone \termone},1)}{(\seq{\abstr \varone \termtwo},1)} = \appl{\metrcbv}{\termone}{\termtwo}. $$ 
\end{proposition}


\begin{figure*}
\begin{center}
\fbox{
\begin{minipage}{0.95 \textwidth}
\scriptsize{
\begin{center}
\begin{tikzpicture}[auto]

\node [draw, rectangle, rounded corners] (Q) at (-0.4,-3) {$
\begin{array}{c}
(\abstr \varthree {\ms \termone},n) \\
\ms \termone \text{ a }n-\text{sequence}
\end{array}
$};
\node [draw, rectangle, rounded corners] (R) at (6.6,-3) {$
(\abstr \varthree {\ms \termone};\valtwo, n+1) 
$};
\draw[->](Q) to node[above]{$
\begin{array}{c}
\actappltupl \valone i\\
\fv \valone \subseteq \{\varone_1, \ldots, \varone_n \}
\end{array}
$ } node[below]{$ \sem{\subst {\termone_i}{\varthree}{\subst \valone {\ms \varone}{\abstr \varthree{\ms \termone}}}}(\valtwo)$} (R);

\end{tikzpicture}
\end{center}}
\end{minipage}}
\end{center}
\caption{Definition of $\lmccbv$}\label{fig:lmccbv}
\end{figure*}


\section{Related Work}
\newcommand{\PCF}{\ensuremath{\mathsf{PCF}}} This is definitely \emph{not}
the first work on metrics in the context of programming languages
semantics. A very nice introduction to the topic, together with a
comprehensive (although outdated) list of references can be found
in~\cite{vanBreugel2001}. One of the many uses of metrics is as an
alternative to order-theoretic semantics. This has also been applied
to higher-order languages, and to \emph{deterministic}
\PCF~\cite{Escardo1999}.

If one focuses on probabilistic programming languages, the first
attempts at using metrics as a way to measure ``how far'' two programs
are, algebraically or behaviourally, are due to Giacalone et
al.~\cite{Giacalone1990}, and Desharnais et
al.~\cite{DesharnaisCONCUR99,DesharnaisLICS02}, who both consider
process algebras in the style of Milner's CCS. Most of further work in
this direction has focused on concurrent specifications. Among the
recent advances in this direction (and without any hope of being
comprehensive), we can cite Gebler et al.'s work on uniform continuity
as a way to enforce compositionality in metric
reasoning~\cite{Gebler2015CONCUR,Gebler2015FoSSACS}.  

Finally, great inspiration for this work came from the many contributions on
metrics for labelled Markov chains and processes appeared in the last
twenty years (e.g. \cite{Worrell,DLT2008}). 
\section{Conclusions}
We have shown \emph{how} the context distance can be characterized so
as to simplify concrete proofs, and \emph{to which extent} this metric
trivializes.  All this has been done in a universal linear
$\lambda$-calculus for probabilistic computation. This clarifies to
which extent refining equivalences into metrics is worth in such a
scenario. The tuple-based techniques in Section \ref{subsect:Upto} are
potentially very interesting in view of possible applications to
cryptography, as hinted in \cite{DalLagoCappai}. This is indeed what
we are working on currently.

\bibliographystyle{abbrv}
\bibliography{biblio}

\begin{thebibliography}{10}

\bibitem{Barendregt84}
Hendrik~Pieter Barendregt.
\newblock {\em {The Lambda Calculus -- Its Syntax and Semantics}}, volume 103
  of {\em Studies in Logic and the Foundations of Mathematics}.
\newblock North-Holland, 1984.

\bibitem{positivetype}
Hendrik~Pieter Barendregt, Wil Dekkers, and Richard Statman.
\newblock {\em Lambda Calculus with Types}.
\newblock Perspectives in logic. Cambridge University Press, 2013.

\bibitem{BizjakBirkedal}
Ales Bizjak and Lars Birkedal.
\newblock Step-indexed logical relations for probability.
\newblock In {\em Proc. of FoSSaCS}, pages 279--294, 2015.

\bibitem{DalLagoCappai}
Alberto Cappai and Ugo Dal~Lago.
\newblock On equivalences, metrics, and polynomial time.
\newblock In {\em Proc. of FCT}, pages 311--323, 2015.

\bibitem{CrubilleDalLago2014ESOP}
Rapha{\"{e}}lle Crubill{\'{e}} and Ugo Dal~Lago.
\newblock On probabilistic applicative bisimulation and call-by-value
  {\(\lambda\)}-calculi.
\newblock In {\em Proc. of ESOP}, pages 209--228, 2014.

\bibitem{longversion}
Rapha{\"{e}}lle Crubill{\'{e}} and Ugo Dal~Lago.
\newblock Metric reasoning about $\lambda$-terms: The general case (long
  version).
\newblock Available at \url{http://eternal.cs.unibo.it/mrltgc.pdf}, 2016.

\bibitem{CrubilleDalLagoLICS2015}
Rapha{\"{e}}lle Crubill{\'{e}} and Ugo~Dal Lago.
\newblock Metric reasoning about {\(\lambda\)}-terms: The affine case.
\newblock In {\em Proc. of LICS}, pages 633--644, 2015.

\bibitem{VignudelliCrubilleDalLagoSangiorgi}
Rapha{\"{e}}lle Crubill{\'{e}}, Ugo~Dal Lago, Davide Sangiorgi, and Valeria
  Vignudelli.
\newblock On applicative similarity, sequentiality, and full abstraction.
\newblock In {\em Proc. of Correct System Design - Symposium in Honor of
  Ernst-R{\"{u}}diger Olderog on the Occasion of His 60th Birthday}, pages
  65--82, 2015.

\bibitem{DalLagoSangiorgiAlberti2014POPL}
Ugo Dal~Lago, Davide Sangiorgi, and Michele Alberti.
\newblock On coinductive equivalences for higher-order probabilistic functional
  programs.
\newblock In {\em Proc. of POPL}, pages 297--308, 2014.

\bibitem{DalLagoZorzi}
Ugo Dal~Lago and Margherita Zorzi.
\newblock Probabilistic operational semantics for the lambda calculus.
\newblock {\em {RAIRO} - Theor. Inf. and Applic.}, 46(3):413--450, 2012.

\bibitem{DesharnaisCONCUR99}
Josee Desharnais, Vineet Gupta, Radha Jagadeesan, and Prakash Panangaden.
\newblock Metrics for labeled markov systems.
\newblock In {\em Proc. of CONCUR}, 1999.

\bibitem{DesharnaisLICS02}
Josee Desharnais, Radha Jagadeesan, Vineet Gupta, and Prakash Panangaden.
\newblock The metric analogue of weak bisimulation for probabilistic processes.
\newblock In {\em Proc. of LICS}, pages 413--422, 2002.

\bibitem{DLT2008}
Jos{\'{e}}e Desharnais, Fran{\c{c}}ois Laviolette, and Mathieu Tracol.
\newblock Approximate analysis of probabilistic processes: Logic, simulation
  and games.
\newblock In {\em Proc. of QEST}, pages 264--273, 2008.

\bibitem{EhrhardTassonPagani2014POPL}
Thomas Ehrhard, Christine Tasson, and Michele Pagani.
\newblock Probabilistic coherence spaces are fully abstract for probabilistic
  {PCF}.
\newblock In {\em Proc. of POPL}, pages 309--320, 2014.

\bibitem{Escardo1999}
Martin Escardo.
\newblock A metric model of {PCF}.
\newblock Proceedings of the Workshop on Realizability Semantics and
  Applications. Available at
  \url{http://www.cs.bham.ac.uk/~mhe/papers/metricpcf.pdf}, 1999.

\bibitem{Gebler2015FoSSACS}
Daniel Gebler, Kim~Guldstrand Larsen, and Simone Tini.
\newblock Compositional metric reasoning with probabilistic process calculi.
\newblock In {\em Proc. of {FoSSaCS}}, pages 230--245, 2015.

\bibitem{Gebler2015CONCUR}
Daniel Gebler and Simone Tini.
\newblock {SOS} specifications of probabilistic systems by uniformly continuous
  operators.
\newblock In {\em Proc. of {CONCUR}}, pages 155--168, 2015.

\bibitem{Giacalone1990}
Alessandro Giacalone, Chi chang Jou, and Scott~A. Smolka.
\newblock Algebraic reasoning for probabilistic concurrent systems.
\newblock In {\em Proc. IFIP TC2}, pages 443--458. North-Holland, 1990.

\bibitem{Girard87}
Jean{-}Yves Girard.
\newblock Linear logic.
\newblock {\em Theor. Comput. Sci.}, 50:1--102, 1987.

\bibitem{GoldwasserMicali}
Shafi Goldwasser and Silvio Micali.
\newblock Probabilistic encryption.
\newblock {\em J. Comput. Syst. Sci.}, 28(2):270--299, 1984.

\bibitem{Church}
Noah~D. Goodman, Vikash~K. Mansinghka, Daniel~M. Roy, Keith Bonawitz, and
  Joshua~B. Tenenbaum.
\newblock Church: a language for generative models.
\newblock In {\em {UAI} 2008}, pages 220--229, 2008.

\bibitem{JonesPlotkin}
C.~Jones and Gordon~D. Plotkin.
\newblock A probabilistic powerdomain of evaluations.
\newblock In {\em Prof. of LICS}, pages 186--195, 1989.

\bibitem{JungTix1998}
Achim Jung and Regina Tix.
\newblock The troublesome probabilistic powerdomain.
\newblock {\em Electr. Notes Theor. Comput. Sci.}, 13:70--91, 1998.

\bibitem{DalLagoSangiorgiAlberti2014}
Ugo~Dal Lago, Davide Sangiorgi, and Michele Alberti.
\newblock On coinductive equivalences for higher-order probabilistic functional
  programs.
\newblock In {\em Proc. of POPL}, pages 297--308, 2014.

\bibitem{manning1999foundations}
Christopher~D Manning and Hinrich Sch{\"u}tze.
\newblock {\em Foundations of statistical natural language processing}, volume
  999.
\newblock MIT Press, 1999.

\bibitem{MardareDoctoral}
Radu Mardare.
\newblock Logical foundations of metric behavioural theory for markov
  processes.
\newblock Doctoral Thesis. In Preparation, 2016.

\bibitem{Pfenning}
Sungwoo Park, Frank Pfenning, and Sebastian Thrun.
\newblock A probabilistic language based on sampling functions.
\newblock {\em ACM Trans. Program. Lang. Syst.}, 31(1), 2008.

\bibitem{pearl1988probabilistic}
Judea Pearl.
\newblock {\em Probabilistic reasoning in intelligent systems: networks of
  plausible inference}.
\newblock Morgan Kaufmann, 1988.

\bibitem{Plotkin77}
Gordon~D. Plotkin.
\newblock {LCF} considered as a programming language.
\newblock {\em Theor. Comput. Sci.}, 5(3):223--255, 1977.

\bibitem{Ramsey}
Norman Ramsey and Avi Pfeffer.
\newblock Stochastic lambda calculus and monads of probability distributions.
\newblock In {\em Prof. of POPL}, pages 154--165, 2002.

\bibitem{SahebDjahromi}
N.~Saheb-Djahromi.
\newblock Probabilistic {LCF}.
\newblock In {\em Proc. of MFCS}, pages 442--451, 1978.

\bibitem{Sangiorgi98}
Davide Sangiorgi.
\newblock On the bisimulation proof method.
\newblock {\em Mathematical Structures in Computer Science}, 8:447--479, 1998.

\bibitem{SangiorgiV16}
Davide Sangiorgi and Valeria Vignudelli.
\newblock Environmental bisimulations for probabilistic higher-order languages.
\newblock In {\em Proceedings of the 43rd Annual {ACM} {SIGPLAN-SIGACT}
  Symposium on Principles of Programming Languages, {POPL} 2016, St.
  Petersburg, FL, USA, January 20 - 22, 2016}, pages 595--607, 2016.

\bibitem{Simpson2005}
Alex~K. Simpson.
\newblock Reduction in a linear lambda-calculus with applications to
  operational semantics.
\newblock In {\em Proc. of RTA}, pages 219--234, 2005.

\bibitem{thrun2002robotic}
Sebastian Thrun.
\newblock Robotic mapping: A survey.
\newblock {\em Exploring artificial intelligence in the new millennium}, pages
  1--35, 2002.

\bibitem{vanBreugel2001}
Franck van Breugel.
\newblock An introduction to metric semantics: operational and denotational
  models for programming and specification languages.
\newblock {\em Theor. Comput. Sci.}, 258(1-2):1--98, 2001.

\bibitem{Worrell}
Franck van Breugel and James Worrell.
\newblock A behavioural pseudometric for probabilistic transition systems.
\newblock {\em Theor. Comput. Sci.}, 331(1):115--142, 2005.

\end{thebibliography}
\end{document}